\documentclass[11pt,reqno]{amsart}
\usepackage{amssymb}
\usepackage{amsfonts}
\usepackage{amsmath}
\usepackage{mathrsfs}
\usepackage{stmaryrd}
\usepackage{physics}
\usepackage{braket}
\usepackage{dsfont}
\usepackage{bbold}
\usepackage{graphicx}
\usepackage{relsize}
\usepackage{makecell}
\usepackage[table,xcdraw]{xcolor}
\usepackage{enumerate}
\usepackage[pagebackref, colorlinks = true, linkcolor = blue, urlcolor  = blue, citecolor = red]{hyperref}
\usepackage[margin=1in]{geometry}
\usepackage{enumitem}
\usepackage{mathtools}
\DeclarePairedDelimiter{\ceil}{\lceil}{\rceil}
\DeclarePairedDelimiter\floor{\lfloor}{\rfloor}

\usepackage{graphbox}
\usepackage{comment}
\usepackage{float}
\usepackage[font=scriptsize]{caption}

\newcommand{\Hil}{\mathcal{H}}
\newcommand{\Kil}{\mathcal{K}}

\newcommand{\cH}{\mathcal{H}}

\newcommand{\cC}{\mathcal{C}}
\newcommand{\cD}{\mathcal{D}}
\newcommand{\cL}{\mathcal{L}}

\newcommand{\iden}{\mathbb{1}}
\newcommand{\eps}{\varepsilon}
\renewcommand{\epsilon}{\varepsilon}
\renewcommand{\phi}{\varphi}
\newcommand{\overbar}[1]{\mkern 1.5mu\overline{\mkern-1.5mu#1\mkern-1.5mu}\mkern 1.5mu}

\newcommand{\C}[1]{\mathbb{C}^{#1}}

\newcommand{\B}[1]{\mathcal{L}({#1})}
\newcommand{\State}[1]{\mathcal{D}({#1})}

\newcommand{\id}{{\rm{id}}}

\newcommand{\supp}{{\rm{supp }}}

\newtheorem{theorem}{Theorem}[section]
\newtheorem{definition}[theorem]{Definition}
\newtheorem*{definition*}{Definition}

\newtheorem{corollary}[theorem]{Corollary}
\newtheorem{lemma}[theorem]{Lemma}
\newtheorem{remark}[theorem]{Remark}

\newtheorem{question}[theorem]{Question}

\newtheorem*{conjecture*}{Conjecture}

\newcommand\vertarrowbox[3][6ex]{%
  \begin{array}[t]{@{}c@{}} #2 \\
  \left\uparrow\vcenter{\hrule height #1}\right.\kern-\nulldelimiterspace\\
  \makebox[0pt]{\scriptsize#3}
  \end{array}%
}

\theoremstyle{definition}

\definecolor{darkgreen}{rgb}{0,0.392,0}

\author{Satvik Singh}
\email{satvik.singh@tum.de}
\address{\parbox{\linewidth}{Department of Mathematics, Technical University of Munich, \\[2pt]
Department of Applied Mathematics and Theoretical Physics, \\ University of Cambridge, Cambridge, United Kingdom }}

\author{Nilanjana Datta}
\email{n.datta@damtp.cam.ac.uk}
\address{\parbox{\linewidth}{Department of Applied Mathematics and Theoretical Physics, \\ University of Cambridge, Cambridge, United Kingdom}}

\title[Information storage and transmission under Markovian noise]{Information storage and transmission under \\ Markovian noise}

\begin{document}

\begin{abstract}
We study the information transmission capacities of quantum Markov semigroups $(\Psi^t)_{t\in \mathbb{N}}$ acting on $d-$dimensional quantum systems. We show that, in the limit of $t\to \infty$, the capacities can be efficiently computed in terms of the structure of the peripheral space of $\Psi$, are strongly additive, and satisfy the strong converse property. We also establish convergence bounds to show that the infinite-time capacities are reached after time $t\gtrsim d^2\ln (d)$. From a data storage perspective, our analysis provides tight bounds on the number of bits or qubits that can be reliably stored for long times in a quantum memory device that is experiencing Markovian noise. From a practical standpoint, we show that typically, an $n-$qubit quantum memory, with Markovian noise acting independently and identically on all qubits and a fixed time-independent global error correction mechanism, becomes useless for storage after time $t\gtrsim n2^{2n}$. In contrast, if the error correction is local, we prove that the memory becomes useless much more quickly, i.e., after time $t\gtrsim \ln(n)$. In the setting of point-to-point communication between two spatially separated parties, our analysis provides efficiently computable bounds on the optimal rate at which bits or qubits can be reliably transmitted via long Markovian communication channels $(\Psi^l)_{l\in \mathbb{N}}$ of length $l\gtrsim d^2 \ln(d)$, both in the finite block-length and asymptotic regimes. 

\end{abstract}

\maketitle
\tableofcontents

\section{Introduction}
Suppose that two parties, Alice and Bob, are spatially separated, and Alice wants to send information encoded in $d-$dimensional quantum systems (qudits) to Bob. The communication link between them is modelled by a (memoryless) noisy quantum channel $\Phi:\B{\Hil_A}\to \B{\Hil_B}$, where $\Hil_A\simeq \Hil_B\simeq\Hil\simeq \mathbb{C}^d$, and $\B{\Hil}$ denotes the algebra of linear operators acting on a Hilbert space $\Hil$. The capacity of $\Phi$ describes the best rate at which Alice can send information to Bob by using the channel many times (say $n$) such that the error incurred in transmission vanishes as $n\to \infty$ \cite{Shannon1948communication, Cover2005book,  watrous2018theory}. Depending on the type of information being sent: classical, private classical, and quantum, one obtains different capacities: $C(\Phi), P(\Phi)$, and $Q(\Phi)$, respectively. The infimum of communication rates for which the error incurred in transmission goes to $1$ in the limit $n \to \infty$ gives the strong converse capacities: $C^{\dagger}(\Phi), P^{\dagger}(\Phi)$ and $Q^{\dagger}(\Phi)$. These denote the threshold values of the rates above which information transmission fails with certainty. Clearly, $Q(\Phi)\leq Q^{\dagger}(\Phi)$, and we say that $\Phi$ satisfies the strong converse property for quantum capacity if $Q(\Phi)=Q^{\dagger}(\Phi)$, with the interpretation being that the capacity $Q(\Phi)$ provides a sharp threshold between achievable and unachievable rates of communication. The strong converse property for classical and private classical capacity are defined similarly. Determining whether the strong converse property holds for different types of capacities has been an active area of research (see \cite[Section 3]{cheng2024strong} for a historical survey). The capacities satisfy the following relation: $C(\Phi)\geq P(\Phi)\geq Q(\Phi)$, and the inequalities can be maximally strict \cite{Leung2014privacy, Leung2016privacy}. In particular, the separation between $P(\Phi)$ and $Q(\Phi)$ can be linked to the fact that distilling entanglement is fundamentally different from distilling private classical bits \cite{Horodecki2005private, Horodecki2009private}. 

Although operationally crucial, the capacities of a noisy channel are not even known to be computable~\cite{wolf2011undecidable, wolf2024decidable}, let alone efficiently computable. Furthermore, the capacities can exhibit strange superadditive behavior: there exist pairs of channels (say $\Phi$ and $\Psi$), each of which has zero quantum capacity $Q(\Phi)=Q(\Psi)=0$, but which can be used in tandem to transmit quantum information at a non-zero rate, i.e.~$Q(\Phi \otimes \Psi)>0$ \cite{Smith2008super}. More generally, for any channel $\Phi$, there may exist other channels $\Psi$ that can increase $\Phi$'s communication capacity, in the sense that $Q(\Phi\otimes \Psi)>Q(\Phi)+Q(\Psi)$ or $P(\Phi\otimes \Psi)>P(\Phi)+P(\Psi)$ \cite{Smith2009superadd, Li2009superadd, Smith2011superadd, Lim2019superadd,  Koudia2022superadd, Leditzky2023superadd}. Such exotic superadditive behavior indicates that the capacity of a noisy channel
may not adequately characterize the channel, since the utility of the channel depends on what other contextual channels are available for communication \cite{Winter2016potential}. 

If the error incurred in information transmission via a channel $\Phi:\B{\Hil_A}\to \B{\Hil_B}$ is required to always be zero, one enters the world of zero-error information theory \cite{Shannon1956zero, Korner1998zero}. The corresponding capacities for transmission of classical, private classical, and quantum information are denoted by $C_{\operatorname{zero}}(\Phi), P_{\operatorname{zero}}(\Phi)$, and $Q_{\operatorname{zero}}(\Phi)$, respectively. The constraint of error-free communication makes zero-error information theory much more algebraic/combinatorial in nature, with the central object of study being the so-called \emph{(non-commutative) confusability graph} of the noisy channel \cite{Shannon1956zero, Duan2013noncomm}. All the computability and superadditivity issues mentioned above persist (and arguably become more extreme) for zero-error communication. There is some evidence for the capacities to be uncomputable (even for classical channels) \cite{Alon2006zeroclassical, Shor2008complexity, Boche2020zeroclassical, boche2024computability}. Moreover, exotic examples of channels have been constructed that exhibit \emph{super-duper activation} of channel capacities, i.e., there exist quantum channels $\Phi$ and $\Psi$ with no zero-error classical capacity: $C_{\operatorname{zero}}(\Phi)=C_{\operatorname{zero}}(\Psi)=0$, yet they can be used in tandem to transmit far more delicate quantum information perfectly: $Q_{\operatorname{zero}}(\Phi\otimes \Psi)>0$ \cite{Chen2010zerosuper,Duan2009zerosuper, Cubitt2011zerosuper, Shirokov2015zerosuper}.

In a nutshell, we will show that all the stated obstacles that one usually faces in quantum Shannon theory disappear when one studies the capacities of sequential compositions 
\begin{equation}
        \Psi^m := \underbrace{\Psi\circ \Psi \circ \ldots \circ \Psi}_{m \text{ times}}
\end{equation}
of a channel $\Psi: \B{\Hil}\to \B{\Hil}$ in the limit $m\to \infty$.

\subsection{Our contribution}

Broadly speaking, we address the following question in this paper.

\begin{question}
    Given a channel $\Psi:\B{\Hil}\to \B{\Hil}$, how do the transmission capacities of self-compositions $\Psi^m := \Psi\circ \Psi \ldots \circ \Psi$ behave as a function of the number of compositions $m$ and $\Psi$?
\end{question}

We study this problem from two different perspectives.

\subsubsection{Data storage} 
Suppose we have a physical quantum memory device comprised of $n$ qudits that we wish to use to store data. We model the device as an \emph{open} quantum system that is \emph{weakly} interacting with the environment (or bath). In this so-called \emph{weak-coupling limit}, the decay times of correlation functions of the bath are much
shorter than the typical time scale over which the state of the system changes
significantly. In other words, the bath `forgets' about its interaction with the system and returns to its steady state quickly relative to the speed at which the system evolves. Since in subsequent interactions, the bath does not remember the details of the previous interaction, the dynamics of the system becomes \emph{Markovian} \cite{Alicki2002open, Alicki2007open, Breuer2007open}. Mathematically, the system dynamics in time can be modelled by a \emph{discrete Quantum Markov Semigroup} (dQMS) $(\Psi^t)_{t\in \mathbb{N}}$, where $\Psi :\B{\Hil}\to \B{\Hil}$ is a noisy channel, $\Hil\simeq (\C{q})^{\otimes n}$ is the Hilbert space of the memory, $\Psi^t = \Psi \circ \Psi \cdots \circ \Psi$ denotes the $t$-fold composition of $\Psi$ with itself, and $t\in\mathbb{N}$ plays the role of the time-parameter.

\begin{figure}[H]
    \centering
    \includegraphics[scale=2]{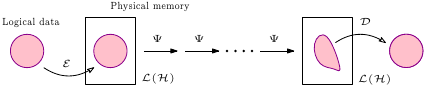}
    \caption{Schematic for a physical memory device experiencing Markovian noise in time modelled by a dQMS $(\Psi^t)_{t\in\mathbb{N}}$, where $\Psi:\B{\Hil}\to \B{\Hil}$ is a quantum channel. Logical data is encoded in the memory at time $t=0$ using an encoding channel $\mathcal{E}:\B{\Hil_{\operatorname{data}}}\to \B{\Hil}$. After some time $t$, a decoding channel $\mathcal{D}:\B{\Hil}\to \B{\Hil_{\operatorname{data}}}$ is applied to recover the data, so that $\mathcal{D}\circ \Psi^t \circ \mathcal{E}\simeq_{\epsilon} \operatorname{id}$ approximately simulates the identity channel on the data system, where $\epsilon\in [0,1)$ is the error allowed in the recovery process. }
    \label{fig:memory-intro}
\end{figure}

In this setting, we want to determine the maximum amount of data that can be stored in the memory in such a way that it can be recovered with some error $\epsilon\in [0,1)$ after time $t\in \mathbb{N}$ (see Figure~\ref{fig:memory-intro}). Building such a quantum memory that is able to store information for a long time is crucial in order to build a reliable quantum computer. Consequently, this task has been studied from different perspectives, most notably from the viewpoint of designing error-correcting codes to develop a fault-tolerant quantum memory \cite{Tehral2015memory, Brown2016memory, gottesman2016surviving}. In this paper, we take a Shannon theoretic perspective, where the goal is to study the ultimate limits of data storage without placing any computational or physical restrictions on the encoding and decoding operations. Note that the asymptotic and memoryless theory of communication that we described in the introduction is not a good framework for this problem. This is because of two reasons. Firstly, the channel $\Psi$ might not act independently and identically on all qubits inside the memory, since the noise might act in a correlated fashion across some of the qubits. Secondly, since current quantum technologies can only coherently manipulate a few hundred qubits at most \cite{Preskill2018nisq, Bharti2022nisq}, it is pertinent to analyze storage capacities of memory devices with a small number of qubits ($n\sim 100$), thus making the asymptotic $n\to \infty$ limit rather unrealistic. These concerns are addressed by the framework of \emph{one-shot} information theory, where the goal is to understand the maximum amount of information $Q_{\epsilon}(\Phi)$ that can be sent through a \emph{single} use of a noisy channel $\Phi$ with some fixed allowed error $\epsilon\in [0,1)$ (Definitions~\ref{def:classical-protocol}-\ref{def:quantum-protocol}). Thus, in Section~\ref{sec:qms-storage}, we analyze the \emph{one-shot} $\epsilon-$error information transmission capacities $Q_{\epsilon}(\Psi^t)$ of the channels $\Psi^t$ as a function of $t,\epsilon$ and $\Psi$, with the focus being on the $t\to \infty$ limit. The peripheral space of the noise $\mathscr{X}(\Psi)$, which can be decomposed into a direct sum of matrix blocks (see Section~\ref{sec:periphery}):
\begin{align}
\mathscr{X} (\Psi) &:= \text{span}\{X\in \B{\Hil} : \exists \, \theta\in \mathbb{R} \text{ with } \Psi(X)= e^{i\theta} X\}   \\
&\simeq 0 \oplus \bigoplus_{k=1}^K (\B{\C{d_k}}\otimes \delta_k), \label{eq:phasespace-intro} 
\end{align}
plays a crucial role in our analysis. In a nutshell, we show that any information stored inside the peripheral space $\mathscr{X} (\Psi)$ is shielded from noise for an arbitrarily long time. Furthermore, as $t\to \infty$, we prove that this is the optimal way to store data to ensure that it can be recovered with good fidelity (c.f. Theorem~\ref{theorem:main-storage}). In fact, our convergence estimates show that this is the optimal strategy for time $t\gtrsim d^2\ln (d)$ scaling \emph{exponentially} with the number of qudits in memory, where $d=\dim\Hil=q^n$ (Section~\ref{subsec:convergence-storage}). A formal statement of the $t\to \infty$ limit of the capacities is given below, where we denote the \emph{one-shot} $\epsilon-$error quantum, private classical, and classical capacities of a channel $\Phi$ by $Q_{\epsilon}(\Phi), C^{\operatorname{p}}_{\epsilon}(\Phi),$ and $C_{\epsilon}(\Phi)$, respectively.

\begin{theorem}\label{theorem:main-storage-intro}
Let $\Psi : \B{\Hil}\to \B{\Hil}$ be a quantum channel and $\epsilon\in [0,1)$. There exist positive integers $K, d_1, \dots d_K$ that can be efficiently computed from $\Psi$ such that 
 \begin{align}
  \label{eq:Q}  \log (\max_{k} d_k) \leq \lim_{t\to \infty} Q_{\epsilon}(\Psi^t) &\leq \log (\max_{k} d_k) + \log (\frac{1}{1-\epsilon}), \\ 
    \log (\max_{k} d_k) \leq \lim_{t\to \infty} C^{\operatorname{p}}_{\epsilon}(\Psi^t) &\leq \log (\max_{k} d_k) + \log (\frac{1}{1-\epsilon}), \\ 
   \label{eq:C}  \log (\sum_{k} d_k) \leq \lim_{t\to \infty}C_{\epsilon}(\Psi^t) &\leq \log (\sum_{k} d_k ) + \log (\frac{1}{1-\epsilon}) .
 \end{align}
 These integers arise from the peripheral space $\mathscr{X}(\Psi) = \text{span}\{X\in \B{\Hil} : \exists \, \theta\in \mathbb{R} \text{ with } \Psi(X)= e^{i\theta} X\}$ which can be decomposed as $0 \oplus \bigoplus_{k=1}^K (\B{\C{d_k}}\otimes \delta_k)$, where $\delta_k$ are some fixed density operators. 
\end{theorem}

\begin{remark}
    Eqs.~\eqref{eq:Q} and \eqref{eq:C} were independently proved in \cite{singh2024zero} (for the $\epsilon=0$ case) and in \cite{fawzi2024error} (for arbitrary $\epsilon\in [0,1)$). The $\epsilon=0$ case of Eq.~\eqref{eq:C} is also proved in \cite{guan2016zero}. 
\end{remark}

As a special case, we also consider the memoryless setting where the noise acts independently and identically on all the qudits. Here, $\Psi=\Gamma^{\otimes n}$, where $\Gamma:\B{\C{q}} \to \B{\C{q}}$ is a local quantum channel that acts on a single qudit. In this case, we prove that since the peripheral space is multiplicative (Lemma~\ref{lemma:peripheral-multi}): $\mathscr{X} (\Psi)=\mathscr{X}(\Gamma)^{\otimes n}$, the infinite-time capacities become additive (see Section~\ref{subsec:IIDconvergence}): 
\begin{align}\label{eq:IID-intro}
  \forall n \in \mathbb{N}: \quad \log (\max_{k} d_k) \leq \lim_{t \to \infty} \frac{1}{n} Q_{\epsilon}((\Gamma^{\otimes n})^t) \leq \log (\max_{k} d_k) + \frac{1}{n}\log(\frac{1}{1-\epsilon}),
\end{align}
where the integers $d_k$ come now from the block decomposition of the peripheral space of the local channel $\mathscr{X}(\Gamma)$. Moreover, the infinite-time capacities are reached quite rapidly after time $t\gtrsim \ln (n)$ scaling \emph{logarithmically} with the number of qudits in memory. If $R_n := \lim_{t \to \infty} Q_{\epsilon}((\Gamma^{\otimes n})^t)/n$ is the optimal (infinite-time) storage \emph{rate} of the $n$-qudit memory device, the 
second inequality in Eq.~\eqref{eq:IID-intro} can be rearranged to show that the error incurred in data recovery $\epsilon$ is bounded as 
\begin{equation}
\epsilon \geq 1 - 2^{-n(R_n - \log (\max_k d_k) )}.    
\end{equation}
Thus, if we try to store qubits in the memory at a rate $R>\log \left(\max_k d_k \right)$, the error incurred in recovery 
approaches $1$ exponentially fast as $n\to \infty$. At the same time, the coding scheme used in Theorem~\ref{theorem:main-storage-intro} shows that any rate $R\leq\log (\max_k d_k)$ is achievable with \emph{exactly} zero error. Hence, for \emph{any} channel $\Gamma$, the strong converse property holds asymptotically as $t\to \infty$, and $\log \left(\max_k d_k \right)$ provides a sharp threshold between achievable and unachievable rates of data storage (see Figure~\ref{fig:error-rate-qms-storage}). 

\begin{figure}[H]
    \centering
    \includegraphics[width=0.5\linewidth]{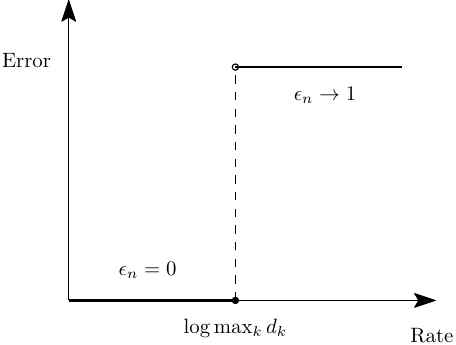}
    \caption{The error-rate curve for data storage in a quantum memory comprised of $n$ physical qudits (Hilbert space $\Hil\simeq (\C{q})^{\otimes n}$) experiencing IID Markovian noise modelled by a dQMS $(\Psi^t)_{t\in \mathbb{N}}$, where $\Psi = \Gamma^{\otimes n}$ for some quantum channel $\Gamma:\B{\C{q}}\to \B{\C{q}}$. One can store $\log\left( \max_k d_k \right)$ many logical qubits per physical qudit inside the largest block in the peripheral space of the local channel $\mathscr{X}(\Gamma)$ (see Eq.~\eqref{eq:phasespace-intro}) perfectly with zero data recovery error. Moreover, any attempt to store data at a higher rate $R>\log \left(\max_k d_k\right)$ fails as $t\to \infty$ with certainty, since the data recovery error $\epsilon_n\to 1$ exponentially as $n\to \infty$.}
    \label{fig:error-rate-qms-storage}
\end{figure}

From a practical standpoint, our analysis shows that typically, an $n-$qubit quantum memory with Markovian noise acting independently and identically on all qubits and a fixed time-independent global error correction mechanism becomes useless for storage after time $t\gtrsim n2^{2n}$ scaling exponentially with the number of qubits. In contrast, if the error correction is local, then the memory becomes useless much more quickly after time $t\gtrsim \ln(n)$; see Section~\ref{subsec:fault-tolerance}.    

\subsubsection{Data transmission}
In Section~\ref{sec:qms-transmission}, we come back to the point-to-point communication setting between two spatially separated parties, where Alice wants to use a noisy channel many times for faithful transmission of information to Bob. We make a crucial assumption that the noise in the communication link between Alice and Bob is \emph{Markovian}, which we model by considering a dQMS $(\Psi^l)_{l \in {\mathbb{N}}}$, where $\Psi :\B{\Hil}\to \B{\Hil}$ is a noisy channel, $\Psi^l = \Psi \circ \Psi \cdots \circ \Psi$ is the $l$-fold composition of $\Psi$ with itself, and $l \in {\mathbb{N}}$ plays the role of the length of communication link. Physically, the noise in each unit length of the link is modelled by $\Psi$, and since the noise is Markovian, the cumulative noise in a length $l$ segment is given by $\Psi^l$. We study capacities of `long' noisy communication links $\Psi^l$ of length $l \gtrsim d^2 \ln (d)$, where $d=\dim\Hil$. Mathematically, we are interested in the capacities of channels $\Phi:\B{\Hil}\to \B{\Hil}$ that are $l$-Markovian divisible\footnote{The notion of  {\em{divisibility}} of quantum channels has long been the focus of active research, especially in the study of open quantum systems. See e.g.~\cite{Wolf2008dividing, Rivas2014dividing, Breuer2016dividing, Chruciski2022dividing} and references therein.} for `large' $l$, i.e.,  channels $\Phi$ for which there exists another channel $\Psi$ such that $\Phi=\Psi^l$ with $l \gtrsim d^2\ln (d)$. We show that the capacities of such highly Markovian divisible channels have very nice properties:

\begin{itemize}
    \item All the capacities can be efficiently approximated.
    \item The quantum and private capacities approximately satisfy the
    strong converse property.
    \item The quantum and private capacities approximately coincide.
    \item All the capacities are approximately strongly additive.
\end{itemize}

All these approximate results become exact for zero-error capacities of channels that are $l$-Markovian divisible with $l\geq d^2$. We provide semi-formal statements of these results below.
\smallskip

\begin{theorem} \label{thm:main-transmission-zero}
 Let $\Phi : \B{\Hil}\to \B{\Hil}$ be an $l$-Markovian divisible channel with $l\geq d^2$ and $d=\dim \Hil$. There exist integers $K, d_1, \dots d_K$ that can be efficiently computed from $\Phi$ such that 
 \begin{align*}
    C_{\operatorname{zero}}(\Phi) &= \log \left(\sum_k d_k \right), \\ 
    P_{\operatorname{zero}}(\Phi) &= \log \left(\max_k d_k \right) = Q_{\operatorname{zero}}(\Phi).
 \end{align*}
 These integers arise from the peripheral space $\mathscr{X}(\Phi) = \text{span}\{X\in \B{\Hil} : \exists \, \theta\in \mathbb{R} \text{ with } \Phi(X)= e^{i\theta} X\}$ which can be decomposed as $0 \oplus \bigoplus_{k=1}^K (\B{\C{d_k}}\otimes \delta_k)$, where $\delta_k$ are some fixed density operators.
 
 Moreover, for any other $l'$-Markovian divisible channel $\Gamma : \B{\Hil} \to \B{\Hil}$ with $l'\geq d^2$, 
 \begin{align*}
     C_{\operatorname{zero}}(\Phi \otimes \Gamma) &= C_{\operatorname{zero}}(\Phi) + C_{\operatorname{zero}}(\Gamma), \\
     P_{\operatorname{zero}}(\Phi \otimes \Gamma) &= P_{\operatorname{zero}}(\Phi) + P_{\operatorname{zero}}(\Gamma), \\
     Q_{\operatorname{zero}}(\Phi \otimes \Gamma) &= Q_{\operatorname{zero}}(\Phi) + Q_{\operatorname{zero}}(\Gamma).
 \end{align*}
\end{theorem}

Let us highlight that to the best of our knowledge, the class of highly Markovian divisible quantum channels as in Theorem~\ref{thm:main-transmission-zero} provides the first example of a non-trivial family of channels for which all the zero-error capacities are additive.  This includes the class of $\infty$-Markovian divisible channels \cite{Denisov1989infdivisble, Wolf2008dividing}, which further includes the class of continuous Quantum Markov semigroups (QMS) generated by a Lindbladian \cite{Gorini1976qms, Lindblad1976qms}. The reason for this additivity boils down to the fact that the non-commutative confusability graphs \cite{Duan2013noncomm} of such channels `look like' $*-$algebras (see Theorems~\ref{thm:stab-opsys}, \ref{theorem:main-zero}) for which the graph independence numbers are nicely behaved (Lemma~\ref{lemma:alpha-algebra}). 

\begin{theorem}\label{thm:main-transmission-nonzero}
\label{thm:main-finitetime}
    Let $\Psi:\cL (\cH)\to \cL (\cH)$ be a channel with $d=\dim \Hil$. There exist integers $K, d_1, \dots d_K$ that can be efficiently computed from $\Psi$ as in Theorem~\ref{thm:main-transmission-zero} such that 
\begin{align*}
       \log \left( \max_k d_k \right) \leq Q(\Psi^l) &\leq P^{\dagger}(\Psi^l) \leq \log \left( \max_k d_k \right) + \log (1+ \frac{\delta_l d}{2}), \\ 
       \log \left( \sum_k d_k \right) \leq C(\Phi^l) &\leq \log \left( \sum_k d_k \right) + \delta_l \log (d^2 -1 ) + 2h \left(\frac{\delta_l}{2} \right).
\end{align*}
Moreover, for any other channel $\Gamma: \B{\Hil}\to \B{\Hil}$, 
\begin{align*}
       \log \left( \sum_k d_k \right) + C(\Gamma) \leq C(\Psi^l &\otimes \Gamma) \leq \log\left( \sum_k d_k \right) + C(\Gamma) + \delta_l \log (d^4 -1 ) + 2h \left(\frac{\delta_l}{2} \right), \\
       \log(\max_k d_k) + P(\Gamma) \leq P(\Psi^l &\otimes \Gamma) \leq \log(\max_k d_k) + P(\Gamma) + 2\delta_l \log (d^4 -1 ) + 4h \left(\frac{\delta_l}{2} \right),
\end{align*}
\begin{align*}
    \log(\max_k d_k) + Q(\Gamma) \leq Q(\Psi^l &\otimes \Gamma) \leq \log(\max_k d_k) + Q(\Gamma) + \delta_l \log (d^4 -1 ) + 2h \left(\frac{\delta_l}{2} \right).
\end{align*}
Here, $h(\epsilon):= -\eps\log \eps - (1-\eps)\log (1-\eps)$ and $\delta_l=\norm{\Psi^l-\Psi^l_{\infty}}_{\diamond}\leq \kappa\mu^l \to 0$ as $l\to \infty$, where
\begin{itemize}
    \item $\mu = \operatorname{spr}(\Psi-\Psi_{\infty}) <1$ is such that $1-\mu$ is the spectral gap of $\Psi$, and
    \item $\kappa$ depends on $\Psi$, $l$, and $d=\dim \cH$.
\end{itemize}
The lower bounds above hold for all $l\in\mathbb{N}$ and the upper bounds hold for $l \gtrsim d^2 \ln (d)$.
\end{theorem}

The crucial idea behind the proof of Theorem~\ref{thm:main-transmission-nonzero} is the fact that for {\em{any}} quantum channel $\Psi$, the corresponding asymptotic channel $\Psi_\infty$ satisfies the strong converse property for the classical, private classical, and quantum capacities, and these are also strongly additive (see Theorems~\ref{thm:QPinf}, \ref{thm:Cinf}, \ref{thm:asym-strongadd}). These properties can then be lifted to the finite length regime by using continuity arguments and the convergence estimate $\norm{\Psi^l-\Psi^l_{\infty}}_{\diamond} \leq \kappa\mu^l \to 0$ (see Theorems~\ref{thm:QPasymp}, \ref{thm:Casmpy}, \ref{theorem:strong-add}).

Finally, since realistic quantum devices can process only a small finite number of qubits $(n\sim 100)$ at a time \cite{Preskill2018nisq, Bharti2022nisq}, it is pertinent to analyze optimal rates of information transmission via a fixed number of uses of a given channel. Hence, we provide finite-blocklength bounds on the capacities of `long' noisy Markovian communication links of length $l\gtrsim d^2\ln(d)$ for a fixed number of channel uses $n\in \mathbb{N}$ and transmission error $\epsilon\in [0,1)$ that are of the following form:
\begin{align*}
    \log(\max_k d_k) \leq \frac{1}{n} Q_{\epsilon}((\Psi^l)^{\otimes n})) &\leq \log(\max_k d_k) +  \frac{\alpha}{\alpha-1} \log (1+\frac{\delta_l d^{\frac{\alpha-1}{\alpha}}}{2}) + \frac{\alpha}{n(\alpha-1)} \log (\frac{n^{d^2}}{1-\epsilon}), \\
    \log(\max_k d_k) \leq \frac{1}{n} Q_{\epsilon}((\Psi^l)^{\otimes n})) &\leq \log(\max_k d_k) +  \log (1+\frac{\delta_l d}{2}) + \frac{1}{n} \log (\frac{1}{1-\epsilon}),
\end{align*}
where the lower bound holds for all $l\in\mathbb{N}$ and the upper bound holds for $l \gtrsim d^2 \ln (d)$ and $\alpha>1$. Similar bounds are derived for the private classical capacity as well (see Theorems~\ref{theorem:QPfinite-alpha}, \ref{theorem:QPfinite-max}, \ref{theorem:QPfinite-assisted}). Observe that as $n\to \infty$ we recover (and improve) the asymptotic bound stated in Theorem~\ref{thm:main-transmission-nonzero}.

\begin{remark}\label{remark}
    The efficient computability of the capacities in Theorems~\ref{theorem:main-storage-intro}, \ref{thm:main-transmission-zero}, and \ref{thm:main-transmission-nonzero} follows from the fact that for a given channel $\Phi$, the structure of its peripheral space $\mathscr{X}(\Phi)$ can be efficiently computed. The linear structure of $\mathscr{X}(\Phi)$ can be efficiently computed using the algorithm given in \cite{BlumeKohout2010algebra}, following which the algebraic structure (i.e., the integers $K,d_1,\ldots ,d_K$ in Eq.~\eqref{eq:phasespace-intro}) can be efficiently computed using the algorithms given in \cite{zarikan2003algebra, Holbrook2003commutant, Guan2018algebra, fawzi2024error}. Note that this is in stark contrast with the general picture, where the capacities of a quantum channel are not even known to be computable \cite{wolf2011undecidable, wolf2024decidable}.
\end{remark}

\subsection{Outline of the paper} 
\begin{itemize}
    \item We review some basic mathematical preliminaries in Section~\ref{sec:prelim}. 
    \item Next, we consider the problem of information \emph{storage} under Markovian noise in Section~\ref{sec:qms-storage}. The main result of this section, which provides tight upper and lower bounds on the long-time storage capacities of quantum memories experiencing Markovian noise, is stated and proved in Theorem~\ref{theorem:main-storage}. Convergence bounds on the time after which the infinite-time capacities are reached are derived in Section~\ref{subsec:convergence-storage}. In Section~\ref{sec:consequences-qms-storage}, we also discuss several applications of Theorem~\ref{theorem:main-storage}. In particular, we examine the special case of memories experiencing IID noise (Section~\ref{subsec:IIDconvergence}) and discuss how lower bounds on the space overhead of fault-tolerant quantum memories can be derived using our techniques (Section~\ref{subsec:fault-tolerance}). In Section~\ref{sec:examples-qms-storage}, we compute storage capacities for some specific examples of noise models. 
    \item Next, we consider the problem of information \emph{transmission} under Markovian noise in Section~\ref{sec:qms-transmission}. The main results of this section in the zero-error regime are presented in Section~\ref{sec:main-zero}. In particular, the primary technical ingredient for these results, namely the structure theorem for the stabilized operator system of any quantum channel, is discussed in Section~\ref{sec:stablized}. The non-zero-error results are discussed in Section~\ref{sec:main:non-zero}. The main capacity bounds are presented in Sections~\ref{subsec:non-zero-classical}-\ref{subsec:non-zero-finiteblock}. The results on strong additivity and rates of convergence are presented in Sections~\ref{subsec:strong} and \ref{subsec:convergence-transmission}, respectively.

    \item Finally, we provide a concluding discussion in Section~\ref{sec:conclude}. Here, we compare the storage vs transmission perspectives presented in the previous two sections and detail a few directions for future research.
    \end{itemize}

\section{Preliminaries} \label{sec:prelim}
We denote quantum systems by capital letters $A,B,C$ and the associated (finite-dimensional) Hilbert spaces by $\Hil_A, \Hil_B$ and $\Hil_C$ with dimensions $d_A, d_B$ and $d_C$, respectively. For a joint system $AB$, the associated Hilbert space is $\Hil_A\otimes \Hil_B$. The space of linear operators acting on $\Hil_A$ is denoted by $\B{\Hil_A}$ and the convex set of quantum states or density operators (i.e.~positive semi-definite operators in $\B{\Hil_A}$ with unit trace) is denoted by $\State{\Hil_A}$. For a unit vector $\ket{\psi}\in \Hil_A$, the pure state $\ketbra{\psi}\in \State{\Hil_A}$ is denoted by $\psi$. 

A quantum channel $\Phi:\B{\Hil_A}\to \B{\Hil_B}$ is a linear, completely positive, and trace-preserving map. By Stinespring's dilation theorem, for every quantum channel $\Phi : \B{\Hil_A}\to \B{\Hil_B}$, there exists an isometry $V: \Hil_A \to \Hil_B\otimes \Hil_E$ (called a Stinespring isometry) such that for all $X \in \B{\Hil_A}$,  $\Phi (X) = \Tr_E (VX V^{\dagger})$, where $\Tr_E$ denotes the partial trace operation over the subsystem $E$ (often called the environment). The corresponding complementary channel $\Phi^c : \B{\Hil_A}\to \B{\Hil_E}$ is then defined as $\Phi^c (X) = \Tr_B (VX V^{\dagger}).$ The adjoint $\Phi^*$ of a quantum channel $\Phi : \B{\Hil_A}\to \B{\Hil_B}$ is defined through the following relation: 
\begin{equation}
    \forall X \in \B{\Hil_A}, \, \forall Y \in \B{\Hil_B} : \quad \Tr(Y \Phi(X))= \Tr (\Phi^*(Y) X).
\end{equation}

\begin{remark}
    To make the systems on which an operator or a channel acts more explicit, we sometimes denote operators $X\in \B{\Hil_A}$ by $X_A$  and linear maps $\Phi:\B{\Hil_A}\to \B{\Hil_B}$ by $\Phi_{A\to B}$.
\end{remark}

For a bipartite operator $X_{RA}$ and a linear map $\Phi_{A\to B}$, we use the shorthand $\Phi_{A\to B}(X_{RA})$ to denote $(\id_R \otimes \Phi_{A\to B}) (X_{RA})$, where $\id_R $ is the identity map on $\B{\Hil_R}$. Similarly, $X_R$ and $X_A$ denote the reduced operators on $R$ and $A$, respectively, i.e., $X_R := \Tr_A X_{RA}$ and $X_A := \Tr_R X_{RA}$. 

The \emph{trace norm} of a linear operator $X\in \B{\Hil_A}$ is defined as $\norm{X}_1 := \Tr \sqrt{X^{\dagger}X}$. The \emph{diamond norm} of a linear map $\Phi: \B{\Hil_A}\to \B{\Hil_B}$ is defined as
\begin{equation}
    \norm{\Phi}_{\diamond} := \sup_{\norm{X}_1\leq 1} \norm{\Phi_{A\to B}(X_{RA})}_1,
\end{equation}
where the supremum is over all $X\in \B{\Hil_R\otimes \Hil_A}$ and $d_R\in \mathbb{N}$ with $\norm{X}_1\leq 1$. Here, if $\Phi$ is a quantum channel, the optimization can be restricted to pure states $\psi\in \State{\Hil_R \otimes \Hil_A}$ with $d_R=d_A$ \cite[Chapter 3]{watrous2018theory}.

The \emph{fidelity} between two quantum states $\rho, \sigma \in \State{\Hil_A}$ is defined as $F(\rho, \sigma):= \norm{\sqrt{\rho}\sqrt{\sigma}}^2_1$. The \emph{fidelity} of a quantum channel $\Phi:\B{\Hil_A}\to \B{\Hil_A}$ is defined as
\begin{equation}
    F(\Phi) := \inf_{\psi_{RA}} \bra{\psi}_{RA} (\id_R \otimes \Phi_{A\to A})(\psi_{RA}) \ket{\psi}_{RA},
\end{equation}
where the infimum is over all pure states $\psi\in \State{\Hil_R \otimes \Hil_A}$ with $d_R=d_A$.

\subsection{Channel capacities}

In this section, we introduce the different information transmission capacities of quantum channels. We begin with the one-shot framework, which will later be used to define the asymptotic and memoryless capacities. In what follows, we denote a classical message set of size $\mathscr{M}\in \mathbb{N}$ by 
\begin{equation}
    [\mathscr{M}] := \{0,1,\ldots ,\mathscr{M}-1 \}.
\end{equation}

\subsubsection{One-shot capacities}

\begin{definition} \label{def:classical-protocol} (Classical communication) \,
    Let $\mathscr{M}\in \mathbb{N}$ and $\epsilon\in [0,1)$. An $(\mathscr{M},\epsilon)$ \emph{classical code} for a channel $\Phi:\B{\Hil_A}\to \B{\Hil_B}$ is defined by the following:
\begin{itemize}
    \item Encoding states $\rho^m_{A}\in \State{\Hil_A}$ that Alice uses to encode a message $m\in [\mathscr{M}]$,
    \item Decoding POVM $\{\Lambda^m_B\}_{m \in [\mathscr{M}]}$ that Bob uses to decode the message,
\end{itemize}
such that for each message $m$, the probability that Bob successfully decodes the intended message is
\begin{equation}
    \Tr [\Lambda^m_{B}  (\Phi_{A\to B}(\rho^m_A))  ] \geq 1 -\epsilon.
\end{equation} 

The \emph{one-shot} $\epsilon-$\emph{error classical capacity} of $\Phi$ is defined as
\begin{equation}
    C_{\epsilon}(\Phi):= \log \max \{ \mathscr{M}\in \mathbb{N}: \exists (\mathscr{M},\epsilon) \text{ classical code for } \Phi \}.
\end{equation}
\end{definition}

\begin{definition} \label{def:private-protocol}
(Private classical communication) \, Let $\mathscr{M}\in \mathbb{N}$ and $\epsilon\in [0,1)$. An $(\mathscr{M},\epsilon)$ \emph{private classical code} for a channel $\Phi: \B{\Hil_A}\to \B{\Hil_B}$ is defined by the following:
\begin{itemize}
    \item Encoding states $\rho^m_{A}\in \State{\Hil_A}$ that Alice uses to encode a message $m\in [\mathscr{M}]$.
    \item Decoding POVM $\{\Lambda^m_B\}_{m\in [\mathscr{M}]}$ with an associated channel $\mathcal{D}_{B\to M}$ defined as \\ $\mathcal{D}(\cdot) = \sum_m \Tr (\Lambda^m_B (\cdot))\ketbra{m}_{M}$ that Bob uses to decode the message,
\end{itemize}
such that for each message $m\in [\mathscr{M}]$,
\begin{equation}\label{eq:F-privatecode}
    F(\ketbra{m}_{M} \otimes \sigma_E , \mathcal{D}_{B\to M}\circ \mathcal{V}_{A\to BE} (\rho^m_A)) \geq 1-\epsilon,
\end{equation}
where $\sigma_E$ is some fixed state independent of $m$ and $\mathcal{V}_{A\to BE}(\cdot ) = V (\cdot) V^{\dagger}$, where $V:\Hil_A\to \Hil_B \otimes \Hil_E$ is a Stinespring isometry of $\Phi_{A\to B}$. 

The \emph{one-shot $\epsilon-$error private classical capacity} of $\Phi$ is defined as
\begin{equation}
    C_{\epsilon}^{\text{p}}(\Phi):= \log \max \{ \mathscr{M}\in \mathbb{N}: \exists (\mathscr{M},\epsilon) \text{ private classical code for } \Phi \}.
\end{equation}

\end{definition}

According to the data processing inequality for the fidelity function, the privacy condition in Eq.~\eqref{eq:F-privatecode} implies that
\begin{align}
    \forall m\in [\mathscr{M}]: \quad \Tr [\Lambda^m_{B}  (\Phi_{A\to B}(\rho^m_A)  ] &\geq 1 -\epsilon, \\
    F(\sigma_E, \Phi^c_{A\to E}(\rho^m_A)) &\geq 1-\epsilon,
\end{align}
where $\Phi^c_{A \to E}$ denotes the complementary channel associated with the isometry $\mathcal{V}_{A\to BE}$. Thus, the goal in private classical communication is for Alice to faithfully transmit classical messages to Bob while at the same time ensuring that no potential eavesdropper (that is modeled via the environment $E$) gets any information about the transmitted messages.

\begin{definition} \label{def:entanglement-classical-protocol} (Entanglement-assisted classical communication) \, \\
Let $\mathscr{M}\in \mathbb{N}$ and $\epsilon\in [0,1)$. An $(\mathscr{M},\epsilon)$ \emph{entanglement-assisted classical code} for a channel $\Phi: \B{\Hil_A}\to \B{\Hil_B}$ is defined by the following:
\begin{itemize}
    \item An entangled state $\psi_{A'B'}\in \State{\Hil_{A'}\otimes \Hil_{B'}}$ shared between Alice and Bob,
    \item Encoding channels $\mathcal{E}^m_{A'\to A}$ that Alice uses to encode a message $m\in [\mathscr{M}]$,
    \item Decoding POVM $\{\Lambda^m_{BB'}\}_{m\in [\mathscr{M}]}$ that Bob uses to decode the message,
\end{itemize}
such that for each message $m$, 
\begin{equation}
   \Tr [\Lambda^m_{BB'}  (\Phi_{A\to B} \circ \mathcal{E}^m_{A'\to A} (\psi_{A'B'}) ) ] \geq 1 -\epsilon.
\end{equation}
The \emph{one-shot $\epsilon-$error entanglement-assisted classical capacity} of $\Phi$ is defined as
\begin{align}
    C_{\epsilon}^{\operatorname{ea}}(\Phi):= \log \max \{ \mathscr{M}\in \mathbb{N}: \exists (\mathscr{M},\epsilon)& \text{ entanglement-assisted classical code for } \Phi\}.
\end{align}
\end{definition}

\begin{definition} \label{def:quantum-protocol} (Quantum communication) \, 
    Let $d\in \mathbb{N}$ and $\epsilon\in [0,1)$. A $(d,\epsilon)$ \emph{quantum code} for a channel $\Phi: \B{\Hil_A}\to \B{\Hil_B}$ is defined by of the following:
\begin{itemize}
    \item An encoding channel $\mathcal{E}_{A'\to A}$ that Alice uses to encode quantum information,
    \item A decoding channel $\mathcal{D}_{B\to A'}$ that Bob uses to decode the information,
\end{itemize}
such that for every pure state $\psi_{RA'}$ with $d=d_R=d_{A'}$:
\begin{equation}\label{eq:quantum-code}
   \bra{\psi}_{RA'} \mathcal{D}_{B\to A'}\circ \Phi_{A\to B}\circ \mathcal{E}_{A'\to A}(\psi_{RA'})\ket{\psi}_{RA'} \geq  1-\epsilon.
\end{equation}
In other words, the fidelity of $\mathcal{D}\circ \Phi \circ \mathcal{E}$ satisfies $F(\mathcal{D}\circ \Phi \circ \mathcal{E})\geq 1-\epsilon$. The \emph{one-shot $\epsilon-$error quantum capacity} of $\Phi$ is defined as
\begin{equation}
    Q_{\epsilon}(\Phi):= \log \max \{ d\in \mathbb{N}: \exists (d,\epsilon) \text{ quantum code for } \Phi \}.
\end{equation}
\end{definition}

By definition, all one-shot $\epsilon$-error capacities take non-negative values and are monotonic in $\epsilon$, i.e., for $0\leq\epsilon\leq \epsilon'<1$ and $\mathbb{Q}\in \{C, C^{\operatorname{p}}, C^{\operatorname{ea}},Q \}$, we have $0\leq \mathbb{Q}_{\epsilon}(\Phi)\leq \mathbb{Q}_{\epsilon'}(\Phi)$. Moreover, for any channel $\Phi_{A\to B}$ and $\epsilon\in [0,1)$, it follows from the definitions that $C^{\operatorname{p}}_{\epsilon}(\Phi)\leq C_{\epsilon}(\Phi)\leq C^{\operatorname{ea}}_{\epsilon}(\Phi)$. Similarly, for $\epsilon\in (0,1)$, one can show that $Q_{\epsilon/2}(\Phi) \leq C^{\operatorname{p}}_{\epsilon}(\Phi) + 1$ \cite{khatri2024principles}. In the zero-error case, it is possible to obtain a stronger inequality: $Q_0(\Phi)\leq C^{\operatorname{p}}_0(\Phi)$, see Lemma~\ref{lemma:Q0<=P0}. Exact computation of the one-shot capacities is hard because of the complicated optimizations involved in their definitions \cite{Costa2010oneshot-NPhard}. However, in this paper, we will see that for any channel $\Psi:\B{\Hil}\to \B{\Hil}$, the one-shot capacities of $m-$fold sequential compositions $\Psi^m=\Psi \circ \Psi \circ \ldots \Psi$ can be efficiently computed in the limit of $m\to \infty$.

Below, we note some simple yet crucial properties of the one-shot capacity functions. 

\begin{lemma}\label{lemma:bottleneck-oneshot} (Bottleneck inequalities)
    Let $\Psi_{A\to B}$, $\Phi_{B\to C}$ be quantum channels. Then, for $\epsilon\in [0,1)$ and $\mathbb{Q}\in \{C, C^{\operatorname{p}}, C^{\operatorname{ea}},Q \} $,
    \begin{equation}
        \mathbb{Q}_{\epsilon}(\Phi\circ \Psi) \leq \min ( \mathbb{Q}_{\epsilon}(\Phi), \mathbb{Q}_{\epsilon}(\Psi) ).
    \end{equation}
\end{lemma}
\begin{proof}
Consider a $(d,\epsilon)$ quantum code for the channel $(\Phi\circ \Psi)_{A\to C}$, defined by the encoder and decoder pair $(\mathcal{E}_{A'\to A}, \mathcal{D}_{C\to A'})$ with $d=d_{A'}$, such that for any pure state $\psi_{RA'}$:
    \begin{equation}
   \bra{\psi_{RA'}} \mathcal{D}_{C\to A'}\circ (\Phi\circ\Psi)_{A\to C}\circ \mathcal{E}_{A'\to A}(\psi_{RA'})\ket{\psi_{RA'}} \geq  1-\epsilon.
\end{equation}
Now, by absorbing either $\Psi$ into the encoding channel $\mathcal{E}_{A'\to A}$ or $\Phi$ into the decoding channel $\mathcal{D}_{C\to A'}$, we see that the same $(d,\epsilon)$ code works for $\Phi$ and $\Psi$, which proves the desired result. We leave very similar proofs for the other capacities to the reader.
\end{proof}

\begin{lemma}\label{lemma:epsilon-delta} (Continuity-like bounds)
    Let $\Phi, \Psi : \B{\Hil_A}\to \B{\Hil_B}$ be quantum channels such that $\norm{\Phi-\Psi}_{\diamond}= \delta$. Then, for $\mathbb{Q}\in \{Q,C,C^{\operatorname{ea}} \}$ and $\epsilon\in [0,1)$ such that $\epsilon+\delta<1$,
    \begin{align}
        \mathbb{Q}_{\epsilon}(\Phi) &\leq \mathbb{Q}_{\epsilon+\delta}(\Psi)
    \end{align}
\end{lemma}
\begin{proof}
Consider a $(\mathscr{M},\epsilon)$ classical code $\{\rho^m_A, \Lambda^m_B \}_{m\in [\mathscr{M}]}$ for $\Phi_{A\to B}$ (see Definition~\ref{def:classical-protocol}) such that
\begin{equation}
\forall m\in [\mathscr{M}]: \quad    \Tr [\Lambda^m_{B}  (\Phi_{A\to B}(\rho^m_A)  ] \geq 1 -\epsilon.
\end{equation}
For each $m$, it is then easy to see that
\begin{align}
    \Tr [\Lambda^m_B  (\Psi_{A\to B}(\rho^m_A)  ] &= \Tr [\Lambda^m_B  (\Phi_{A\to B}(\rho^m)  ] - \Tr [\Lambda^m_B  (\Phi - \Psi)_{A\to B}(\rho^m_A)  ]  \nonumber \\
    &\geq 1- (\epsilon + \delta),
\end{align}
where the last inequality follows from the fact that 
\begin{align}
    \Tr [\Lambda^m_B  (\Phi - \Psi)_{A\to B}(\rho^m_A)  ] &\leq \norm{\Lambda^m_B}_{\infty} \norm{(\Phi-\Psi)_{A\to B}(\rho^m_A)}_1 \nonumber \\
    &\leq \norm{\Phi-\Psi}_{\diamond} = \delta.
\end{align}
Hence, $\{\rho^m_A, \Lambda^m_B \}_{m\in [\mathscr{M}]}$ is a $(\mathscr{M},\epsilon+\delta)$ classical code for $\Psi_{A\to B}$.

Similarly, consider a $(d,\epsilon)$ quantum code $(\mathcal{E}_{A'\to A}, \mathcal{D}_{B\to A'})$ for $\Phi_{A\to B}$ (see Definition~\ref{def:quantum-protocol}) satisfying
\begin{equation}
\bra{\psi}_{RA'} \mathcal{D}_{B\to A'}\circ \Phi_{A\to B}\circ \mathcal{E}_{A'\to A}(\psi_{RA'})\ket{\psi}_{RA'} \geq  1-\epsilon
\end{equation}
for any pure state $\psi_{RA'}$, where $d=d_{A'}=d_R$. Then, it is easy to see that
\begin{align}
    \bra{\psi}&_{RA'} \mathcal{D}_{B\to A'}\circ \Psi_{A\to B}\circ \mathcal{E}_{A'\to A}(\psi_{RA'})\ket{\psi}_{RA'} \nonumber \\
    &=\Tr[\psi_{RA'} (\mathcal{D}_{B\to A'}\circ \Psi_{A\to B}\circ \mathcal{E}_{A'\to A}(\psi_{RA'}) )] \nonumber  \\ 
    &= \Tr[\psi_{RA'} (\mathcal{D}_{B\to A'}\circ \Phi_{A\to B}\circ \mathcal{E}_{A'\to A}(\psi_{RA'})) ] - \Tr[\psi_{RA'} (\mathcal{D}_{B\to A'}\circ (\Phi-\Psi)_{A\to B}\circ \mathcal{E}_{A'\to A}(\psi_{RA'})) ] \nonumber \\
    &\geq 1-(\epsilon + \delta),
\end{align}
where the last inequality follows from the fact that
\begin{align}
    \Tr[\psi_{RA'} (\mathcal{D}_{B\to A'}\circ (\Phi-\Psi)_{A\to B}\circ \mathcal{E}_{A'\to A}(\psi_{RA'})) ] &\leq \norm{\psi_{RA'}}_{\infty} \norm{\mathcal{D}_{B\to A'}\circ (\Phi-\Psi)_{A\to B}\circ \mathcal{E}_{A'\to A}(\psi_{RA'})}_1 \nonumber \\
    &\leq \norm{\mathcal{D}_{B\to A'}\circ (\Phi-\Psi)_{A\to B}\circ \mathcal{E}_{A'\to A}}_{\diamond} \nonumber \\
    &\leq \norm{\Phi-\Psi}_{\diamond} = \delta,
\end{align}
where we have used sub-multiplicativity of the diamond norm \cite[Proposition 3.48]{watrous2018theory} and the fact that $\norm{\Phi}_{\diamond}=1$ for any channel $\Phi$ \cite[Proposition 3.44]{watrous2018theory}. Thus, $(\mathcal{E}_{A'\to A}, \mathcal{D}_{A\to B'})$ is a $(d,\epsilon+\delta)$ quantum code for $\Psi_{A\to B}$. The proof for the entanglement-assisted classical capacity works similarly.
\end{proof}

\subsubsection{Asymptotic capacities}

By using the one-shot channel capacities introduced in the previous section, we can succinctly define the (asymptotic) channel capacities as the optimal rates of information transmission via asymptotically many uses of a given channel.

\begin{definition}\label{def:capacity} 
    Let $\Phi:\B{\Hil_A}\to \B{\Hil_B}$ be a quantum channel. We define the \emph{classical}, \emph{entanglement-assisted classical}, \emph{private classical}, and \emph{quantum} capacity of $\Phi$, respectively, as
 \begin{align*}
     C(\Phi) &:= \inf_{\epsilon\in (0,1)} \liminf_{n\to \infty} \frac{1}{n} C_{\epsilon}(\Phi^{\otimes n}), \\
     C_{\operatorname{ea}}(\Phi) &:= \inf_{\epsilon\in (0,1)} \liminf_{n\to \infty} \frac{1}{n} C^{\operatorname{ea}}_{\epsilon}(\Phi^{\otimes n}), \\
     P(\Phi) &:= \inf_{\epsilon\in (0,1)} \liminf_{n\to \infty} \frac{1}{n} C^{\operatorname{p}}_{\epsilon}(\Phi^{\otimes n}), \\  
     Q(\Phi) &:= \inf_{\epsilon\in (0,1)} \liminf_{n\to \infty} \frac{1}{n} Q_{\epsilon}(\Phi^{\otimes n}). 
 \end{align*}   
 The corresponding \emph{strong converse} capacities are defined as  
  \begin{align*}
     C^{\dagger}(\Phi) &:= \sup_{\epsilon\in (0,1)} \limsup_{n\to \infty} \frac{1}{n} C_{\epsilon}(\Phi^{\otimes n}), \\
     C_{\operatorname{ea}}^{\dagger}(\Phi) &:= \sup_{\epsilon\in (0,1)} \limsup_{n\to \infty} \frac{1}{n} C^{\operatorname{ea}}_{\epsilon}(\Phi^{\otimes n}), \\ 
     P^{\dagger}(\Phi) &:= \sup_{\epsilon\in (0,1)} \limsup_{n\to \infty} \frac{1}{n} C^{\operatorname{p}}_{\epsilon}(\Phi^{\otimes n}), \\   
     Q^{\dagger}(\Phi) &:= \sup_{\epsilon\in (0,1)} \limsup_{n\to \infty} \frac{1}{n} Q_{\epsilon}(\Phi^{\otimes n}). 
 \end{align*} 
\end{definition}

By definition, the strong converse capacities are always at least as large as the normal capacities:
\begin{align}
    Q (\Phi) \leq Q^{\dagger}(\Phi), \quad P(\Phi)\leq P^{\dagger}(\Phi), \quad C(\Phi) \leq C^{\dagger}(\Phi), \quad C_{\operatorname{ea}}(\Phi)\leq C^{\dagger}_{\operatorname{ea}}(\Phi).
\end{align}
As already noted, if $C(\Phi)=C^{\dagger}(\Phi)$, the capacity provides a sharp divide between achievable and unachievable rates of communication, and we say that $\Phi$ satisfies the \emph{strong converse} property for classical capacity. One can similarly define the strong converse property for other capacities. It turns out that all quantum channels satisfy the strong converse property for the entanglement-assisted classical capacity \cite{Bennett1999entanglement, Bennett2002entanglement, Holevo2002entanglement}. For the other capacities, determining if the strong converse property holds for all quantum channels has been a fundamental open problem (see \cite[Section 3]{cheng2024strong} and references therein). In this work, we will see that $m-$fold sequential iterates $\Psi^m=\Psi \circ \Psi \circ \ldots \circ \Psi$ of \emph{any} channel $\Psi:\B{\Hil}\to \B{\Hil}$ satisfy the strong converse property for all the capacities as $m\to \infty$.

The capacities satisfy the following inequalities:
\begin{alignat}{2}\label{eq:Q<P<C}
    Q(\Phi) &\leq \,P(\Phi) &&\leq C(\Phi) \leq C_{\operatorname{ea}}(\Phi) \\ 
    Q^{\dagger}(\Phi) &\leq P^{\dagger}(\Phi) &&\leq C^{\dagger}(\Phi) \leq C_{\operatorname{ea}}(\Phi),
\end{alignat}
where the relations between $P,C,$ and $C_{\operatorname{ea}}$ are immediate from Definitions~\ref{def:classical-protocol}-\ref{def:entanglement-classical-protocol}, and the relation between $Q$ and $P$ follows from \cite{Devetak2005capacity} (see also \cite[Theorem 13.6.1]{Wilde2013book}).

\begin{definition}\label{def:zero-error-capacity}
    Let $\Phi:\B{\Hil_A}\to \B{\Hil_B}$ be a quantum channel. We define the \emph{zero-error} \emph{classical}, \emph{private classical}, \emph{entanglement-assisted classical,} and \emph{quantum} capacity of $\Phi$, respectively, as
\begin{align*}
     C_{\operatorname{zero}}(\Phi) &:= \lim_{n\to \infty} \frac{1}{n} C_{0}(\Phi^{\otimes n}) = \sup_{n\in \mathbb{N}} \frac{1}{n} C_{0}(\Phi^{\otimes n}), \\
     P_{\operatorname{zero}}(\Phi) &:= \lim_{n\to \infty} \frac{1}{n} C^{\operatorname{p}}_{0}(\Phi^{\otimes n})= \sup_{n\in \mathbb{N}} \frac{1}{n} C^{\operatorname{p}}_{0}(\Phi^{\otimes n}), \\ 
     C^{\operatorname{ea}}_{\operatorname{zero}}(\Phi) &:= \lim_{n\to \infty} \frac{1}{n} C^{\operatorname{ea}}_{0}(\Phi^{\otimes n}) = \sup_{n\in \mathbb{N}} \frac{1}{n} C^{\operatorname{ea}}_{0}(\Phi^{\otimes n}), \\
     Q_{\operatorname{zero}}(\Phi) &:= \lim_{n\to \infty} \frac{1}{n} Q_{0}(\Phi^{\otimes n}) = \sup_{n\in \mathbb{N}} \frac{1}{n} Q_{0}(\Phi^{\otimes n}). 
 \end{align*}   
\end{definition}

\begin{remark}\label{remark:zero-super}
Suppose $(\mathcal{E}_1, \mathcal{D}_1)$ and $(\mathcal{E}_2, \mathcal{D}_2)$ are zero-error quantum codes (Definition~\ref{def:quantum-protocol}) for channels $\Phi_1$ and $\Phi_2$, respectively. Then, it is easy to show that $(\mathcal{E}_1\otimes \mathcal{E}_2, \mathcal{D}_1 \otimes \mathcal{D}_2)$ is a zero-error quantum code for $\Phi_1\otimes \Phi_2$. Hence, the one-shot zero-error quantum capacity is super-additive: $Q_0(\Phi_1 \otimes \Phi_2) \geq Q_0(\Phi_1)+Q_0(\Phi_2)$. A similar argument establishes the superadditivity for the other capacities. Consequently, the limits in Definition~\ref{def:zero-error-capacity} can be shown to exist by a simple application of Fekete's Lemma \cite{Fekete1923}, and are equal to the suprema of the corresponding sequences. 
\end{remark}

As before, the capacities satisfy the following relation:
\begin{equation}\label{eq:Q<P<Czero}
    Q_{\operatorname{zero}}(\Phi) \leq P_{\operatorname{zero}}(\Phi) \leq C_{\operatorname{zero}}(\Phi) \leq C^{\operatorname{ea}}_{\operatorname{zero}}(\Phi).
\end{equation}
Note that the inequalities in Eqs.~\eqref{eq:Q<P<C}, \eqref{eq:Q<P<Czero} between $Q,P,C$ can be (maximally) strict \cite{Leung2014privacy, Leung2016privacy}. We will study zero-error capacities in more detail in Section~\ref{sec:zero-error}.

We collect some bottleneck inequalities between the channel capacities below.

\begin{lemma}\label{lemma:bottleneck-regular}
    Let $\Psi_{A\to B}$, $\Phi_{B\to C}$ be quantum channels. Then, for $\mathbb{Q}\in \{C,P,C^{\operatorname{ea}}, Q \}$,
    \begin{align}
        \mathbb{Q}(\Phi\circ \Psi) &\leq \min ( \mathbb{Q}(\Phi), \mathbb{Q}(\Psi) ) \\ 
        \mathbb{Q}^{\dagger}(\Phi\circ \Psi) &\leq \min ( \mathbb{Q}^{\dagger}(\Phi), \mathbb{Q}^{\dagger}(\Psi) ) \\ 
        \mathbb{Q}_{\operatorname{zero}}(\Phi\circ \Psi) &\leq \min ( \mathbb{Q}_{\operatorname{zero}}(\Phi), \mathbb{Q}_{\operatorname{zero}}(\Psi) ).
    \end{align}
\end{lemma}
\begin{proof}
    The proof follows easily from the one-shot bottleneck inequalities (Lemma~\ref{lemma:bottleneck-oneshot}).
\end{proof}

Before closing this section, we note that if Alice and Bob are allowed to classically communicate with each other, then they can potentially boost the optimal rates of information transmission via the noisy channel that they share. We discuss the definitions of channel capacities assisted by classical communication in Appendix~\ref{appen:assisted}.

\subsection{Entropic quantities}
In this section, we define some divergences and entropies that are used to characterize capacities of quantum channels. For a more exhaustive account of these functions, we refer the readers to the excellent books \cite[Chapters 3-4]{Tomamichel2016book} and \cite[Chapter 7]{khatri2024principles}. 

\begin{definition}\label{def:divergence}
    Let $\rho\in \State{\Hil_A}$ be a state and $\sigma\in \B{\Hil_A}$ be a positive semi-definite operator.
    \begin{itemize}
        \item The \emph{(Umegaki) relative entropy} between $\rho$ and $\sigma$ is defined as \cite{Umegaki1962relative}
\begin{align*}
    D(\rho \Vert \sigma) := \begin{cases}
        \Tr \rho (\log \rho -\log \sigma) \quad &\text{if } \supp \rho \subseteq \supp \sigma \\
        + \infty  &\text{otherwise}
    \end{cases}
\end{align*}
        \item The $\alpha$-\emph{Petz R\'enyi relative entropy} between $\rho$ and $\sigma$ with $\alpha\in (0,1)\cup (1,\infty)$ is defined as \cite{Petz1985relative, Petz1986relative}
\begin{align*}
    D_{\alpha}(\rho \Vert \sigma) := \begin{cases}
        \frac{1}{\alpha-1} \log \Tr \left[ \rho^{\alpha} \sigma^{1-\alpha}  \right] \quad &\begin{cases}
            \text{if } \alpha\in (0,1) \text{ and } \rho\sigma\neq 0 \\
            \text{or } \alpha\in (1,\infty) \text{ and } \supp \rho \subseteq \supp \sigma
        \end{cases}  \\
        + \infty &\text{otherwise}
    \end{cases}
\end{align*}
        \item The $\alpha$-\emph{sandwiched R\'enyi relative entropy} between $\rho$ and $\sigma$ with $\alpha\in (0,1)\cup (1,\infty)$ is defined as \cite{MullerLennert2013sandwich, Wilde2014converse}
\begin{align*}
    \widetilde{D}_{\alpha}(\rho \Vert \sigma) := \begin{cases}
        \frac{1}{\alpha-1} \log \Tr \left[ \left(\sigma^{\frac{1-\alpha}{2\alpha}} \rho \sigma^{\frac{1-\alpha}{2\alpha}} \right)^{\alpha}  \right] \quad &\begin{cases}
            \text{if } \alpha\in (0,1) \text{ and } \rho\sigma \neq 0 \\
            \text{or } \alpha\in (1,\infty) \text{ and } \supp \rho \subseteq \supp \sigma
        \end{cases} \\
        + \infty &\text{otherwise}
    \end{cases}
\end{align*}
    \item The \emph{max-relative entropy} between $\rho$ and $\sigma$ is defined as \cite{Datta2009max, Renner2006PhD}
\begin{align*}
    D_{\max}(\rho \Vert \sigma) := \begin{cases}
        \log \norm{\sigma^{-1/2} \rho \sigma^{-1/2} }_{\infty} \quad &\text{if } \supp \rho \subseteq \supp \sigma \\
        + \infty &\text{otherwise}
    \end{cases}
\end{align*}
    \item The \emph{min-relative entropy} between $\rho$ and $\sigma$ is defined as \cite{Datta2009max}
\begin{align*}
    D_{\min}(\rho \Vert \sigma) := \begin{cases}
        -\log  \Tr (\Pi_{\rho} \sigma) \quad &\text{if } \rho\sigma\neq0 \\
        + \infty &\text{otherwise},
    \end{cases}
\end{align*}
    where $\Pi_{\rho}$ denotes the orthogonal projection onto support of $\rho$.
    \item The $\epsilon-$\emph{hypothesis testing relative entropy} between $\rho$ and $\sigma$ with $\epsilon\in [0,1]$ is defined as \cite{Buscemi2010hypothesis, WRenner2012hypo}
    \end{itemize}
    \begin{equation*}
       D^{\epsilon}_H (\rho \Vert \sigma) := -\log \inf \{ \Tr (\Lambda \sigma  ) : 0\leq \Lambda \leq \iden, \, \Tr (\Lambda \rho) \geq 1-\epsilon \}
    \end{equation*}
\end{definition}

\begin{remark}\label{remark:Dmax-quasi}
    The max-relative entropy can be alternatively expressed as \cite{Datta2009max}
    \begin{equation}
        D_{\max}(\rho \Vert \sigma) = \log \inf \{\lambda : \rho \leq \lambda \sigma \},
    \end{equation}
    where the infimum over an empty set is assumed to be $+\infty$. Moreover, it is quasi-convex: for a probability distribution $\{p_i \}_i$, states $\{\rho_i \}_i$, and positive operators $\{\sigma_i \}_i$, 
    \begin{equation}
        D_{\max} \left( \sum_i p_i \rho_i \Vert \sum_i p_i \sigma_i \right) \leq \max_i D_{\max}(\rho_i \Vert \sigma_i),
    \end{equation}
    where equality holds if $\operatorname{supp}\rho_i, \operatorname{supp}\sigma_i \subseteq \Hil_i$ and $\Hil_i \perp \Hil_j$ for $i\neq j$.
\end{remark}

The \emph{data-processing inequality} (DPI) is the defining property of these divergences. 

\begin{lemma} \label{lemma:DPI} (Data-processing inequality) \,
     Let $\rho\in \State{\Hil_A}$ be a state, $\sigma\in \B{\Hil_A}$ be a positive semi-definite operator, and $\Phi:\B{\Hil_A}\to \B{\Hil_B}$ be a quantum channel. Then, for $\mathbb{D}\in \{D,D_{\alpha}, \widetilde{D}_{\alpha},D_{\max},D_{\min},D_H^{\epsilon}\}$, 
    \begin{align*}
        \mathbb{D}(\Phi(\rho)\Vert \Phi(\sigma)) &\leq \mathbb{D}(\rho \Vert \sigma).
    \end{align*}
    For $\mathbb{D}=D_{\alpha}$, this holds for $\alpha\in (0,1)\cup(1,2]$ and for $\mathbb{D}=\widetilde{D}_{\alpha}$, this holds for $\alpha\in [1/2, 1) \cup (1,\infty)$.
\end{lemma}
\begin{proof}
    See \cite[Chapter 3]{Tomamichel2016book} or \cite[Chapter 7]{khatri2024principles} and references therein.
\end{proof}

Using these divergences as parent quantities, we now introduce several information measures for states and channels. We refer the readers to \cite{khatri2024principles} for a more coherent discussion of these quantities.

\begin{definition}\label{def:state-measures} (State measures)
    Let $\rho_{AB}\in \State{\Hil_A\otimes \Hil_B}$ be a bipartite state. We define the
\begin{itemize}
    \item \emph{mutual}, $\alpha$-\emph{mutual}, \emph{max-mutual}, and $\epsilon$-\emph{hypothesis testing mutual information} of $\rho_{AB}$ as 
    \begin{align*}
        I(A:B)_{\rho} &:= \inf_{\sigma_B} D(\rho_{AB} \Vert \rho_A \otimes \sigma_B ), \\
        \widetilde{I}_{\alpha}(A:B)_{\rho} &:= \inf_{\sigma_B} \widetilde{D}_{\alpha}(\rho_{AB} \Vert \rho_A \otimes \sigma_B ), \\
        I_{\max}(A:B)_{\rho} &:= \inf_{\sigma_B} D_{\max}(\rho_{AB} \Vert \rho_A \otimes \sigma_B ), \\
        I^{\epsilon}_H(A:B)_{\rho} &:= \inf_{\sigma_B} D^{\epsilon}_H(\rho_{AB} \Vert \rho_A \otimes \sigma_B ),
    \end{align*}
    respectively, where the infimum is over all states $\sigma_B\in \State{\Hil_B}$.
    \item \emph{coherent}, $\alpha$-coherent, \emph{max-coherent}, and $\epsilon$-\emph{hypothesis testing coherent information} of $\rho_{AB}$ as 
    \begin{align*}
        I(A\rangle B)_{\rho} &:= \inf_{\sigma_B} D(\rho_{AB}\Vert \iden_A \otimes \sigma_B), \\
        \widetilde{I}_{\alpha}(A\rangle B)_{\rho} &:= \inf_{\sigma_B} \widetilde{D}_{\alpha}(\rho_{AB}\Vert \iden_A \otimes \sigma_B), \\
        I_{\max}(A\rangle B)_{\rho} &:= \inf_{\sigma_B} D_{\max}(\rho_{AB}\Vert \iden_A \otimes \sigma_B) \\
        I^{\epsilon}_H(A\rangle B)_{\rho} &:= \inf_{\sigma_B} D_H^{\epsilon}(\rho_{AB}\Vert \iden_A \otimes \sigma_B),
    \end{align*}
    respectively, where the infimum is over all states $\sigma_B\in \State{\Hil_B}$.
\end{itemize}
\end{definition}

\begin{definition}\label{def:channel-measures} (Channel measures)
    Let $\Phi:\B{\Hil_A}\to \B{\Hil_B}$ be a quantum channel. We define 
    \begin{itemize}
        \item \emph{mutual}, $\alpha$-\emph{mutual}, \emph{max-mutual}, and $\epsilon$-\emph{hypothesis testing mutual information} of $\Phi$ as 
        \begin{align*}
            I(\Phi) &:= \sup_{\rho_{RA}} I(R:B)_{\Phi_{A\to B}(\rho_{RA})} \\
            \widetilde{I}_{\alpha}(\Phi) &:= \sup_{\rho_{RA}} \widetilde{I}_{\alpha}(R:B)_{\Phi_{A\to B}(\rho_{RA})}, \\
            I_{\max}(\Phi) &:= \sup_{\rho_{RA}} I_{\max}(R:B)_{\Phi_{A\to B}(\rho_{RA})}, \\
            I_H^{\epsilon}(\Phi) &:= \sup_{\rho_{RA}} I_H^{\epsilon}(R:B)_{\Phi_{A\to B}(\rho_{RA})},
        \end{align*}
        where the optimization is over all states $\rho_{RA}\in \State{\Hil_R \otimes \Hil_A}$ and $d_R\in \mathbb{N}$.
        \item \emph{Holevo}, $\alpha$-\emph{Holevo}, \emph{max-Holevo}, and $\epsilon$-\emph{hypothesis testing Holevo information} of $\Phi$ as
        \begin{align*}
            \chi(\Phi) &:= \sup_{\rho_{MA}} I(M:B)_{\Phi_{A\to B}(\rho_{MA})} \\
            \widetilde{\chi}_{\alpha}(\Phi) &:= \sup_{\rho_{MA}} \widetilde{I}_{\alpha}(M:B)_{\Phi_{A\to B}(\rho_{MA})}, \\
            \chi_{\max}(\Phi) &:= \sup_{\rho_{MA}} I_{\max}(M:B)_{\Phi_{A\to B}(\rho_{MA})}, \\
            \chi_H^{\epsilon}(\Phi) &:= \sup_{\rho_{MA}} I_H^{\epsilon}(M:B)_{\Phi_{A\to B}(\rho_{MA})},
        \end{align*}
        where the optimization is over all cq states $\rho_{MA}$ and $d_M\in \mathbb{N}$.
        \item \emph{coherent}, $\alpha$-\emph{coherent}, \emph{max-coherent}, and $\epsilon$-\emph{hypothesis testing coherent information} of $\Phi$ as 
        \begin{align*}
            I_c(\Phi) &:= \sup_{\rho_{RA}} I(R\rangle B)_{\Phi_{A\to B}(\rho_{RA})} \\
            \widetilde{I}^c_{\alpha}(\Phi) &:= \sup_{\rho_{RA}} \widetilde{I}_{\alpha}(R \rangle B)_{\Phi_{A\to B}(\rho_{RA})}, \\
            I^c_{\max}(\Phi) &:= \sup_{\rho_{RA}} I_{\max}(R \rangle B)_{\Phi_{A\to B}(\rho_{RA})}, \\
            I^{c,\epsilon}_{H}(\Phi) &:= \sup_{\rho_{RA}} I_H^{\epsilon}(R \rangle B)_{\Phi_{A\to B}(\rho_{RA})},
        \end{align*}
        where the optimization is over all states $\rho_{RA}\in \State{\Hil_R \otimes \Hil_A}$ and $d_R\in \mathbb{N}$.
        \item \emph{private information} of $\Phi$ as
        \begin{align*}
            I_p(\Phi) &:= \sup_{\rho_{MA}} \left(I(M : B)_{\Phi_{A\to B}(\rho_{MA})} - I(M:E)_{\Phi^c_{A\to E}(\rho_{MA})} \right)
        \end{align*}
        where the optimization is over all cq states $\rho_{MA}$ and $d_M\in \mathbb{N}$.
    \end{itemize}
\end{definition}

\begin{remark}\label{remark:pure}
    The optimizations in Definition~\ref{def:channel-measures} over states $\rho_{RA}\in \State{\Hil_R \otimes \Hil_A}$ with arbitrary dimension $d_R\in \mathbb{N}$ of the reference system can be restricted to just pure states $\psi_{RA}\in \State{\Hil_R\otimes \Hil_A}$, where the reference system has the same dimension as $A$: $d_R=d_A$. The argument is a standard one, and uses purification and Schmidt decomposition, along with data-processing of the underlying divergences (Lemma~\ref{lemma:DPI}), see e.g., \cite{khatri2024principles}.
\end{remark} 

\subsubsection{Entropic bounds on channel capacities}

By connecting the task of channel coding with 
hypothesis testing, one can establish converse bounds on the one-shot channel capacities in terms of the various correlation measures introduced above \cite{Polyanskiy2010meta, Polyanskiy2010converse, Sharma2013meta}. We note these bounds in the following lemma.

\begin{lemma}\label{lemma:one-shot-converse}
    Let $\Phi:\B{\Hil_A}\to \B{\Hil_B}$ be a quantum channel, $\epsilon\in [0,1)$, and $\alpha>1$. Then,
    \begin{alignat*}{2}
        C_{\epsilon}(\Phi) &\leq \chi^{\epsilon}_H(\Phi) &&\leq \widetilde{\chi}_{\alpha}(\Phi) + \frac{\alpha}{\alpha-1} \log (\frac{1}{1-\epsilon}), \\
        C^{\operatorname{ea}}_{\epsilon}(\Phi) &\leq I^{\epsilon}_H(\Phi) &&\leq \widetilde{I}_{\alpha}(\Phi) + \frac{\alpha}{\alpha-1} \log (\frac{1}{1-\epsilon}), \\
        Q_{\epsilon}(\Phi) &\leq I^{c,\epsilon}_{H}(\Phi) &&\leq \widetilde{I}^c_{\alpha}(\Phi) + \frac{\alpha}{\alpha-1} \log (\frac{1}{1-\epsilon}).
    \end{alignat*}
    Similarly, 
    \begin{alignat*}{2}
        C_{\epsilon}(\Phi) &\leq \chi_H^{\epsilon}(\Phi) &&\leq \chi_{\max}(\Phi) + \log (\frac{1}{1-\epsilon}), \\
        C^{\operatorname{ea}}_{\epsilon}(\Phi) &\leq I_H^{\epsilon}(\Phi) &&\leq I_{\max}(\Phi) + \log (\frac{1}{1-\epsilon}), \\
        Q_{\epsilon}(\Phi) &\leq I_H^{c,\epsilon}(\Phi) &&\leq I^c_{\max}(\Phi) + \log (\frac{1}{1-\epsilon}).
    \end{alignat*}
\end{lemma}
\begin{proof}
    The bounds on classical and entanglement-assisted classical capacity are from \cite{Matthews2014converse} (see also \cite{Wilde2014converse}). Simpler proofs for the entanglement-assisted capacity bound are given in \cite{Anshu2019oneshot, Qi2018position}. We refer the readers to \cite{khatri2024principles} for a unified derivation of all the bounds.
\end{proof}

It turns out that the one-shot upper bounds, given above, become tight in the limit of the number of channel uses $n\to \infty$, which allows us to provide regularized entropic formulas for the asymptotic channel capacities from Definition~\ref{def:capacity}. 

\begin{theorem}\label{theorem:Cea-I} \cite{Bennett1999entanglement, Bennett2002entanglement, Holevo2002entanglement}
    Let $\Phi:\B{\Hil_A}\to \B{\Hil_B}$ be a quantum channel. Then,
    \begin{equation*}
        C_{\operatorname{ea}}(\Phi) = C^{\dagger}_{\operatorname{ea}}(\Phi) = I(\Phi).
    \end{equation*}
\end{theorem}

\begin{theorem}\label{theorem:LSD+CP} 
    For a quantum channel $\Phi:\B{\Hil_A}\to \B{\Hil_B}$, the following is true:
    \begin{align*}
      \text{\cite{Holevo1998capacity, Schumacher1997capacity}} \quad  C(\Phi) &= \lim_{n\to \infty} \frac{1}{n} \chi (\Phi^{\otimes n}) = \sup_{n\in \mathbb{N}} \frac{1}{n} \chi (\Phi^{\otimes n}), \\
      \text{\cite{Cai2004private, Devetak2005capacity} } \,\,\,\,  P(\Phi) &= \lim_{n\to \infty} \frac{1}{n} I_p (\Phi^{\otimes n}) = \sup_{n\in \mathbb{N}} \frac{1}{n} I_p (\Phi^{\otimes n}), \\
      \text{\cite{Lloyd1997capacity, Shor2002capacity, Devetak2005capacity}} \quad  Q(\Phi) &= \lim_{n\to \infty} \frac{1}{n} I_c (\Phi^{\otimes n}) = \sup_{n\in \mathbb{N}} \frac{1}{n} I_c (\Phi^{\otimes n}).
    \end{align*}
\end{theorem}

The information measure $\chi$ is super-additive: $\chi(\Phi \otimes \Psi)\geq \chi(\Phi)+\chi(\Psi)$ for all channels $\Phi$ and $\Psi$, and the same is true for $I_p$ and $I_c$. Moreover, the inequality here can be strict \cite{smolin1998noisy, Hastings2009, Smith2009superadd}. Hence, apart from special channels for which these information measures are additive (such as for Hadamard channels \cite{Winter2016potential}), the above capacity expressions become intractable because of regularization. In fact, the capacities are not even known to be computable in general \cite{wolf2024decidable}. Moreover, no such expressions are known for the strong converse capacities. 

In the following, we collect some simple upper bounds on the strong converse capacities in terms of entanglement measures. Let us begin by introducing the relevant definitions.

\begin{definition} \label{def:entanglement-measures} (Entanglement measures) \,
    Let $\rho_{AB}$ be a bipartite state. We define the \emph{relative}, $\alpha$-\emph{relative}, \emph{max-relative}, and $\epsilon$-\emph{hypothesis testing relative entropy of entanglement} of $\rho_{AB}$ as 
\begin{align*}\label{eq:EHstate}
    E(A:B)_{\rho} &:= \inf_{\sigma_{AB}\in \textrm{SEP}(A:B)} D(\rho_{AB}\Vert \sigma_{AB}) \\  
    \widetilde{E}_{\alpha}(A:B)_{\rho} &:= \inf_{\sigma_{AB}\in \textrm{SEP}(A:B)} \widetilde{D}_{\alpha}(\rho_{AB}\Vert \sigma_{AB}), \\
    E_{\max}(A:B)_{\rho} &:= \inf_{\sigma_{AB}\in \textrm{SEP}(A:B)} D_{\max}(\rho_{AB}\Vert \sigma_{AB}), \\
    E^{\epsilon}_H(A:B)_{\rho} &:= \inf_{\sigma_{AB}\in \textrm{SEP}(A:B)} D^{\epsilon}_H(\rho_{AB}\Vert \sigma_{AB}),
\end{align*}
respectively, where the optimization is over the set of separable states $\textrm{SEP}(A:B)$.
\end{definition}

The defining property of any entanglement measure is that it is non-increasing under \emph{local operations and classical communication} (LOCC), which encapsulates the requirement that starting from an arbitrary state $\rho_{AB}$, Alice and Bob should not be able to increase the amount of entanglement they share by only doing local operations and communicating classically. It is easy to use data-processing of the underlying divergences (Lemma~\ref{lemma:DPI}) to show that the measures in Definition~\ref{def:entanglement-measures} satisfy this property: for any LOCC channel $\mathcal{L}_{AB\to A'B'}$, we have $E(A':B')_{\mathcal{L}(\rho)}\leq E(A:B)_{\rho}$.

For any quantum channel $\Phi_{A\to B}$, we can use the entanglement measures from Definition~\ref{def:entanglement-measures} to quantify the channel's ability to preserve entanglement with an arbitrary reference $R$ as follows.

\begin{definition}\label{def:channel-entanglement-measures}
    Let $\Phi_{A\to B}$ be a quantum channel. We define the \emph{relative}, $\alpha$-\emph{relative}, \emph{max-relative}, and $\epsilon$-\emph{hypothesis testing relative entropy of entanglement} of $\Phi$ as 
        \begin{align*}
            E(\Phi ) &:= \sup_{\rho_{RA}} E(R:B)_{\Phi_{A\to B}(\rho_{RA})} \\
            \widetilde{E}_{\alpha}(\Phi ) &:= \sup_{\rho_{RA}} \widetilde{E}_{\alpha}(R:B)_{\Phi_{A\to B}(\rho_{RA})}, \\
            E_{\max}(\Phi ) &:= \sup_{\rho_{RA}} E_{\max}(R:B)_{\Phi_{A\to B}(\rho_{RA})}, \\
            E_H^{\epsilon}(\Phi ) &:= \sup_{\rho_{RA}} E_{H}^{\epsilon}(R:B)_{\Phi_{A\to B}(\rho_{RA})},
        \end{align*}
        where the optimization is over all states $\rho_{RA}\in \State{\Hil_R \otimes \Hil_A}$ and $d_R\in \mathbb{N}$. It suffices to restrict the optimization to pure states $\psi_{RA}$ with $d_R=d_A$ (see Remark~\ref{remark:pure}).
\end{definition}

With the relevant definitions in place, we are ready to state the promised converse bounds.

\begin{lemma}\label{lemma:QP<=E}
Let $\Phi:\B{\Hil_A}\to \B{\Hil_B}$ be a quantum channel, $\epsilon\in [0,1)$, and $\alpha>1$. Then,
    \begin{align}
        Q_{\epsilon}(\Phi) &\leq E^{\epsilon}_H(\Phi) \leq \widetilde{E}_{\alpha}(\Phi) + \frac{\alpha}{\alpha-1} \log (\frac{1}{1-\epsilon}), \\
        C^p_{\epsilon}(\Phi) &\leq E^{\epsilon}_H(\Phi) \leq \widetilde{E}_{\alpha}(\Phi) + \frac{\alpha}{\alpha-1} \log (\frac{1}{1-\epsilon}).
    \end{align}
\end{lemma}

\begin{lemma}\label{lemma:strong-converse}
    Let $\Phi:\B{\Hil_A}\to \B{\Hil_B}$ be a quantum channel. Then, for all $\alpha>1$
    \begin{align}
        Q^{\dagger}(\Phi) \leq P^{\dagger}(\Phi) &\leq \widetilde{E}_{\alpha} (\Phi) \leq E_{\max}(\Phi) \\ 
        C^{\dagger}(\Phi) &\leq \lim_{n\to \infty}  \frac{1}{n} \widetilde{\chi}_{\alpha} (\Phi^{\otimes n}) = \sup_{n\in \mathbb{N}} \frac{1}{n} \widetilde{\chi}_{\alpha} (\Phi^{\otimes n}).
    \end{align}
\end{lemma}
\begin{proof}
    The relative entropy of entanglement upper bound on the strong converse private capacity was proven in \cite{Wilde2017private}. The max-relative entropy of entanglement is known to be an upper bound even on the private classical capacity assisted with two-way classical communication \cite{Christandl2017max}. 
    For the upper bound on classical capacity, note that for all $\alpha>1$, $\eps\in[0,1)$ and $n\in\mathbb{N}$ (Lemma~\ref{lemma:one-shot-converse}):
    \begin{equation}
       \frac{1}{n} C_{\epsilon}(\Phi^{\otimes n}) \leq \frac{1}{n} \widetilde{\chi}_{\alpha}(\Phi^{\otimes n}) + \frac{\alpha}{n(\alpha-1)}\log(\frac{1}{1-\eps}),
    \end{equation}
    which proves the desired bound by taking approprite limits (see Definition~\ref{def:capacity}). Since $\widetilde{\chi}_{\alpha}$ is superadditive: $\widetilde{\chi}_{\alpha}(\Phi \otimes \Psi)\geq \widetilde{\chi}_{\alpha}(\Phi)+\widetilde{\chi}_{\alpha}(\Psi)$ \cite{Beigi2013sandwich}, the following limit exists \cite{Fekete1923}:
    \begin{equation}
        \lim_{n\to \infty}  \frac{1}{n} \widetilde{\chi}_{\alpha} (\Phi^{\otimes n}) = \sup_{n\in \mathbb{N}} \frac{1}{n} \widetilde{\chi}_{\alpha} (\Phi^{\otimes n}).
    \end{equation}
\end{proof}

Finally, we note bottleneck inequalities for all the channel measures that we will employ later.

\begin{lemma}\label{lemma:channel-bottlenecks}
    Let $\Psi:\B{\Hil_A}\to \B{\Hil_B}$ and $\Phi:\B{\Hil_B}\to \B{\Hil_C}$ be quantum channels. Then, for all channel measures $\mathbb{I}$ from Definitions~\ref{def:channel-measures} and entanglement measures $\mathbb{E}$ from Definition~\ref{def:channel-entanglement-measures}, 
    \begin{align}
        \mathbb{I}(\Phi\circ\Psi) &\leq \min (\mathbb{I}(\Phi),\mathbb{I}(\Psi)), \\
        \mathbb{E}(\Phi\circ \Psi) &\leq \min(\mathbb{E}(\Phi),\mathbb{E}(\Psi)).
    \end{align}
\end{lemma}
\begin{proof}
    The claims follow from data-processing of the underlying divergences (Lemma~\ref{lemma:DPI}). 
\end{proof}

\subsection{Zero-error communication}\label{sec:zero-error}
The constraint of perfect, i.e., error-free, communication gives the theory of zero-error communication a much more algebraic/combinatorial flavor \cite{Shannon1956zero, Korner1998zero, Duan2013noncomm}. In this section, we give a short background on the basics of this theory. The primary object of interest here is the so-called non-commutative (confusability) graph of a quantum channel \cite{Duan2013noncomm}, which is a non-commutative generalization of the confusability graph of classical stochastic channels \cite{Shannon1956zero}. Recently, there has also been a growing interest in the theory of non-commutative graphs independently of its connection with zero-error information theory. We refer the interested readers to the review article \cite{Daws2024qgraph} for further details.

\begin{definition}\label{def: op-sys}
    Let $\Phi: \B{\Hil_A}\to \B{\Hil_B}$ have a Kraus representation $\Phi(X)=\sum_{i=1}^n K_i XK_i^{\dagger}$. The operator system (also called the \emph{non-commutative (confusability) graph}) of $\Phi$ is defined as
\begin{equation*}
    S_{\Phi} := {\rm{span}} \{K^{\dagger}_i K_j : \, 1\leq i,j \leq n\} \subseteq \B{\Hil_A}. 
\end{equation*}
\end{definition}

It is easy to check that the above definition is independent of the chosen Kraus representation of $\Phi$. Moreover, $\sum_{i=1}^n K_i^{\dagger} K_i={\iden}_A \in S_{\Phi}$ (since $\Phi$ is trace-preserving) and $X\in S_{\Phi} \implies X^{\dagger} \in S_{\Phi}$. Such $\dagger-$closed subspaces $S\subseteq \B{\Hil_A}$ containing the identity are called \emph{operator systems} \cite{paulsen-book}. Moreover, any such operator system $S$ arises as the non-commutative graph of some channel $\Phi$ \cite{Duan2009zerosuper}. One can check that if $\Phi_c:\B{\Hil_A}\to \B{\Hil_E}$ is complementary to $\Phi$, then the operator system is obtained as the image of the environment algebra under $(\Phi_c)^*$ \cite{Duan2013noncomm}:
\begin{equation}
    S_{\Phi} = (\Phi_c)^*(\B{\Hil_E}) := \{(\Phi_c)^* (X) : X\in\B{\Hil_E}\}.
\end{equation}

\begin{remark}
    It is easy to see that the operator systems are multiplicative, i.e., for two quantum channels $\Phi:\B{\Hil_A}\to \B{\Hil_B}$ and $\Psi:\B{\Hil_C}\to \B{\Hil_D}$, we have $S_{\Phi\otimes \Psi} = S_{\Phi}\otimes S_{\Psi} \subseteq \B{\Hil_A\otimes \Hil_C}$.
\end{remark}

The following parameters were introduced as the non-commutative generalizations of classical graph parameters (such as the independence number of a graph) in \cite{Duan2013noncomm}. Below, orthogonality between operators is with repsect to the Hilbert Schmidt inner product on $\B{\Hil}$, i.e., we write $X\perp Y$ if $\Tr (X^{\dagger}Y)=0$. Moreover, for an operator $X\in \B{\Hil}$ and subspace $S\subseteq \B{\Hil}$, $X\perp S$ means that $X\perp Y$ for all $Y\in S$.

\begin{definition}\cite{Duan2013noncomm} \label{def:op-parameters}
    For an operator system $S\subseteq \B{\cH}$,
\begin{itemize}
    \item the maximum number $\mathscr{M}$ such that there exist states $\{\rho_m \}_{m\in [\mathscr{M}]} \subseteq \State{\Hil}$ such that 
\begin{equation}
    \forall m\neq m': \forall \ket{\psi}\in \operatorname{supp}\rho_m, \forall \ket{\phi}\in \operatorname{supp}\rho_{m'}: \quad |\psi\rangle \langle\phi | \perp S
\end{equation}
is called the \emph{independence number} of $S$ (denoted as $\alpha(S)$) . 
\item the maximum number $\mathscr{M}$ such that there exist states $\{\rho_m \}_{m\in [\mathscr{M}]} \subseteq \State{\Hil}$ such that 
\begin{align}
    \forall m\neq m': \forall \ket{\psi}\in \operatorname{supp}\rho_m, \forall \ket{\phi}\in \operatorname{supp}\rho_{m'}: \quad |\psi\rangle \langle\phi | \perp S \,\,\text{and} \,\, (\rho_m-\rho_{m'}) \perp S
\end{align}
is called the \emph{private independence number} of $S$ (denoted as $\alpha_p(S)$) . 
\item the maximum number $\mathscr{M}$ such that there exist Hilbert spaces $\cH_{A_0}, \cH_R$, a state $\rho\in \State{\Hil_{A_0}}$, and isometries $\{V_m: \Hil_{A_0} \to \Hil\otimes \Hil_R\}_{m\in [\mathscr{M}]}$ such that 
\begin{equation}
\forall m\neq m': \quad V_m \rho V_{m'} \perp S\otimes \B{\Hil_R},
\end{equation}
is called the \emph{entanglement-assisted independence number} of $S$ (denoted ${\alpha}_{ea}(S)$).  
\item the maximum number $d$ such that there exists a subspace $\cC\subseteq \Hil$ with $\dim \cC=d$ satisfying $P_\cC S P_\cC = \mathbb{C} P_\cC$,
(where $P_\cC$ is the orthogonal projection onto $\mathcal{C}$)\footnote{This is exactly the Knill-Laflamme error-correction condition \cite{knil-laf} for the subspace $\mathcal{C}$.}
is called the \emph{quantum independence number} of $S$ (denoted as $\alpha_q(S)$). 
\end{itemize}
\end{definition}

Exactly as in classical zero-error information theory \cite{Shannon1956zero}, the above graph parameters are closely linked to the one-shot zero-error capacities of the corresponding noisy channel.

\begin{theorem}\cite{Duan2013noncomm}\label{thm:DSW}
    For a quantum channel $\Phi: \B{\cH_A}\to \B{\cH_B}$,
\begin{align*}
    C_0 (\Phi) &= \log \alpha (S_{\Phi}), \\ 
    C_0^{\operatorname{p}}(\Phi) &= \log \alpha_p (S_{\Phi}), \\
    C_0^{\operatorname{ea}}(\Phi) &= \log \alpha_{ea} (S_{\Phi}), \\
    Q_0 (\Phi) &= \log \alpha_q (S_{\Phi}).
\end{align*}
\end{theorem}

The notions of pre- and post-processing by quantum channels are captured by homomorphisms and inclusions in the language of operator systems, as we note below.

\begin{definition} \cite{Stahlke2016zero} \label{def:op-homo}
    Let $S\subseteq \B{\Hil_A}$ and $T\subseteq\B{\Hil_B}$ be operator systems. We say that $S$ is \emph{homomorphic} to $T$ (denoted as $S\longrightarrow T$) if there exists an isometry $V:\Hil_A\to \Hil_B\otimes \Hil_E$ such that 
    \begin{equation*}
         V^{\dagger} (T \otimes \B{\Hil_E} ) V \subseteq S.
    \end{equation*}
\end{definition}

Let us note some basic properties of graph homomorphisms and inclusions below. 

\begin{lemma} \cite{Stahlke2016zero} \label{lemma:op-homo-2}
    Let $Q,R, S$ and $T$ be operator systems.
    \begin{itemize}
        \item If $R\longrightarrow S$ and $S\longrightarrow T$, then, $R\longrightarrow T$. 
        \item If $Q\longrightarrow R$ and $S\longrightarrow T$, then $Q \otimes S \longrightarrow R \otimes T$.
    \end{itemize}
    
\end{lemma}

\begin{lemma}\label{lemma:op-homo}
    Let $\Phi:\B{\Hil_A}\to \B{\Hil_B}$ and $\Psi:\B{\Hil_B}\to \B{\Hil_C}$ be quantum channels. Then,
    \begin{equation*}
        S_{\Phi} \subseteq S_{\Psi\circ \Phi} \longrightarrow S_{\Psi}.
    \end{equation*}
\end{lemma}
\begin{proof}
    Choose Kraus representations $\Phi(X)=\sum_{i=1}^n K_i XK_i^{\dagger}$ and $\Psi(Y)=\sum_{j=1}^m F_j YF_j^{\dagger}$. Then, 
    \begin{equation}
        S_{\Psi\circ \Phi} = \operatorname{span}\{K_i^{\dagger}F_j^{\dagger} F_q K_p : 1\leq i,p\leq n, 1\leq j,q\leq m \}.
    \end{equation}
    Clearly, for all $1\leq i,p\leq n$, we have $K_i^{\dagger}K_p = \sum_j K_i^{\dagger}F_j^{\dagger}F_j K_p \in S_{\Psi\circ \Phi}$, since $\sum_j F^{\dagger}_j F_j = \iden_B$. Hence, $S_{\Phi} = \operatorname{span}\{K^{\dagger}_i K_p : 1\leq i,p \leq n \}\subseteq S_{\Psi\circ \Phi}$. 

    To prove the second claim, let $V:\Hil_A\to \Hil_B\otimes \Hil_E$ be a Stinespring isometry for $\Phi:\B{\Hil_A}\to \B{\Hil_B}$. Then, it is easy to check that (see for e.g. \cite{Duan2013noncomm})
    \begin{equation}
         V^{\dagger} (S_{\Psi}\otimes \B{\Hil_E}) V =S_{\Psi\circ \Phi}.
    \end{equation}
\end{proof}

\begin{lemma}\label{lemma:op-bottleneck}
    Let $S,T\subseteq \B{\Hil}$ be two operator systems such that $S\subseteq T$. Then,
    \begin{equation*}
        \alpha(T)\leq \alpha(S), \quad \alpha_p(T)\leq \alpha_p(S), \quad \alpha_q (T)\leq \alpha_q (S).
    \end{equation*}
    Similarly, let $S\subseteq \B{\Hil_A}$ and $T\subseteq\B{\Hil_B}$ be two operator systems such that $S\longrightarrow T$. Then,
    \begin{equation*}
        \alpha(S)\leq \alpha(T), \quad \alpha_p(S)\leq \alpha_p(S), \quad \alpha_q (S)\leq \alpha_q (T).
    \end{equation*} 
\end{lemma}
\begin{proof}
    If $S,T\subseteq \B{\Hil}$ are operator systems satisfying $S\subseteq T$, then $T^{\perp} \subseteq S^{\perp}$. Hence, it is clear from Definition~\ref{def:op-parameters} that the desired relations hold.

    Now, consider operator systems $S\subseteq \B{\Hil_{A}}$ and $T\subseteq\B{\Hil_{B}}$ such that $S\longrightarrow T$. Then, according to \cite{Duan2009zerosuper}, there exist quantum channels $\Phi_{A\to A'}$ and $\Psi_{B\to B'}$ such that $S=S_{\Phi}$ and $T=S_{\Psi}$. Moreover, by Definition~\ref{def:op-homo}, there exists a channel $\mathcal{N}:\B{\Hil_{A}}\to \B{\Hil_{B}}$ with Stinespring isometry $V:\Hil_{A}\to \Hil_{B}\otimes \Hil_E$ such that 
    \begin{equation}
      S_{\Psi\circ \mathcal{N}} = V^{\dagger}(S_{\Psi} \otimes \B{\Hil_E})V \subseteq S_{\Phi}.
    \end{equation}
    This yields the desired relation: 
    \begin{equation}
        \alpha(S_{\Phi}) \leq \alpha (S_{\Psi \circ \mathcal{N}}) \leq \alpha(S_{\Psi}),
    \end{equation}
    where we used the correspondence between independence numbers and one-shot channel capacities from Theorem~\ref{thm:DSW} along with the bottleneck inequality from Lemma~\ref{lemma:bottleneck-oneshot}.
\end{proof}

    In general, computing the independence numbers of operator systems -- or, equivalently, computing the one-shot zero-error capacities of quantum channels -- is difficult \cite{Shor2008complexity}. Moreover, the independence numbers are highly non-multiplicative \cite{Chen2010zerosuper}. However, In Section~\ref{sec:qms-transmission}, we will prove that if an operator system $S\subseteq \B{\Hil}$ also has the structure of an algebra (i.e., it is closed under matrix multiplication), then its independence numbers can be efficiently and explicitly computed and are multiplicative.

\subsection{Spectral properties}\label{sec:spectral}
Let $\Psi:\B{\Hil}\to \B{\Hil}$ be a quantum channel. Then, $\Psi$ admits a Jordan decomposition \cite[Chapter 6]{Wolf2012Qtour}
\begin{equation}
    \Psi = \sum_{i} \lambda_i \mathcal{P}_i + \mathcal{N}_i \quad \text{with} \quad \mathcal{N}_i \mathcal{P}_i = \mathcal{P}_i \mathcal{N}_i = \mathcal{N}_i \,\,\, \text{and} \,\,\, \mathcal{P}_i \mathcal{P}_j = \delta_{ij}\mathcal{P}_i,
\end{equation}
where the sum runs over the distinct eigenvalues $\lambda_i$ of $\Psi$, $\mathcal{P}_i$ are projectors whose rank equals the algebraic multiplicity of $\lambda_i$, and $\mathcal{N}_i$ denote the corresponding nilpotent operators. All the eigenvalues $\lambda_i$ of $\Psi$ satisfy $\abs{\lambda_i}\leq 1$ and they are either real or come in complex conjugate pairs. Since $\Psi$ always admits a fixed point, $\lambda=1$ is always an eigenvalue of $\Psi$. Moreover, all $\lambda_i$ with $\abs{\lambda_i}=1$ have equal algebraic and geometric multiplicities, so that $\mathcal{N}_i=0$ for all such eigenvalues. As $l\to \infty$, we expect the image of 
\begin{equation}
    \Psi^l := \underbrace{\Psi \circ \Psi \circ \ldots \circ \Psi}_{l\,  \text{times}}
\end{equation}
to converge to the \emph{peripheral space} $\mathscr{X} (\Psi):= \text{span}\{X\in \B{\Hil} : \exists \,\theta\in \mathbb{R} \text{ s.t. } \Psi(X)= e^{i\theta} X\}$. 
\begin{definition}\label{def:peripheral-proj}
   Let $\Psi:\B{\Hil}\to \B{\Hil}$ be a quantum channel. The \emph{asymptotic part} of $\Psi$ and the projector onto the peripheral space $\mathscr{X} (\Psi)$, are respectively defined as follows: 
   \begin{equation}\label{eq:phiinf-proj}
\Psi_{\infty}:= \sum_{i:\, |\lambda_i|=1}\lambda_i \mathcal{P}_i \quad \text{and} \quad  \mathcal{P}_{\Psi} = \sum_{i: \, |\lambda_i|=1} \mathcal{P}_i .
\end{equation}
\end{definition}

Clearly, $\Psi_{\infty}=\Psi_{\infty}\circ \mathcal{P}_{\Psi} = \mathcal{P}_{\Psi}\circ \Psi_{\infty}$.
Notably, both $\Psi_{\infty}:\B{\Hil}\to \B{\Hil}$ and $\mathcal{P}_{\Psi}:\B{\Hil}\to \B{\Hil}$ arise as limit points of the set $(\Psi^l )_{l\in \mathbb{N}}$ \cite[Lemma 3.1]{Szehr2014specconvergence}. Since the set of quantum channels acting on $\Hil$ is closed, both $\Psi_{\infty}$ and $\mathcal{P}_{\Psi}$ are quantum channels themselves. Crucially, the action of $\Psi$ on $\mathscr{X}(\Psi)$ is reversible in the sense of the following lemma.

\begin{lemma}\label{lemma:reverse} \cite{wolf2010inverse, Wolf2012Qtour}
    The action of a channel $\Psi:\B{\Hil}\to \B{\Hil}$ on its peripheral space $\mathscr{X}(\Psi)$ is reversible, i.e.,  there exists a channel $\mathcal{R}_{\Psi}:\B{\Hil}\to \B{\Hil}$ such that $\mathcal{R}_{\Psi}\circ\Psi = \mathcal{P}_{\Psi}=\Psi\circ\mathcal{R}_{\Psi}$.
\end{lemma}

As $l$ increases, $\norm{\Psi^l - \Psi^l_{\infty}}_{\diamond}$ approaches zero. More precisely, the convergence behavior is like 
\begin{equation}\label{eq:converge}
    \norm{\Psi^l - \Psi^l_{\infty}}_{\diamond} \leq \kappa \mu^l,
\end{equation}
where $\mu = \operatorname{spr}(\Psi-\Psi_{\infty})<1$ is the spectral radius of $\Psi-\Psi_{\infty}$ (i.e., $\mu$ is the largest magnitude of the eigenvalues of $\Psi-\Psi_{\infty}$) and $\kappa$ depends on the spectrum of $\Psi$, on $l$, and on the dimension $d=\dim\Hil$ \cite{Szehr2014specconvergence}. The dependence of $\kappa$ on $l$ is sub-exponential, which captures the fact that for large $l$, the RHS of~\eqref{eq:converge} exponentially decays as $\mu^l$. For example, by only using the spectral gap $\mu$, one can obtain a convergence estimate of the following form for $l>\mu/(1-\mu)$ \cite{Szehr2014specconvergence}:
\begin{equation}\label{eq:converge-gap}
    \norm{\Psi^l - \Psi^l_{\infty}}_{\diamond} \leq \frac{4e^2 d (d^2 + 1)}{ \left( 1- (1+ \frac{1}{l})\mu \right)^{3/2}}  \left(\frac{l (1-\mu^2)}{\mu}\right)^{d^2 -1} \mu^l,
\end{equation}
where $d=\dim \Hil$. A more complete knowledge about the Jordan decomposition of $\Psi$ can be used to sharpen the above estimate (see \cite{Szehr2014specconvergence}). In this paper, we will only work with the general spectral gap bound of the form in Eq.~\eqref{eq:converge-gap}.

\subsubsection{The peripheral space}\label{sec:periphery}
Recall the definition of the peripheral space of a channel $\Psi$:
\begin{equation}
    \mathscr{X} (\Psi):= \text{span}\{X\in \B{\Hil} : \exists \, \theta\in \mathbb{R} \text{ with } \Psi(X)= e^{i\theta} X\}.
\end{equation}
The structure of the peripheral space is well-understood. There exists an orthogonal decomposition of the underlying Hilbert space $\Hil = \Hil_{0} \oplus \bigoplus_{k=1}^K \Hil_{k,1}\otimes \Hil_{k,2}$, and positive definite states $\delta_{k}\in \State{\Hil_{k,2}}$ such that the following block decomposition holds \cite{Lindblad1999} \cite[Chapter 6]{Wolf2012Qtour}: 
    \begin{equation}\label{eq:phasespace}
       \mathscr{X} (\Psi) = 0 \oplus \bigoplus_{k=1}^K (\B{\Hil_{k,1}}\otimes \delta_k). 
    \end{equation}
Moreover, there exist unitaries $U_k\in \B{\Hil_{k,1}}$ and a permutation $\pi$ which permutes within subsets of $\{1,2,\ldots ,K \}$ for which the corresponding $\Hil_{k,1}$'s have the same dimension, such that for
    \begin{equation}\label{eq:phaseaction}
        X = 0 \oplus \bigoplus_{k=1}^K x_k \otimes \delta_k, \quad\text{we have}\quad \Psi (X) = \Psi_{\infty}(X) = 0 \oplus \bigoplus_{k=1}^K U^{\dagger}_k x_{\pi (k)} U_k \otimes \delta_k,
    \end{equation}
see \cite{wolf2010inverse} 
\cite[Chapter 6]{Wolf2012Qtour}. Let $\Hil = \Hil_0 \oplus \Hil_0^{\perp}$, where we have identified $\Hil_0^{\perp}=\bigoplus_{k=1}^K \Hil_{k,1}\otimes \Hil_{k,2}$, and let $V:\Hil_0^{\perp} \hookrightarrow \Hil $ be the canonical inclusion isometry. Recall that $\mathcal{P}_{\Psi}:\B{\Hil}\to \B{\Hil}$ defined in Eq.~\eqref{eq:phiinf-proj} projects onto the peripheral space $\mathscr{X} (\Psi)$. It is easy to see that 
\begin{align}
    \forall X\in \B{\Hil}: \quad \mathcal{P}_{\Psi}(X) &= 0 \oplus V^{\dagger} \mathcal{P}_{\Psi}(X) V  \nonumber \\
    &= 0 \oplus R_V (\mathcal{P}_{\Psi} (X)),
\end{align}
where the channel $R_V:\B{\Hil}\to \B{\Hil_0^{\perp}}$ is the restriction channel defined as $R_V (Y) = V^{\dagger} Y V + \Tr [(\iden - VV^{\dagger})Y ]\sigma $ for some state $\sigma\in \State{\Hil_0^{\perp}}$. Moreover, since $\mathcal{P}_{\Psi}=\mathcal{P}_{\Psi}^2$, we get
\begin{align}
  \forall X\in \B{\Hil}: \quad   \mathcal{P}_{\Psi}(X) = \mathcal{P}_{\Psi}(\mathcal{P}_{\Psi}(X)) &= \mathcal{P}_{\Psi}(0 \oplus R_V( \mathcal{P}_{\Psi} (X))) = 0 \oplus \overbar{\mathcal{P}}_{\Psi} ( R_V ( \mathcal{P}_{\Psi} (X))),
\end{align}
where $\overbar{\mathcal{P}}_{\Psi}:\B{\Hil_0^{\perp}}\to \B{\Hil_0^{\perp}}$ is defined as follows (see \cite[Theorem 12]{Lami2016entsaving}): 
\begin{align} \label{eq:phaseproj-1}
  \forall X\in \B{\Hil_0^{\perp}}: \quad \overbar{\mathcal{P}}_{\Psi}(X) &= \bigoplus_{k=1}^K \Tr_{k,2} (V_k^{\dagger} X V_k) \otimes \delta_k, 
\end{align}
where $V_k: \Hil_{k,1}\otimes \Hil_{k,2}\to \Hil_0^{\perp}$ are the canonical isometries and $\Tr_{k,2}$ is the partial trace over $\Hil_{k,2}$. By using the underlying decomposition $\Hil_0^{\perp}=\bigoplus_{k=1}^K \Hil_{k,1}\otimes \Hil_{k,2}$, we can equivalently write
\begin{align}\label{eq:phaseproj-2}
    \overbar{\mathcal{P}}_{\Psi} &= \bigoplus_k \id_{k,1} \otimes \mathcal{R}_{k,2},
\end{align}
where for each $k$, $\id_{k,1}:\B{\Hil_{k,1}}\to \B{\Hil_{k,1}}$ is the identity channel and $\mathcal{R}_{k,2}:\B{\Hil_{k,2}}\to \B{\Hil_{k,2}}$ is the replacer channel which acts as follows: $\mathcal{R}_{k,2}(X) = \Tr(X)\delta_k$. We should emphasize that $\overbar{\mathcal{P}}_{\Psi}:\B{\Hil_0^{\perp}}\to \B{\Hil_0^{\perp}}$ arises as the restriction of $\mathcal{P}_{\Psi}:\B{\Hil}\to \B{\Hil}$ to $\Hil_0^{\perp}$, in the sense that 
\begin{equation}
    \forall X\in \B{\Hil_0^{\perp}}: \quad \mathcal{P}_{\Psi}(0 \oplus X) = 0 \oplus \overbar{\mathcal{P}}_{\Psi}(X).
\end{equation}
From the above discussion, the following identities are easy to verify
\begin{align}
    \mathcal{P}_{\Psi} &= \mathcal{V}\circ R_V \circ \mathcal{P}_{\Psi}, \label{eq:PPbar-relations} \\
    R_V \circ \mathcal{P}_{\Psi} &= \overbar{\mathcal{P}}_{\Psi}\circ R_V\circ \mathcal{P}_{\Psi}, \label{eq:PPbar-relations2} \\ 
    \overbar{\mathcal{P}}_{\Psi} &=R_V\circ \mathcal{P}_{\Psi}\circ \mathcal{V},\label{eq:PPbar-relations3}
\end{align}
where $\mathcal{V}:\B{\Hil_0^{\perp}}\to \B{\Hil}$ is the isometric channel $\mathcal{V}(X)=VXV^{\dagger}$. These relations imply that the communication capacities of $\mathcal{P}_{\Psi}$ and $\overbar{\mathcal{P}}_{\Psi}$ are identical. We note this in the following two lemmas.

\begin{lemma}\label{lemma:PPar-op}
    For a quantum channel $\Psi:\B{\Hil}\to \B{\Hil}$, the operator systems associated with the peripheral projections $\mathcal{P}_{\Psi}$ and $\overbar{\mathcal{P}}_{\Psi}$ as above are isomorphic:
    \begin{equation}
        S_{\mathcal{P}_{\Psi}} \longrightarrow S_{\overbar{\mathcal{P}}_{\Psi}}, \quad S_{\overbar{\mathcal{P}}_{\Psi}} \longrightarrow S_{\mathcal{P}_{\Psi}}.
    \end{equation}
\end{lemma}
\begin{proof}
    The operator system relations (Lemma~\ref{lemma:op-homo}) along with Eq.~\eqref{eq:PPbar-relations} imply that 
    \begin{equation}
    S_{\mathcal{P}_{\Psi}} \subseteq S_{R_V \circ \mathcal{P}_{\Psi}} \subseteq S_{\mathcal{V}\circ R_V \circ \mathcal{P}_{\Psi}} = S_{\mathcal{P}_{\Psi}}.
\end{equation}
Hence, using Lemma~\ref{lemma:op-homo} again with Eq.~\eqref{eq:PPbar-relations2}, we obtain
\begin{equation}
     S_{R_V \circ \mathcal{P}_{\Psi}} = S_{\mathcal{P}_{\Psi}} \longrightarrow S_{\overbar{\mathcal{P}}_{\Psi}}.
\end{equation}
Finally, Lemma~\ref{lemma:op-homo} and Eq.~\eqref{eq:PPbar-relations3} show $ S_{\overbar{\mathcal{P}}_{\Psi}} = S_{R_V\circ \mathcal{P}_{\Psi}\circ \mathcal{V}} \longrightarrow S_{R_V\circ \mathcal{P}_{\Psi}} = S_{\mathcal{P}_{\Psi}}$.
\end{proof}

\begin{lemma}\label{lemma:PPar}
    Let $\Psi:\B{\Hil}\to \B{\Hil}$ be a channel with associated peripheral projections $\mathcal{P}_{\Psi}$ and $\overbar{\mathcal{P}}_{\Psi}$ as above. Then, for all $\epsilon\in [0,1)$ and $\mathbb{Q}\in \{C, C^{\operatorname{p}}, C^{\operatorname{ea}}, Q \}$, we have $\mathbb{Q}_{\epsilon}(\mathcal{P}_{\Psi})= \mathbb{Q}_{\epsilon}(\overbar{\mathcal{P}}_{\Psi})$.
\end{lemma}
\begin{proof}
    The bottleneck inequalities (Lemma~\ref{lemma:bottleneck-oneshot}) along with Eq.~\eqref{eq:PPbar-relations} imply that 
    \begin{equation}
        \mathbb{Q}_{\epsilon}(\mathcal{P}_{\Psi}) \geq \mathbb{Q}_{\epsilon}(R_V\circ \mathcal{P}_{\Psi}) \geq \mathbb{Q}_{\epsilon}(\mathcal{V}\circ R_V \circ \mathcal{P}_{\Psi}) = \mathbb{Q}_{\epsilon}(\mathcal{P}_{\Psi}),
    \end{equation}
    so that $\mathbb{Q}_{\epsilon}(\mathcal{P}_{\Psi}) = \mathbb{Q}_{\epsilon}(R_V\circ \mathcal{P}_{\Psi})$. Another use of Lemma~\ref{lemma:bottleneck-oneshot} along with Eq.~\eqref{eq:PPbar-relations2} shows
    \begin{equation}
        \mathbb{Q}_{\epsilon}(\mathcal{P}_{\Psi}) = \mathbb{Q}_{\epsilon}(R_V\circ \mathcal{P}_{\Psi}) \leq \mathbb{Q}_{\epsilon}(\overbar{\mathcal{P}}_{\Psi}).
    \end{equation}
    For the reverse inequality, we use Lemma~\ref{lemma:bottleneck-oneshot} and Eq.~\eqref{eq:PPbar-relations3} to write
    \begin{equation}
      \mathbb{Q}_{\epsilon}(\overbar{\mathcal{P}}_{\Psi}) =\mathbb{Q}_{\epsilon}(R_V\circ \mathcal{P}_{\Psi}\circ \mathcal{V}) \leq   \mathbb{Q}_{\epsilon}(\mathcal{P}_{\Psi}\circ \mathcal{V}) \leq \mathbb{Q}_{\epsilon}(\mathcal{P}_{\Psi}).
    \end{equation}
\end{proof}

Finally, we note that the peripheral space of quantum channels is known to be multiplicative \cite[Lemma 3.1]{fawzi2024error}. Below, we provide a different proof of this fact.

\begin{lemma}\label{lemma:peripheral-multi}
    For two channels $\Phi:\B{\Hil}\to \B{\Hil}$ and $\Psi:\B{\mathcal{K}}\to \B{\mathcal{K}}$, we have 
    \begin{equation}
        \mathscr{X} (\Phi\otimes \Psi) = \mathscr{X}(\Phi)\otimes \mathscr{X} (\Psi).
    \end{equation}
\end{lemma}
\begin{proof}
    Consider the Jordan decompositions of the two channels
    \begin{align}
    \Phi &= \sum_{i:|\lambda_i|=1} \lambda_i \mathcal{P}_i + \sum_{i:|\lambda_i|<1} \lambda_i\mathcal{P}_i + \mathcal{N}_i \\
    \Psi &= \sum_{j:|\mu_j|=1} \mu_j \mathcal{Q}_j + \sum_{j:|\mu_j|<1}\mu_j \mathcal{Q}_j + \mathcal{M}_j, 
\end{align}
where $\lambda_i, \mu_j$ are the distinct (respective) eigenvalues, $\mathcal{P}_i,\mathcal{Q}_j$ are the corresponding (respective) projectors and $\mathcal{N}_i, \mathcal{M}_j$ are the nilpotent parts. Note that
\begin{align}
    \Phi\otimes \Psi = \sum_{i,j: |\lambda_i|=|\mu_j|=1} \lambda_i\mu_j \mathcal{P}_i \otimes \mathcal{Q}_j  + \ldots,
\end{align}
where the remaining terms above contribute to the Jordan structure of $\Phi\otimes \Psi$ associated with non-peripheral eigenvalues $\lambda_i\mu_j$ with $|\lambda_i\mu_j|<1$ (see \cite[Theorem 4.3.17]{Horn1991matrix}). Hence, it is clear that $\mathcal{P}_{\Phi\otimes\Psi} = \sum_{i,j:|\lambda_i|=|\mu_j|=1} \mathcal{P}_i \otimes \mathcal{Q}_j = (\sum_{i:|\lambda_i|=1} \mathcal{P}_i) \otimes (\sum_{j: |\mu_j|=1} \mathcal{Q}_j) = \mathcal{P}_{\Phi}\otimes \mathcal{P}_{\Psi}$, so that 
\begin{equation}
    \mathscr{X}(\Phi\otimes \Psi) = \operatorname{range} (\mathcal{P}_{\Phi\otimes\Psi}) = \operatorname{range} (\mathcal{P}_{\Phi} \otimes \mathcal{P}_{\Psi}) = \operatorname{range} \mathcal{P}_{\Phi} \otimes \operatorname{range}\mathcal{P}_{\Psi} = \mathscr{X} (\Phi) \otimes \mathscr{X} (\Psi).
\end{equation}
\end{proof}

\section{Storage perspective} \label{sec:qms-storage}

In this section, we consider a quantum memory that we model as an open quantum system with a Markovian noise model $(\Psi^t )_{t\in \mathbb{N}}$ as described in the introduction. Our task is to store as much information in the memory as possible, in such a way that it can be reliably recovered (with some error $\epsilon\in [0,1)$) after the memory is left to evolve for some time $t\in \mathbb{N}$ (see Figure~\ref{fig:memory}). Building such a quantum memory that is able to store information for a long time is crucial in order to build a reliable quantum computer. Consequently, this task has been studied from different perspectives, most notably from the viewpoint of designing error-correcting codes to develop a fault-tolerant quantum memory \cite{Tehral2015memory, Brown2016memory, gottesman2016surviving}. Here, we adopt a Shannon-theoretic viewpoint, where we want to analyze the maximum amount of information that can be stored in the memory without placing any physical or computational restrictions on the encoding and decoding operations. Equivalently, we are interested in characterizing the {\em{one-shot}} $\epsilon-$error information-transmission capacities (Definitions~\ref{def:classical-protocol}-\ref{def:quantum-protocol}) of $\Psi^t$ for a given error $\epsilon\in [0,1)$ and time $t\in \mathbb{N}$. We will primarily focus on the long time $t\to \infty$ limit, and derive convergence bounds on when the storage capacities reach their infinite-time values. We will show that any information stored inside the peripheral space (see Section~\ref{sec:periphery})
\begin{align}
\mathscr{X} (\Psi) &:= \text{span}\{X\in \B{\Hil} : \exists \, \theta\in \mathbb{R} \text{ with } \Psi(X)= e^{i\theta} X\}   \nonumber \\
&\simeq 0 \oplus \bigoplus_{k=1}^K (\B{\C{d_k}}\otimes \delta_k)
\end{align}
is shielded from noise for an arbitrarily long time. Furthermore, as $t\to \infty$, we show that this is the optimal way to store data so that it can be recovered with good fidelity.

\begin{figure}[H]
    \centering
    \includegraphics[scale=2]{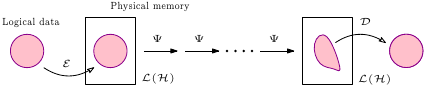}
    \caption{Schematic for a physical memory device experiencing Markovian noise modelled by a dQMS $(\Psi^t)_{t\in\mathbb{N}}$, where $\Psi:\B{\Hil}\to \B{\Hil}$ is a quantum channel. Logical data is encoded in the memory at time $t=0$ using an encoding channel $\mathcal{E}:\B{\Hil_{\operatorname{data}}}\to \B{\Hil}$. After some time $t$, a decoding channel $\mathcal{D}:\B{\Hil}\to \B{\Hil_{\operatorname{data}}}$ is applied to recover the data, so that $\mathcal{D}\circ \Psi^t \circ \mathcal{E}\simeq_{\epsilon} \operatorname{id}$ approximately simulates the identity channel on the data system, where $\epsilon\in [0,1)$ is the error allowed in the recovery process. }
    \label{fig:memory}
\end{figure}

Note that the Markovian model of noise given by a dQMS $(\Psi^t)_{t\in \mathbb{N}}$ covers the following settings:
    \begin{itemize}
        \item \textbf{Passive error-correction}: Here, $\Psi=\Psi_{\operatorname{noise}}$ and the dQMS models pure noise. No intermediate error correction is allowed in between time steps. 
        \item \textbf{Active error-correction}: Here, $\Psi=\Psi_{\operatorname{ecc}}\circ \Psi_{\operatorname{noise}}$, where $\Psi_{\operatorname{ecc}}$ is a fixed time-independent error-correction mechanism designed to detect/correct for errors induced by the noise $\Psi_{\operatorname{noise}}$ actively as they occur. We leave the analysis of more general adaptive error-correction procedures for future work.
    \end{itemize}

Furthermore, two defining features of the Shannon-theoretic framework that distinguishes it from the literature on error-correction are worth highlighting:

\begin{itemize}
    \item \textbf{Noise model:} The Shannon-theoretic framework aims to study the storage capacities of a quantum memory device with a fixed noise modelled by some dQMS $(\Psi^t)_{t\in \mathbb{N}}$. While error-correcting codes are traditionally designed to work for generic (local) noise, it is reasonable to expect that any physical implementation of a quantum computer includes information about additional noise structure. By understanding the dominant noise in the system, it is possible to employ strategies specifically designed for that noise model, thereby achieving more cost-effective suppression of the dominant noise. In a nutshell, the focus on Shannon theory is on the noise model (quantum channel), rather than a specific error correcting code. 
    
    \item \textbf{Encoding/decoding}: The Shannon-theoretic framework eliminates all physical and computational restrictions on how the information is encoded/decoded inside the memory. Moreover, it is assumed that the encoding/decoding can be performed perfectly without faults. Thus, this framework provides the most fundamental bounds on a memory's storage capacity that are allowed by quantum physics. Constraints on encoding/decoding can then be placed on top, depending on the specific physical implementation of the memory device under consideration. In this regard, it is pertinent to mention some recent papers that incorporate gate errors in a (asymptotic and memoryless) Shannon-theoretic setting, thus effectively doing Shannon theory in a fault-tolerant way; see \cite{Christandl2024fault, Belzig2024fault}.
\end{itemize}

\subsection{Main result} \label{sec:main-qms-storage}

Before we state and prove the main result of this section, it would be helpful for the reader to recall the spectral properties of quantum channels from Section~\ref{sec:spectral} and the structure of the peripheral space from Section~\ref{sec:periphery}. Here, we recall that for any channel $\Psi:\B{\Hil}\to \B{\Hil}$, there exists an orthogonal decomposition of the underlying Hilbert space $\Hil = \Hil_{0} \oplus \Hil_0^{\perp} = \Hil_0 \oplus \bigoplus_{k=1}^K \Hil_{k,1}\otimes \Hil_{k,2}$, with respect to which the peripheral space can be written as a direct sum of full matrix algebras of dimensions $d_k=\dim \Hil_{k,1}$: 
\begin{equation}\label{eq:phasespace-storage}
       \mathscr{X} (\Psi) = 0 \oplus \bigoplus_{k=1}^K (\B{\Hil_{k,1}}\otimes \delta_k). 
    \end{equation}
\begin{theorem}\label{theorem:main-storage}
Let $\Psi:\B{\Hil}\to \B{\Hil}$ be a quantum channel, $(\Psi^t)_{t\in \mathbb{N}}$ be the associated dQMS, and $\epsilon\in [0,1)$. Then, for all $t\in\mathbb{N}$, the one-shot $\epsilon-$error capacities satisfy:
\begin{align}   Q_{\epsilon}(\Psi^t) \geq Q_0(\Psi^t) &\geq \log (\max_k d_k ), \label{qlo}\\
C_{\epsilon}^{\operatorname{p}} (\Psi^t) \geq C_{0}^{\operatorname{p}}(\Psi^t)  &\geq \log (\max_k d_k), \label{cplo}\\\
C_{\epsilon}(\Psi^t) \geq C_{0}(\Psi^t) &\geq \log (\sum_k d_k ), \label{clo}\\
C_{\epsilon}^{\operatorname{ea}}(\Psi^t) \geq C_{0}^{\operatorname{ea}}(\Psi^t) &\geq \log (\sum_k d^2_k ).\label{cealo}
\end{align}
Moreover, for $t$ large enough, the following converse bounds hold:
\begin{align}
      Q_{\epsilon}(\Psi^t) &\leq \log (\max_k d_k) + \log(\frac{1}{1-\epsilon-\delta_t}), \label{eq:Qconverse}\\ 
      C_{\epsilon}^{\operatorname{p}}(\Psi^t) &\leq \log (\max_k d_k ) + \log(\frac{1}{1-\epsilon- \delta_t }), \label{eq:Pconverse} \\
     C_{\epsilon}(\Psi^t) &\leq  \log (\sum_k d_k) + \log(\frac{1}{1-\epsilon-\delta_t}), \label{eq:Cconverse}\\
     C_{\epsilon}^{\operatorname{ea}}(\Psi^t) &\leq  \log (\sum_k d_k^2 ) + \log(\frac{1}{1-\epsilon-\delta_t}).\label{eq:Ceaconverse}
\end{align}
Here, $d_k=\dim \Hil_{k,1}$ for $k\in \{1,2,\ldots ,K \}$ are the block dimensions in the decomposition of $\chi (\Psi)$ (see Eq.~\eqref{eq:phasespace-storage}), $\delta_t=\norm{\Psi^t-\Psi^t_{\infty}}_{\diamond} \leq \kappa\mu^t\to 0$ as $t\to \infty$, where $\mu=\operatorname{spr}(\Psi-\Psi_{\infty}), \kappa$ govern the convergence as in Eq.~\eqref{eq:converge}, and $t$ is large enough so that $\epsilon + \delta_t<1$.
\end{theorem}

\begin{proof}
We start by proving the achievability bounds in Eqs.~\eqref{qlo}-\eqref{cealo}. Throughout the proof, we work with orthogonal decomposition of the underlying Hilbert space $\Hil = \Hil_{0} \oplus \Hil_0^{\perp} = \Hil_0 \oplus \bigoplus_{k=1}^K \Hil_{k,1}\otimes \Hil_{k,2}$, with respect to which $\mathscr{X}(\Psi)$ assumes the decomposition in Eq.~\eqref{eq:phasespace-storage}.
\medskip

\noindent
{\textbf{Achievability: Quantum communication~\eqref{qlo}}}
\smallskip

Recall from Lemma~\ref{lemma:reverse} that the action of $\Psi$ on its peripheral space $\mathscr{X}(\Psi)$ is reversible, i.e., there exists a channel $\mathcal{R}:\B{\Hil}\to \B{\Hil}$ such that $\mathcal{R}\circ \Psi = \mathcal{P}_{\Psi}$, where $\mathcal{P}_{\Psi}$ is the projection onto the peripheral space (see Section~\ref{sec:periphery}). Thus, $\mathcal{R}^t\circ \Psi^t = \mathcal{P}_{\Psi}$ for all $t$. Using this fact, we can construct a $(d_k,0)$ quantum code (see Definition~\ref{def:quantum-protocol}) for $\Psi^t$ for all $t$ as follows. Let $V_k : \Hil_{k,1}\otimes \Hil_{k,2} \hookrightarrow \Hil_0^{\perp}$ and $V:\Hil_0^{\perp}\hookrightarrow \Hil$ be the canonical isometries. Define the encoder $\mathcal{E}_k : \B{\Hil_{k,1}}\to \B{\Hil}$ as
\begin{equation}
   \forall X\in \B{\Hil_{k,1}}: \quad \mathcal{E}_k(X) = V V_k (X \otimes \delta_k) (VV_k)^{\dagger}, 
\end{equation}
and the decoder $\mathcal{D}_{t,k}:\B{\Hil}\to \B{\Hil_{k,1}}$ as $\mathcal{D}_{t,k} = \Tr_{k,2} \circ R_{VV_k} \circ \mathcal{R}^t$, where $R_{VV_k}:\B{\Hil}\to \B{\Hil_{k,1}\otimes \Hil_{k,2}}$ is the restriction channel 
\begin{equation}
    R_{VV_k}(\cdot)= (VV_k)^{\dagger} (\cdot) VV_k + \Tr [(\iden_{\Hil} - VV_k (VV_k)^{\dagger})(\cdot) ]\sigma_k,
\end{equation}
for some arbitrary state $\sigma_k\in \State{\Hil_{k,1}\otimes \Hil_{k,2}}$ and $\Tr_{k,2}$ is the partial trace over $\Hil_{k,2}$. It is then evident that $\mathcal{D}_{t,k}\circ \Psi^t \circ \mathcal{E}_k : \B{\Hil_{k,1}}\to \B{\Hil_{k,1}}$ acts as the identity. Hence,
\begin{equation}
     \forall t\in \mathbb{N}, \forall \epsilon\in [0,1): \quad   \log (\max_k d_k) \leq Q_0 (\Psi^t) \leq Q_{\epsilon}(\Psi^t). 
\end{equation}

\medskip

\noindent
{\textbf{Achievability: Private classical communication~\eqref{cplo}}}
\smallskip

Since $Q_0 (\Psi)\leq C^{\operatorname{p}}_0 (\Psi)$ holds for any channel $\Psi$ (see Lemma~\ref{lemma:Q0<=P0}), we obtain
    \begin{equation}
    \forall t\in \mathbb{N}, \forall\epsilon\in [0,1): \quad   \log (\max_k d_k) \leq Q_0 (\Psi^t) \leq C^{\operatorname{p}}_0 (\Psi^t) \leq  C^{\operatorname{p}}_{\epsilon}(\Psi^t). 
    \end{equation}
  \medskip

\noindent
{\textbf{Achievability: Classical communication~\eqref{clo}}}
\smallskip

    For classical communication (Definition~\ref{def:classical-protocol}), we can send $\sum_{k=1}^K d_k$ messages perfectly (i.e., with $\epsilon=0$ error) through $\Psi^t$ for all $t$ by using the encoding states $\{ \ketbra{i_k} \otimes \delta_k \}$ for $k=1,2,\ldots ,K$ and $i_k = 1,2,\ldots , d_k$, where $\ketbra{i_k}$ are the diagonal matrix units in $\B{\Hil_{k,1}}$ and $\delta_k$ are given in Eq.~\eqref{eq:phasespace-storage}. Note that for each $k$, the state $\ketbra{i_k}\otimes \delta_k$ is supported only on $\Hil_{k,1}\otimes \Hil_{k,2}$. From the permutation+unitary action of $\Psi$ on its peripheral space (see Eq.~\eqref{eq:phaseaction}), it is clear that the outputs of these states under $\Psi^t$ are mutually orthogonal for all $t$ and hence, are perfectly distinguishable. Hence, 
    \begin{equation}
    \forall t\in \mathbb{N}, \forall\epsilon\in [0,1): \quad \log \left(\sum_{k=1}^K d_k\right)\leq C_0(\Psi^t)\leq  C_{\epsilon} (\Psi^t).
    \end{equation}
  \medskip

\noindent
{\textbf{Achievability: Entanglement-assisted classical communication~\eqref{cealo}}}
 \smallskip

    With entanglement assistance (Definition~\ref{def:entanglement-classical-protocol}), we can perfectly transmit $\sum_{k=1}^K d^2_k$ classical messages through $\Psi^t$ for all $t$. To see this, we start with an entangled state 
    \begin{equation}
        \frac{1}{K} \bigoplus_{k=1}^K (\psi^+_k \otimes \delta_k) \in \State{\oplus_k (\Hil \otimes \Hil_{k,1}\otimes \Hil_{k,2}) },
    \end{equation}
    where $\psi^+_k\in \State{\Hil\otimes \Hil_{k,1}}$ are maximally entangled states of Schmidt rank $d_k= \dim \, \Hil_{k,1}$ and $\delta_k$ are the states given in Eq.~\eqref{eq:phasespace-storage}. For each $k$, we apply an orthogonal set of unitary operators in $\B{\Hil_{k,1}}$ locally on $\Hil_{k,1}$ to encode $d_k^2$ many classical messages in orthogonal states\footnote{This is exactly the encoding scheme employed in the superdense coding protocol \cite{Bennett1992dense, Werner2001dense}.}, thus encoding $\sum_k d_k^2$ messages in total. The permutation+unitary action of $\Psi$ on its peripheral space (see Eq.~\eqref{eq:phaseaction}) ensures that these states remain orthogonal (and hence perfectly distinguishable) after the action of $\id_{\Hil} \otimes \Psi^t$ for all $t$. Thus,
    \begin{equation}
    \forall t\in \mathbb{N}, \forall\epsilon\in [0,1): \quad \log \left(\sum_{k=1}^K d^2_k\right)\leq  C_0^{\operatorname{ea}}(\Psi^t)\leq  C_{\epsilon}^{\operatorname{ea}} (\Psi^t).
    \end{equation}
\medskip

Next, we prove the converse bounds in Eqs.~\eqref{eq:Qconverse}-\eqref{eq:Ceaconverse}.
\smallskip

The proofs of these bounds for the quantum, classical, and entanglement-assisted classical capacities start similarly and so we consider them together below. Let us fix $\epsilon\in [0,1)$. Note that $\delta_t=\norm{\Psi^t-\Psi^t_{\infty}}_{\diamond} \to 0$ as $t\to \infty$ 
so that for $t$ large enough such that $\epsilon+\delta_t< 1$, we can use the `continuity-like' bound from Lemma~\ref{lemma:epsilon-delta} to obtain the following inequalities:
\begin{align}
    Q_{\epsilon}(\Psi^t) &\leq Q_{\epsilon + \delta_t} (\Psi^t_{\infty}) \leq Q_{\epsilon + \delta_t} (\mathcal{P}_{\Psi}),  \\
    C_{\epsilon}(\Psi^t) &\leq C_{\epsilon + \delta_t} (\Psi^t_{\infty}) \leq C_{\epsilon + \delta_t} (\mathcal{P}_{\Psi}),  \\
    C_{\epsilon}^{\operatorname{ea}}(\Psi^t) &\leq C^{\operatorname{ea}}_{\epsilon + \delta_t} (\Psi^t_{\infty}) \leq C^{\operatorname{ea}}_{\epsilon + \delta_t} (\mathcal{P}_{\Psi}). 
\end{align}
Note that the second set of inequalities above follow from the bottleneck inequalities in Lemma~\ref{lemma:bottleneck-oneshot}, since $\Psi^t_{\infty}=\Psi^t_{\infty} \circ \mathcal{P}_{\Psi} = \mathcal{P}_{\Psi}\circ \Psi^t_{\infty}$ for all $t\in\mathbb{N}$ (see Definition~\ref{def:peripheral-proj}). Next, without loss of generality, we assume that $\Hil_0=\{0\}$ in Eq.~\eqref{eq:phasespace-storage}\footnote{If $\Hil_0 \neq \{0\}$, we can choose to work with $\overbar{\mathcal{P}}_{\Psi}$ instead, which has the same capacities as $\mathcal{P}_{\Psi}$, see Lemma~\ref{lemma:PPar}.}, so that $\Hil = \oplus_{k=1}^K \Hil_{k,1}\otimes \Hil_{k,2}$ and $\mathcal{P}_{\Psi}=\overbar{\mathcal{P}}_{\Psi}$ projects onto 
\begin{equation}
\chi (\Psi) = \bigoplus_{k=1}^K (\B{\Hil_{k,1}}\otimes \delta_k)
\end{equation}
in the following way
\begin{align}\label{eq:phase-proj}
  \forall X\in \B{\Hil}: \quad \mathcal{P}_{\Psi}(X) &= \bigoplus_{k=1}^K \Tr_{k,2} (V^{\dagger}_k X V_k) \otimes \delta_k, 
\end{align}
where $V_k : \Hil_{k,1}\otimes \Hil_{k,2}\to \Hil$ are the canonical (inclusion) isometries and $\Tr_{k,2}$ denotes the partial trace over $\Hil_{k,2}$ (see Section~\ref{sec:periphery}). Before proceeding further, let us fix a label $A$ to the Hilbert space $\Hil$, so that our channel is defined as $\Psi:\B{\Hil_A}\to \B{\Hil_A}$. This will help us keep track of the subsystems involved.
\medskip

\noindent
{\textbf{Converse: Quantum communication~\eqref{eq:Qconverse}}}
\smallskip

For the quantum capacity, Lemma~\ref{lemma:one-shot-converse} shows that 
\begin{align}
    Q_{\epsilon + \delta_t}(\mathcal{P}_{\Psi}) &\leq I^c_{\max}(\mathcal{P}_{\Psi}) + \log(\frac{1}{1-\epsilon-\delta_t}) \\ 
    &= \sup_{\psi_{RA}} \inf_{\sigma_A} D_{\max}(\mathcal{P}_{A\to A}(\psi_{RA}) || \iden_R \otimes \sigma_A ) + \log(\frac{1}{1-\epsilon - \delta_t}), 
\end{align}
where the supremum is over all pure states $\psi_{RA}$ with $d_R=d_A$ (Remark~\ref{remark:pure}). We bound the first term above as follows. For a pure state $\psi_{RA}$, we use Eq.~\eqref{eq:phase-proj} to write
\begin{equation}
    \mathcal{P}_{A\to A}(\psi_{RA}) = \bigoplus_k \lambda_k \frac{1}{\lambda_k} \Tr_{k,2} \left[ (\iden_R\otimes V^{\dagger}_k) \psi (\iden_R \otimes V_k) \right] \otimes \delta_k = \bigoplus_k \lambda_k \theta_k \otimes \delta_k, 
\end{equation}
where $\lambda_k = \Tr \left[ (\iden_R\otimes V^{\dagger}_k) \psi (\iden_R \otimes V_k) \right]$ and each $\theta_k$ is a state in $\State{\Hil_R\otimes \Hil_{k,1}}$. Thus, by choosing $\sigma_A = \oplus_k \lambda_k \sigma_k \otimes \delta_k$, where $\sigma_k$ are arbitrary states in $\State{\Hil_{k,1}}$, we get 

\begin{align}
    \inf_{\sigma} D_{\max}(\mathcal{P}_{A\to A}(\psi_{RA}) || \iden_R \otimes \sigma_A ) &\leq \inf_{ \{ \sigma_k \}_k} 
    D_{\max}\left(\bigoplus_k \lambda_k \theta_k \otimes \delta_k \bigg\| \bigoplus_k \lambda_k \iden_R \otimes \sigma_k \otimes \delta_k \right) \nonumber \\ 
    &= \inf_{ \{\sigma_k\}_k } \max_k D_{\max} (\theta_k || \iden_R \otimes \sigma_k) \nonumber \\ 
    &= \max_k \inf_{\sigma_k} D_{\max} (\theta_k || \iden_R \otimes \sigma_k) \nonumber \\
    &\leq \log (\max_k d_k),
\end{align}
where the first equality follows from quasi-convexity of $D_{\max}$ (Remark~\eqref{remark:Dmax-quasi}), the second equality follows from Lemma~\ref{lemma:infmax}, and the last inequality follows from the fact that for any state $\rho_{AB}$,
\begin{equation}
    \inf_{\sigma\in \State{\Hil_B}} D_{\max}(\rho_{AB} \Vert \iden_A \otimes \sigma_B) \leq D_{\max}(\rho_{AB} \Vert \iden_A \otimes \iden_B/d_B) \leq \log d_B.
\end{equation}

\medskip

\noindent
{\textbf{Converse: Classical communication~\eqref{eq:Cconverse}}}
\smallskip

For the classical capacity, Lemma~\ref{lemma:one-shot-converse} shows that
\begin{align}
    C_{\epsilon + \delta_t}(\mathcal{P}_{\Psi}) &\leq \chi_{\max}(\mathcal{P}_{\Psi}) + \log(\frac{1}{1-\epsilon-\delta_t}) \\    
    &=\sup_{\rho_{MA}} \inf_{\sigma_A} D_{\max}(\mathcal{P}_{A\to A}(\rho_{MA}) || \rho_M \otimes \sigma_A ) + \log(\frac{1}{1-\epsilon - \delta_t}), 
\end{align}
where the supremum is over all cq states $\rho_{MA}$ (Definition~\ref{def:channel-measures}). Now, for an arbitrary cq state $\rho_{MA}=\sum_m p_m \ketbra{m}_M \otimes \rho^m_A$, note that 
\begin{align}
    \inf_{\sigma}D_{\max}(\mathcal{P}_{A\to A}(\rho_{MA}) || \rho_M \otimes \sigma_A ) = \inf_{\sigma} \max_m  D_{\max}(\mathcal{P}_{A \to A}(\rho^m_A) || \sigma_A ) \leq \log (\sum_k d_k),
\end{align}
where the inequality follows by choosing $\sigma = (\oplus_k \iden_k \otimes \delta_k)/\sum_k d_k$. This is because for any state $\rho$, its projection $\mathcal{P}(\rho)$ onto the peripheral space is dominated by $\oplus_k (\iden_k \otimes \delta_k )$ (see Eq.~\eqref{eq:phase-proj}).

\medskip
\medskip

\noindent
{\textbf{Converse: Entanglement-assisted classical communication~\eqref{eq:Ceaconverse}}}
\smallskip

We again use Lemma~\ref{lemma:one-shot-converse} to write
\begin{align}
    C^{\operatorname{ea}}_{\epsilon+\delta_t}(\mathcal{P}_{\Psi}) &\leq I_{\max}(\mathcal{P}_{\Psi}) + \log(\frac{1}{1-\epsilon - \delta_t}) \\ 
    &=\sup_{\psi_{RA}} \inf_{\sigma_A} D_{\max}(\mathcal{P}_{A\to A}(\psi_{RA}) || \psi_R \otimes \sigma_A) + \log(\frac{1}{1-\epsilon-\delta_t}), 
\end{align}
where the supremum is over all pure states $\psi_{RA}$ with $d=d_R=d_A$ (Remark~\ref{remark:pure}). Note that the supremum here is achieved by a maximally entangled state (see \cite[Remark 2]{Fang2020smooth} and Lemma~\ref{lemma:Imax})
\begin{equation}
  \psi^+_{RA}= \Omega_{RA}/d  \quad \text{where} \quad  \ket{\Omega}_{RA} = \sum_i \ket{\alpha_i}_R \otimes \ket{\alpha_i}_A.
\end{equation}
We assume that the basis of $\Hil_R \simeq \Hil_A = \oplus_k \Hil_{k,1}\otimes \Hil_{k,2}$ chosen above is 
\begin{equation}
    \ket{\alpha_{(i_k, j_k)}} = \ket{i_k}_{\Hil_{k,1}} \otimes \ket{j_k}_{\Hil_{k,2}},
\end{equation}
where for each $k$, $\{\ket{i_k}_{\Hil_{k,1}}\}_{i_k}$ and $\{\ket{j_k}_{\Hil_{k,2}}\}_{j_k}$ are the canonical bases for $\Hil_{k,1}$ and $\Hil_{k,2}$, respectively. Now, 
\begin{align}
     \mathcal{P}_{A\to A}(\psi^+_{RA}) = \frac{1}{d} \bigoplus_k \Tr_{k,2} \left[ (\iden_R\otimes V^{\dagger}_k) \Omega_{RA} (\iden_R \otimes V_k) \right] \otimes \delta_k = \frac{1}{d} \bigoplus_k \theta_k \otimes \delta_k, 
\end{align}
where $\theta_k = \Tr_{k,2} \left[ (\iden_R\otimes V^{\dagger}_k) \Omega_{RA} (\iden_R \otimes V_k) \right]$ is a positive operator in $\B{\Hil_R\otimes \Hil_{k,1}}$ with $\norm{\theta_k}_{\infty}=d_k$. Let us choose $\sigma_A = \oplus_k \lambda_k (\iden_k /d_k \otimes \delta_k )$, where $\{ \lambda_k = d^2_k / \sum_{l} d^2_{l} \}_{k=1}^K$ is a probability distribution. Then,
\begin{align}
 \inf_{\sigma} D_{\max}(\mathcal{P}_{A\to A}(\psi^+_{RA}) || \psi_R^+ \otimes \sigma_A) &\leq D_{\max}\left(\frac{1}{d} \bigoplus_k \theta_k \otimes \delta_k \bigg\|  \frac{1}{d} \bigoplus_k \lambda_k \iden_R \otimes \frac{\iden_k}{d_k} \otimes \delta_k \right) \nonumber \\
  &= \max_k D_{\max} \left( \theta_k \bigg\| \lambda_k \iden_R \otimes \frac{\iden_k}{d_k} \right) \nonumber \\
  &= \max_k \log \left( \frac{d_k}{\lambda_k} \norm{\theta_k}_{\infty}  \right) \nonumber \\
  &= \max_k\log (d^2_k / \lambda_k  ) = \log \left(\sum_k d^2_k\right),
\end{align}
where the first equality follows from quasi-convexity of $D_{\max}$ (Remark~\ref{remark:Dmax-quasi}) and the second equality follows from Definition~\ref{def:divergence} 
\medskip

\noindent
{\textbf{Converse: Private classical communication~\eqref{eq:Pconverse}}}
\smallskip

The proof of the converse bound for private classical capacity requires a slightly different line of argumentation. We again fix $\epsilon\in [0,1)$ and note that $\delta_t=\norm{\Psi^t-\Psi^t_{\infty}}_{\diamond} \to 0$ as $t\to \infty$, 
so that for $t$ large enough such that $\epsilon+\delta_t< 1$, we can use Theorem~\ref{lemma:QP<=E}, \ref{lemma:Eepsilon-delta} and \ref{lemma:channel-bottlenecks}, and the fact that $\Psi^t_{\infty}=\Psi^t_{\infty} \circ \mathcal{P}_{\Psi} = \mathcal{P}_{\Psi}\circ \Psi^t_{\infty}$ for all $t\in \mathbb{N}$ to write 
\begin{align}
    C^{\operatorname{p}}_{\epsilon}(\Psi^t) \leq E_H^{\epsilon} (\Psi^t) &\leq E_H^{\epsilon + \delta_t} (\Psi^t_{\infty})  \nonumber\\ 
    &\leq E_H^{\epsilon + \delta_t} (\mathcal{P}_{\Psi}) \nonumber \\
    &\leq E_{\max}(\mathcal{P}_{\Psi}) + \log (\frac{1}{1-\epsilon-\delta_t}) \\
    &= \sup_{\psi_{RA}} \inf_{\sigma_{RA}\in \operatorname{SEP}(R:A)} D_{\max} (\mathcal{P}_{A\to A}(\psi_{RA}) \Vert \sigma_{RA}) + \log (\frac{1}{1-\epsilon-\delta_t}),
\end{align}
where the supremum is over all pure states $\psi_{RA}\in \State{\Hil_R \otimes \Hil_A}$ with $d_R=d_A$. We bound the first term above as follows. For a pure state $\psi_{RA}$, we use Eq.~\eqref{eq:phase-proj} to write
\begin{equation}
    \mathcal{P}_{A\to A}(\psi_{RA}) = \bigoplus_k \lambda_k \frac{1}{\lambda_k} \Tr_{k,2} \left[ (\iden_R\otimes V^{\dagger}_k) \psi (\iden_R \otimes V_k) \right] \otimes \delta_k = \bigoplus_k \lambda_k \theta_k \otimes \delta_k, 
\end{equation}
where $\lambda_k = \Tr \left[ (\iden_R\otimes V^{\dagger}_k) \psi (\iden_R \otimes V_k) \right]$ and each $\theta_k$ is a state in $\State{\Hil_R\otimes \Hil_{k,1}}$. Thus, by choosing $\sigma_{RA} = \oplus_k \lambda_k \sigma_k \otimes \delta_k$, where $\sigma_k$ are arbitrary separable states in $\State{\Hil_R\otimes \Hil_{k,1}}$, we get 

\begin{align}
    \inf_{\sigma\in \operatorname{SEP}(R:A)} D_{\max}(\mathcal{P}_{A\to A}(\psi_{RA}) || \sigma_{RA} ) &\leq \inf_{ \{ \sigma_k \}_k } 
    D_{\max}\left(\bigoplus_k \lambda_k \theta_k \otimes \delta_k \bigg\| \bigoplus_k \lambda_k \sigma_k \otimes \delta_k \right) \nonumber \\ 
    &= \inf_{ \{\sigma_k\}_k } \max_k D_{\max} (\theta_k || \sigma_k) \nonumber \\ 
    &= \max_k \inf_{\sigma_k} D_{\max} (\theta_k || \sigma_k) \nonumber \\
    &\leq \log ( \max_k d_k),
\end{align}
where the first equality follows from the quasi-convexity of $D_{\max}$ (Remark~\ref{remark:Dmax-quasi}), the second equality follows from Lemma~\ref{lemma:infmax}, and the last inequality follows from the fact that for any state $\rho_{AB}$ (see Lemma~\ref{lemma:Emax}),
\begin{equation}
    \inf_{\sigma\in \operatorname{SEP}(A:B)} D_{\max}(\rho_{AB}\Vert \sigma_{AB}) \leq \log \min (d_A,d_B ).
\end{equation}
\end{proof}

\subsection{Applications} \label{sec:consequences-qms-storage}

In this section, we collect some simple consequences of Theorem~\ref{theorem:main-storage}. Firstly, we note that the capacities of the peripheral projection channel $\mathcal{P}_{\Psi}$ can be bounded as follows.

\begin{lemma}\label{corollary:cap-projection}
    Let $\Psi:\B{\Hil}\to \B{\Hil}$ be a channel with associated peripheral projection channels $\mathcal{P}_{\Psi}:\B{\Hil}\to \B{\Hil}$ and $\overbar{\mathcal{P}}_{\Psi}:\B{\Hil}\to \B{\Hil}$ from Section~\ref{sec:periphery}. Let $\epsilon\in [0,1)$. Then,
    
    \begin{alignat*}{2}
    \log (\max_k d_k) &\leq  Q_{\epsilon}(\mathcal{P}_{\Psi})\leq  \log (\max_k d_k) + \log (\frac{1}{1-\epsilon}),  \\ 
     \log (\max_k d_k) &\leq C^{\operatorname{p}}_{\epsilon}(\mathcal{P}_{\Psi}) \leq  \log (\max_k d_k) + \log (\frac{1}{1-\epsilon}),  \\ 
     \log\left(\sum_k d_k\right) &\leq C_{\epsilon}(\mathcal{P}_{\Psi}) \leq \log\left(\sum_k d_k\right) + \log (\frac{1}{1-\epsilon}), \\
     \log\left(\sum_k d^2_k\right) &\leq C^{\operatorname{ea}}_{\epsilon}(\mathcal{P}_{\Psi}) \leq \log\left(\sum_k d^2_k\right) + \log (\frac{1}{1-\epsilon}),
    \end{alignat*}
    where $d_k=\dim \Hil_{k,1}$ for $k\in \{1,2,\ldots ,K \}$ are the block dimensions in the decomposition of $\chi (\Psi)$.
\end{lemma}
\begin{proof}
    The lower bounds follow trivially, since $\mathcal{P}_{\Psi}$ projects onto the peripheral space $\mathscr{X}(\Psi)$ and hence acts as identity on the matrix blocks $\B{\Hil_{k,1}}$ in Eq.~\eqref{eq:phasespace-storage}. The converse bounds were obtained in the proof of Theorem~\ref{theorem:main-storage}. Also note that the communication capacities of $\mathcal{P}_{\Psi}$ and $\overbar{\mathcal{P}}_{\Psi}$ are identical, see Lemma~\ref{lemma:PPar}. 
\end{proof}

Next, we derive the infinite-time capacities of a dQMS by taking the limit $t\to \infty$ in Theorem~\ref{theorem:main-storage}.

\begin{corollary}\label{corollary:main}
For a channel $\Psi:\B{\Hil}\to \B{\Hil}$ and $\epsilon\in [0,1)$, the following holds true:
    \begin{align}
     \log (\max_k d_k) \leq\lim_{t\to \infty } Q_{\epsilon}(\Psi^t) &\leq  \log (\max_k d_k) + \log (\frac{1}{1-\epsilon}) \label{eq:Qinfty}  \\ 
     \log (\max_k d_k) \leq\lim_{t\to \infty } C^{\operatorname{p}}_{\epsilon}(\Psi^t) &\leq  \log (\max_k d_k) + \log (\frac{1}{1-\epsilon})  \\ 
     \log\left(\sum_k d_k\right)\leq \lim_{t\to \infty } C_{\epsilon}(\Psi^t) &\leq \log\left(\sum_k d_k\right) + \log (\frac{1}{1-\epsilon}) \label{eq:Cinfty} \\
     \log\left(\sum_k d^2_k\right)\leq \lim_{t\to \infty } C^{\operatorname{ea}}_{\epsilon}(\Psi^t) &\leq \log\left(\sum_k d^2_k\right) + \log (\frac{1}{1-\epsilon}),
\end{align}
where $d_k=\dim \Hil_{k,1}$ for $k\in \{1,2,\ldots ,K \}$ are the block dimensions in the decomposition of $\chi (\Psi)$.
\end{corollary}

\begin{remark}
    The bottleneck relation (Lemma~\ref{lemma:bottleneck-oneshot}) implies that the limit $\lim_{t\to \infty} \mathbb{Q}(\Psi^t)$ always exists for any channel $\Psi:\B{\Hil}\to \B{\Hil}$ and $\mathbb{Q}\in \{Q_{\epsilon}, C^{\operatorname{p}}_{\epsilon}, C_{\epsilon}, C^{\operatorname{ea}}_{\epsilon} \}$, even though the sequence $(\Psi^t)_{t\in \mathbb{N}}$ might not itself have a limit.
\end{remark}

\begin{remark}
    Eqs.~\eqref{eq:Qinfty} and \eqref{eq:Cinfty} were independently proved in \cite{singh2024zero} (for the $\epsilon=0$ case) and in \cite{fawzi2024error} (for arbitrary $\epsilon\in [0,1)$). The $\epsilon=0$ case of Eq.~\eqref{eq:Cinfty} is also proved in \cite{guan2016zero}. 
\end{remark}

\begin{remark}
    The achievability bounds in Eqs.~\eqref{qlo}-\eqref{cealo} of Theorem~\ref{theorem:main-storage} are obtained by constructing codes that work with $\epsilon=0$ error. It is unclear whether these bounds can be improved by explicitly taking $\epsilon$ into account. In this regard, we note that the achievability bounds in \cite{fawzi2024error} on quantum and classical capacities take $\epsilon$ into account and are slightly better than the ones in Eqs.~\eqref{qlo},\eqref{clo}. However, the error criteria they consider when defining the capacities are of the average kind, as opposed to the worst-case error criteria that we employ (see Definitions~\ref{def:classical-protocol}-\ref{def:quantum-protocol}). 
\end{remark}

\subsubsection{Rate of convergence}\label{subsec:convergence-storage}
Given the infinite-time capacities of a dQMS generated by $\Psi:\B{\Hil}\to \B{\Hil}$ with $d=\dim \Hil$ as in Corollary~\ref{corollary:main}, it is natural to ask how quickly the capacities converge to their infinite time values. Below, we provide a rough estimate. 

\begin{lemma}\label{lemma:lambert-channel}
    Let $\Psi:\B{\Hil}\to \B{\Hil}$ be a quantum channel with $\mu=\operatorname{spr}(\Psi-\Psi_{\infty})<1$. Furthermore, let $\mu_0$ be such that $\mu<\mu_0<1$ and let $\delta\in (0,1)$. Then, $\norm{\Psi^t-\Psi^t_{\infty}}_{\diamond}\leq \delta$ if $t\geq\frac{\mu}{\mu_0-\mu}$ and 
    \begin{equation*}
        t \geq \frac{D}{\ln (1/\mu)} \left( \frac{ \ln (D^{D+3/2}/\delta') }{D} - \ln \left(\frac{\mu \ln (1/\mu)}{1-\mu^2} \right) + \sqrt{2}\sqrt{\frac{ \ln (D^{D+3/2}/\delta')  }{D} - \ln \left(\frac{\mu \ln (1/\mu)}{1-\mu^2} \right) -1} \right),
    \end{equation*}
    where $d=\dim\Hil$, $D=d^2$ and $\delta'=\frac{\delta(1-\mu_0)^{3/2}}{8e^2}$.
\end{lemma}
\begin{proof}
    From \cite[Corollary 4.4]{Szehr2014specconvergence}, we know that for $t>\mu/(1-\mu)$:
\begin{align}
    \norm{\Psi^t - \Psi^t_{\infty}}_{\diamond} &\leq \frac{4e^2 \sqrt{D} (D + 1)}{ \left( 1- (1+ \frac{1}{t})\mu \right)^{3/2}}  \left(\frac{t (1-\mu^2)}{\mu}\right)^{D-1 } \mu^t  \nonumber \\
    &\leq \frac{8e^2 D^{3/2}}{(1-\mu_0)^{3/2}} \left(\frac{t (1-\mu^2)}{\mu}\right)^{D} \mu^t,  \label{eq:LHS<delta}
\end{align}
where the latter inequality follows from $t\geq\frac{\mu}{\mu_0-\mu}>\frac{\mu}{1-\mu}$. Then, Lemma~\ref{lemma:lambert} shows that if 
\begin{equation}
    t \geq \frac{D}{\ln (1/\mu)} \left( \frac{ \ln (D^{D+3/2}/\delta') }{D} - \ln \left(\frac{\mu \ln (1/\mu)}{1-\mu^2} \right) + \sqrt{2}\sqrt{\frac{ \ln (D^{D+3/2}/\delta')  }{D} - \ln \left(\frac{\mu \ln (1/\mu)}{1-\mu^2} \right) -1} \right),
\end{equation}
the RHS in Eq.~\eqref{eq:LHS<delta} is bounded as 
\begin{equation}
    D^{3/2} \left(\frac{t (1-\mu^2)}{\mu}\right)^{D} \mu^t \leq \frac{\delta(1-\mu_0)^{3/2}}{8e^2} = \delta'.
\end{equation}
\end{proof}

The above estimate shows that for any $d-$dimensional quantum memory experiencing noise modelled by a dQMS $(\Psi^t)_{t\in \mathbb{N}}$, 
the asymptotic $(t\to \infty)$ behavior of the noise starts kicking in after time $t \gtrsim d^2\ln d$. Hence, in accordance with Theorem~\ref{theorem:main-storage}, the storage capacity after this time is bounded as
\begin{equation}
    \log(\max_k d_k) \leq Q_{\epsilon}(\Psi^t) \lesssim \log(\max_k d_k) + \log (\frac{1}{1-\epsilon}).
\end{equation}
As discussed in \cite{Szehr2014specconvergence}, the stated bound $t\gtrsim d^2 \ln (d) \geq 2^{2n}$ (which grows exponentially in the number of qubits $n$ in the memory) cannot be improved in general by only utilizing knowledge about the spectrum of $\Psi$. However, in the zero-error case, it is possible to obtain a slightly stronger convergence estimate. Recall that the zero-error capacities of a channel $\Psi:\B{\Hil}\to \B{\Hil}$ admit alternate characterizations in terms of the independence numbers of the non-commutative confusability graph of $\Psi$ (Theorem~\ref{thm:DSW}):
\begin{equation}
    S_{\Psi} = \operatorname{span}\{K^{\dagger}_i K_j : 1\leq i,j\leq p\},
\end{equation}
where $\Psi(X)=\sum_{i=1}^p K_iXK_i^{\dagger}$ is a Kraus representation of $\Psi$ (see Definition~\ref{def: op-sys}). Now, for any channel $\Psi$, we prove in the next section that the operator systems $(S_{\Psi^t})_{t\in \mathbb{N}}$ form an increasing chain of subspaces that stabilizes after time $T$, where $T\leq d^2-\dim S_{\Psi}$ (see Lemma~\ref{lemma:op-chain}):
\begin{equation}
     S_{\Psi} \subset S_{\Psi^2} \subset \ldots \subset S_{\Psi^T} = S_{\Psi^{T+1}} = \ldots
\end{equation}
Hence, the zero-error storage capacities of any $d-$dimensional quantum memory device with an arbitrary dQMS $(\Psi^t)_{t\in \mathbb{N}}$ also stabilize after the same time. We note this in the following result.

\begin{theorem}\label{theorem:main-zero-storage}
Let $\Psi:\B{\Hil}\to \B{\Hil}$ be a quantum channel. Then, $\exists T\leq d^2-\dim S_{\Psi}$ such that
    \begin{equation}
    Q_0(\Psi) \geq Q_0(\Psi^2) \geq \ldots \geq Q_0(\Psi^T)=Q_0(\Psi^{T+1}) = \ldots = \log (\max_k d_k),
    \end{equation}
where $d=\dim \Hil$ and $d_k=\dim \Hil_{k,1}$ for $k\in \{1,2,\ldots ,K \}$ are the block dimensions in the decomposition of $\chi (\Psi)$. Similar stabilization happens for all the other capacities from Theorem~\ref{theorem:main-storage}.
\end{theorem}

\subsubsection{Memories with independent and identical (IID) noise} \label{subsec:IIDconvergence}

In this section, we deal with the special case of IID memory devices. Consider a quantum memory comprised of $n$ qudits, so that the Hilbert space $\Hil \simeq (\C{q})^{\otimes n}$ and the noise $\Psi$ acts independently and identically on all the qudits, i.e., $\Psi=\Gamma^{\otimes n}$ for some channel $\Gamma:\B{\C{q}}\to \B{\C{q}}$. Now, Lemma~\ref{lemma:peripheral-multi} shows that the peripheral space is multiplicative under tensor product: $\chi(\Psi) = \chi (\Gamma)^{\otimes n}$. Hence, applying Corollary~\ref{corollary:main}, we can infer that the infinite-time capacities are additive under tensor product:
\begin{equation}\label{eq:IID}
  \forall \epsilon\in [0,1), \forall n \in \mathbb{N}: \quad \log (\max_{k} d_k) \leq \lim_{t \to \infty} \frac{1}{n} Q_{\epsilon}((\Gamma^{\otimes n})^t) \leq \log (\max_{k} d_k) + \frac{1}{n}\log(\frac{1}{1-\epsilon}),
\end{equation}
where the integers $d_k$ now come from the block decomposition of $\chi (\Gamma)\subseteq \B{\C{q}}$. Hence, the optimal data storage rate of the memory, i.e., the number of logical qubits stored per physical qudit, in the asymptotic time $t\to \infty$ and space $n\to \infty$ limit is independent of $\epsilon$:
\begin{equation}
  \forall \epsilon\in [0,1): \quad \lim_{n\to \infty} \lim_{t \to \infty} \frac{1}{n} Q_{\epsilon}((\Gamma^{\otimes n})^t) = \log (\max_k d_k).
\end{equation}

Let us illuminate the above result from a different perspective. Suppose we try to store qubits in the $n-$qudit memory device at a rate $R>0$ and fidelity $\epsilon\in [0,1)$ by using $(\ceil{2^{nR}},\epsilon)$ quantum codes $(\mathcal{E}_n, \mathcal{D}_{n,t})$ (Definition~\ref{def:quantum-protocol}) for the noise channel $(\Gamma^{\otimes n})^t$. The encoder $\mathcal{E}_{n}:\B{\C{\ceil{2^{nR}}}}\to \B{\Hil}$ encodes logical qubits in the physical memory at time $t=0$, and the decoder $\mathcal{D}_{n,t}:\B{\Hil}\to \B{\C{\ceil{2^{nR}}}}$ is designed to recover the logical data at time $t$, such that the channel fidelity $F(\mathcal{D}_{n,t}\circ (\Gamma^{\otimes n})^t \circ \mathcal{E}_n ) \geq 1-\epsilon$ for all time $t$. Then, by rearranging the right inequality in Eq.~\eqref{eq:IID}, we see that 
\begin{equation}
    \epsilon \geq 1 - 2^{-n \left( R - \log (\max_k d_k) \right) }.
\end{equation}
Thus, if we try to store data in the memory at a rate $R>\log(\max_k d_k)$, the recovery error $\epsilon\to 1$ exponentially as $n\to \infty$. On the other hand, the coding scheme used in the proof of Theorem~\ref{theorem:main-storage} shows that the rate $R=\log(\max_k d_k)$ is achievable with exactly zero-error, since $Q_0(\Gamma^{\otimes n})\geq n \log (\max_k d_k)$ because of superadditivity of the zero-error capacities (Remark~\ref{remark:zero-super}). In other words, any Markovian IID noise process satisfies the \emph{strong converse property} in the asymptotic time $t\to \infty$ limit (see Figure~\ref{fig:error-rate-qms-storage-chapter}). Similar
results hold for the other capacities as well.

\begin{figure}
    \centering
    \includegraphics[width=0.5\linewidth]{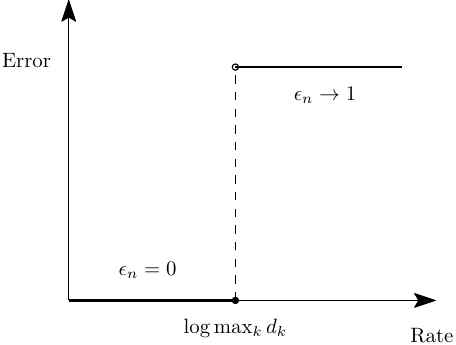}
    \caption{The error-rate curve for data storage in a quantum memory with $n$ qudits experiencing IID Markovian noise modelled by a dQMS $((\Gamma^{\otimes n})^t)_{t\in \mathbb{N}}$, where $\Gamma:\B{\C{q}}\to \B{\C{q}}$ is a quantum channel. One can store $n(\log\max_k d_k)$ many logical qubits inside the largest block in the peripheral space $\mathscr{X}(\Gamma)$ (see Eq.~\eqref{eq:phasespace-storage}) perfectly with no recovery error $\epsilon_n=0$. Moreover, any attempt to store data at a higher rate $R>\log\max_k d_k$ fails with certainty as $t\to \infty$, since the data recovery error $\epsilon_n\to 1$ exponentially as $n\to \infty$.}
    \label{fig:error-rate-qms-storage-chapter}
\end{figure}

As before, we can ask how quickly the capacities converge to their infinite-time limits. In order to apply Theorem~\ref{theorem:main-storage} and the convergence estimate from Section~\ref{subsec:convergence-storage}, we need to ensure that the joint channel $\Gamma^{\otimes n}$ is close to the asymptotic part $(\Gamma^{\otimes n})_{\infty}$, which factors like $(\Gamma_{\infty})^{\otimes n}$. This is because the peripheral space is multiplicative, so the projection $\mathcal{P}_{\Gamma^{\otimes n}}$ onto $\mathscr{X}(\Gamma^{\otimes n})=\mathscr{X}(\Gamma)^{\otimes n}$ also factors $\mathcal{P}_{\Gamma^{\otimes n}}=\mathcal{P}_{\Gamma}^{\otimes n}$, see Lemma~\ref{lemma:peripheral-multi}. Hence, 
\begin{equation}
    (\Gamma^{\otimes n})_{\infty} = \mathcal{P}_{\Gamma^{\otimes n}} \circ \Gamma^{\otimes n} = \mathcal{P}_{\Gamma}^{\otimes n} \circ \Gamma^{\otimes n} = (\mathcal{P}_{\Gamma}\circ \Gamma)^{\otimes n} = \Gamma_{\infty}^{\otimes n}.
\end{equation}
Now, we can use the following estimate:
\begin{align}
  \norm{\Psi^t - \Psi_{\infty}^t}_{\diamond} &= \norm{(\Gamma^{\otimes n})^t - (\Gamma_{\infty}^{\otimes n})^t}_{\diamond} \\
  &= \norm{(\Gamma^t)^{\otimes n} - (\Gamma_{\infty}^t)^{\otimes n}}_{\diamond}
\leq n \norm{\Gamma^t - \Gamma_{\infty}^t}_{\diamond},
\end{align}
where we have used \cite[Proposition 3.48]{watrous2018theory} to obtain the inequality. Let $\mu=\operatorname{spr}(\Gamma-\Gamma_{\infty})$ define the spectral gap $1-\mu$ of $\Gamma$ and $\mu_0\in (\mu,1)$. Then, applying Lemma~\ref{lemma:lambert-channel} to the local channel $\Gamma$ shows that for $t\geq \frac{\mu}{\mu_0-\mu}$ and 
\begin{align}
    t \geq \frac{q^2}{\ln (1/\mu)} \Bigg( \frac{ \ln (n) + \ln( q^{2(q^2+3/2)}/\delta') }{q^2}& - \ln \left(\frac{\mu \ln (1/\mu)}{1-\mu^2} \right) \nonumber \\ 
    &+ \sqrt{2}\sqrt{\frac{ \ln (n) + \ln (q^{2(q^2+3/2)}/\delta')}{q^2} - \ln \left(\frac{\mu \ln (1/\mu)}{1-\mu^2} \right) -1} \Bigg),
\end{align}
we have $\norm{\Gamma^t - \Gamma_{\infty}^t}_{\diamond}\leq \delta/n$, where $\delta'=\frac{\delta(1-\mu_0)^{3/2}}{8e^2}$. Hence, after time $t \gtrsim \ln (n)$, the storage capacities are bounded as
\begin{equation}
    \log(\max_k d_k) \leq \frac{1}{n} Q_{\epsilon}((\Gamma^{\otimes n})^t) \lesssim \log(\max_k d_k) + \frac{1}{n}\log (\frac{1}{1-\epsilon}).
\end{equation}
Note that the convergence here is incredibly rapid: for a memory comprised of $n$ qubits (so, $q=2$), the infinite time capacities in the IID case are reached after time $t\gtrsim\ln (n)$, where as in the non-IID case, it takes exponential time $t\gtrsim 2^{2n}$ to do so (Section~\ref{subsec:convergence-storage}).

\subsubsection{Lower bound on overhead of fault-tolerance} \label{subsec:fault-tolerance}

In this subsection, we provide an example of how our results can be used to provide lower bounds on the space overhead of fault-tolerant quantum memories. To this end, let us assume that the memory is comprised of $n$ qubits, undergoing IID depolarizing noise, so that $\Hil = (\C{2})^{\otimes n}$ and the dQMS is given by $((\mathcal{D}_p^{\otimes n})^t)_{t\in\mathbb{N}}$, where $\mathcal{D}_p:\B{\C{2}}\to \B{\C{2}}$ (for $0<p<1$) is the qubit depolarizing channel 
\begin{equation}
    \mathcal{D}_p (X) =  (1-p)X + p \operatorname{Tr}(X) \frac{\iden_2}{2}.
\end{equation}
Furthermore, suppose that we engineer an error-correction procedure $\Psi_{\operatorname{ecc}}:\B{\Hil}\to \B{\Hil}$ to periodically detect and correct errors induced by the depolarizing noise (see Figure~\ref{fig:depol-memory}). Thus, the memory experiences a modified Markovian process $(\Psi^t)_{t\in \mathbb{N}}$ in time, where $\Psi =  \Psi_{\operatorname{ecc}} \circ \mathcal{D}_p^{\otimes n}$. Observe that the depolarizing noise is \emph{strictly contractive} \cite{Raginsky2002contract}: 
\begin{equation}
\eta_{\operatorname{Tr}}(\mathcal{D}_p) := \sup_{\rho\neq \sigma} \frac{\norm{\mathcal{D}_p (\rho) - \mathcal{D}_p(\sigma)}_1}{\norm{\rho - \sigma}_1} = 1-p < 1,    
\end{equation}
where the supremum is over distinct states $\rho,\sigma\in \State{\Hil}$. Moreover, \cite{Park2012super, Shirokov2015zerosuper} shows that $\eta_{\operatorname{Tr}}(\mathcal{D}_p^{\otimes n}) < 1$ for all $n\in \mathbb{N}$, since $\mathcal{D}_p$ is a qubit channel\footnote{It is easy to show that the one-shot zero-error classical capacity of a channel $\Phi$ vanishes if and only if $\Phi$ is strictly contractive: $C_0(\Phi)=0 \iff \eta_{\Tr}(\Phi)<1$ \cite[Proposition 4.2]{Hiai2015contraction}. The results in \cite{Park2012super, Shirokov2015zerosuper} show that for qubit channels $\Phi, \Psi:\B{\C{2}}\to \B{\C{2}}$, if $C_0(\Phi)=C_0(\Psi)=0$, then $C_0(\Phi \otimes \Psi)=0$. In higher dimensions, this is no longer true \cite{Chen2010zerosuper}}. Hence, $\eta_{\operatorname{Tr}}(\Psi) \leq \eta_{\operatorname{Tr}}(\mathcal{D}_p^{\otimes n}) < 1$, which implies that there exists a unique fixed state $\sigma\in \State{\Hil}$ such that $\lim_{t\to \infty} \Psi^t (\rho) = \sigma$ for all input states $\rho\in \State{\Hil}$. Hence, the peripheral space is trivial $\chi (\Psi) = \operatorname{Fix}(\Psi) = \operatorname{span} \{\sigma\}$ with $K=1$ and $d_1=1$ (see Section~\ref{subsec:mixing}). Now, if we want to store $m_{\epsilon}(t)$ logical qubits in the memory for time $t$ with error at most $\epsilon\in [0,1)$, then $m_{\epsilon}(t) \leq Q_{\epsilon}(\Psi^t)$ by Definition~\ref{def:quantum-protocol}. As a result, the analysis from Section~\ref{subsec:convergence-storage} shows that when 
\begin{equation*}
    t \geq \frac{D}{\ln (1/\mu)} \left( \frac{ \ln (D^{D+3/2}/\epsilon') }{D} - \ln \left(\frac{\mu \ln (1/\mu)}{1-\mu^2} \right) + \sqrt{2}\sqrt{\frac{ \ln (D^{D+3/2}/\epsilon')  }{D} - \ln \left(\frac{\mu \ln (1/\mu)}{1-\mu^2} \right) -1} \right),
\end{equation*}
the norm $\norm{\Psi^t-\Psi^t_{\infty}}_{\diamond}\leq \epsilon$ and Theorem~\ref{theorem:main-storage} shows that the number of logical qubits 
\begin{align}
    m_{\epsilon}(t) \leq Q_{\epsilon}(\Psi^t) \leq \log \frac{ 1 }{1 - 2\epsilon}.
\end{align}
Here, $D=2^{2n}=(\dim\Hil)^2$, $\mu=\operatorname{spr}(\Psi-\Psi_{\infty})$, $\mu_0\in (\mu,1)$, $\epsilon'=\frac{\epsilon(1-\mu_0)^{3/2}}{8e^2}$ and we assume that $\epsilon<1/2$. In other words, the memory becomes useless for data storage after time $t\gtrsim D\ln(D) \gtrsim n2^{2n}$. Thus, regardless of what error-correction procedure we come up with, the memory's lifetime can scale at most exponentially with the number of qubits. Put differently, in order to store even a single logical qubit in the memory for $t$ time steps, we need at least $n \gtrsim \log t$ physical qubits. This bound is similar to the one obtained in~\cite{fawzi2022lower}.

\begin{figure}
    \centering
    \includegraphics[width=0.6\linewidth]{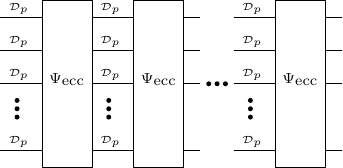}
    \caption{Quantum memory with $n$ physical qubits undergoin IID depolarizing noise with parameter $p\in (0,1)$. An error-correction procedure $\Psi_{\operatorname{ecc}}$ is designed to periodically detect and correct errors induced by the noise. We show that the memory becomes useless for data storage after time $t\gtrsim n2^{2n}$ scaling exponentially with the number of qubits. }
    \label{fig:depol-memory}
\end{figure}

Suppose now that the architecture of the memory device is such that it only allows us to perform error-correction on $k$ qubits at a time, where $k$ is some fixed constant independent of $n$. More precisely, $\Psi_{\operatorname{ecc}}:\B{\Hil}\to \B{\Hil}$ is of the form $\Psi_{\operatorname{ecc}}=\Gamma_{\operatorname{ecc}} \otimes \Gamma_{\operatorname{ecc}} \otimes \ldots$, where $\Gamma_{\operatorname{ecc}}:\B{(\C{2})^{\otimes k}}\to \B{(\C{2})^{\otimes k}}$ is designed to detect and correct errors on a subset of qubits of size $k$ (see Figure~\ref{fig:depol-memory-local}). In this $k$-local setting, the memory experiences a modified IID Markovian process $(\Psi^t)_{t\in \mathbb{N}}$ in time, where $\Psi = \Gamma^{\otimes m}$ and $ \Gamma = \Gamma_{\operatorname{ecc}} \circ \mathcal{D}_p^{\otimes k}$, with the number of qubits being $n=mk$. In this case, the IID convergence analysis from Section~\ref{subsec:IIDconvergence} shows that when 
\begin{align*}
    t \geq \frac{q^2}{\ln (1/\mu)} \Bigg( \frac{ \ln (m) + \ln( q^{2(q^2+3/2)}/\epsilon') }{q^2}& - \ln \left(\frac{\mu \ln (1/\mu)}{1-\mu^2} \right) \nonumber \\ 
    &+ \sqrt{2}\sqrt{\frac{ \ln (m) + \ln (q^{2(q^2+3/2)}/\epsilon')}{q^2} - \ln \left(\frac{\mu \ln (1/\mu)}{1-\mu^2} \right) -1} \Bigg),
\end{align*}
the norm $\norm{\Psi^t-\Psi^t_{\infty}}_{\diamond}\leq m\norm{\Gamma^t - \Gamma^t_{\infty}} \leq \epsilon$ and the number of logical qubits $m_{\epsilon}(t)$ stored in memory is 
\begin{align}
    m(t) \leq Q_{\epsilon}(\Psi^t) \leq \log \frac{ 1 }{1 - 2\epsilon}.
\end{align}
Here, $q=2^{k}$, $\mu=\operatorname{spr}(\Gamma-\Gamma_{\infty})$, $\mu_0\in (\mu,1)$, $\epsilon'=\frac{\epsilon(1-\mu_0)^{3/2}}{8e^2}$ and we assume that $\epsilon<1/2$. Observe how the locality restriction on error-correction severely reduces the memory's lifetime. When $\Psi_{\operatorname{ecc}}$ was allowed to act globally across all $n$ qubits, the memory became useless for data storage after an exponential number of time steps $t\gtrsim n2^{2n}$. In contrast, if $\Psi_{\operatorname{ecc}}$ is $k$-local, the memory becomes useless incredibly rapidly after time $t\gtrsim \ln (n)$. This is similar to the bound obtained in \cite{Razborov2004fault}. 

\begin{remark}
The results in this section also hold true if we replace the depolarizing noise model with any other noisy channel from the set of \emph{strictly contractive} qubit channels, which forms a \emph{dense} convex subset of the set of all qubit channels \cite{Raginsky2002contract}.     
\end{remark}

\begin{figure}
    \centering
    \includegraphics[width=0.6\linewidth]{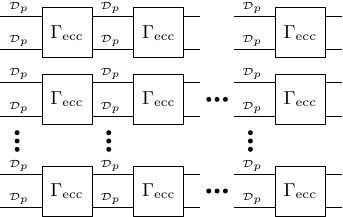}
    \caption{Quantum memory with $n$ physical qubits undergoin IID depolarizing noise with parameter $p\in (0,1)$. A local error-correction procedure $\Psi_{\operatorname{ecc}}=\Gamma_{\operatorname{ecc}}\otimes \Gamma_{\operatorname{ecc}} \otimes \ldots \otimes \Gamma_{\operatorname{ecc}}$ is designed to periodically detect and correct errors induced by the noise, where $\Gamma_{\operatorname{ecc}}$ acts only on two qubits. We show that the memory becomes useless for data-storage after time $t\gtrsim \ln (n)$ scaling logarithmically with the number of qubits. }
    \label{fig:depol-memory-local}
\end{figure}

\subsection{Examples} \label{sec:examples-qms-storage}
In this section, we discuss some examples of Markovian processes $(\Psi^t)_{t\in \mathbb{N}}$ for channels $\Psi:\B{\Hil}\to \B{\Hil}$ and their long-time storage capacities as stated in Theorem~\ref{theorem:main-storage}.

\subsubsection{Mixing noise} \label{subsec:mixing}

What kinds of Markovian processes $(\Psi^t)_{t\in \mathbb{N}}$ are so noisy that they do not allow for even a single classical bit of data to be stored in memory for an arbitrarily long time? According to Theorem~\ref{theorem:main-storage}, these are precisely those processes for which the peripheral space $\chi (\Psi) = 0 \oplus \bigoplus_{k=1}^K (\B{\Hil_{k,1}}\otimes \delta_k)$ is \emph{trivial}, i.e., $\chi (\Psi)$ is such that the following equivalent conditions hold
$$K=1 \text{ and } \dim \Hil_{1,1} = 1  \iff \sum_{k=1}^K \dim \Hil_{k,1} = 1 \iff \sum_{k=1}^K (\dim \Hil_{k,1})^2 =1.$$ 
It is easy to see that the above is equivalent to the process being \emph{mixing}, in the sense that there exists a unique state $\sigma\in \State{\Hil}$ such that 
\begin{equation}
\forall X\in \B{\Hil}: \quad \lim_{t\to \infty} \Psi^t(X) = \Tr (X) \sigma .    
\end{equation}
Hence, for such processes, we get
\begin{align}
    0 \leq \lim_{t\to \infty } C_{\epsilon}(\Psi^t) &\leq  \log (\frac{1}{1-\epsilon})  \\
    0\leq \lim_{t\to \infty } C^{\operatorname{ea}}_{\epsilon}(\Psi^t) &\leq \log (\frac{1}{1-\epsilon}).    
\end{align} 
Notice that even with assistance from pre-shared entanglement, mixing processes are useless for storing classical data as $t\to \infty$. If $\epsilon=0$, the limiting capacity is exactly $0$. Typical examples of such processes are given by the qubit depolarizing $\mathcal{D}_p : \B{\C{2}}\to \B{\C{2}}$ and amplitude damping $\mathcal{A}_{\gamma}: \B{\C{2}}\to \B{\C{2}}$ channels, with respective parameters $p, \gamma\in (0,1)$:
\begin{align}
    \mathcal{D}_p (X) &:= p \operatorname{Tr}(X) \frac{\iden_2}{2} + (1-p)X, \\
    \mathcal{A}_{\gamma} (X) &:= \begin{bmatrix} X_{00} + \gamma X_{11} & \sqrt{1-\gamma} X_{01} \\ \sqrt{1-\gamma} X_{10} & (1-\gamma) X_{11} \end{bmatrix}.
\end{align}
It is easy to check that $\lim_{t\to \infty} \mathcal{D}^t_p (X) = \Tr (X) \iden_2/2$ and $\lim_{t\to \infty} \mathcal{A}^t_{\gamma}(X) = \Tr (X) \ketbra{0}$.

\subsubsection{Asymptotically entanglement breaking noise}
What kinds of Markovian processes are so noisy that they don't allow for even a single logical qubit to be stored in memory for an arbitrarily long time? According to Theorem~\ref{theorem:main-storage}, these are precisely those processes for which the peripheral space $\chi (\Psi) = 0 \oplus \bigoplus_{k=1}^K (\B{\Hil_{k,1}}\otimes \delta_k)$ is \emph{commutative}, all the matrix blocks in $\chi (\Psi)$ are one-dimensional: $\max_k \dim \Hil_{k,1}=1$. One can prove that this is equivalent to the process being \emph{asymptotically entanglement breaking}, in the sense that all the limit points of the set $(\Psi^t)_{t\in \mathbb{N}}$ are entanglement-breaking \cite{Lami2016entsaving}. For such processes, we have
\begin{align}
    0 \leq \lim_{t\to \infty } Q_{\epsilon}(\Psi^t) &\leq  \log (\frac{1}{1-\epsilon})  \\
    0\leq \lim_{t\to \infty } C^{\operatorname{p}}_{\epsilon}(\Psi^t) &\leq \log (\frac{1}{1-\epsilon}).    
\end{align} 
A typical example of such a process is given by the dephasing noise $\Psi_p:\B{\C{2}}\to \B{\C{2}}$
\begin{align}
    \Psi_p (X) := \begin{bmatrix} X_{00} & p X_{01} \\ p X_{10} & X_{11} \end{bmatrix},
\end{align}
where the dephasing parameter $p\in \C{}$ is such that $|p|<1$.

\subsubsection{Collective decoherence noise} 
In this subsection, we consider Markovian processes for which the peripheral space contains non-trivial matrix blocks, so that logical qubits can be stored in memory for an arbitrarily long time. Consider a memory with $n$ physical qubits undergoing Markovian noise $(\Psi^t)_{t\in \mathbb{N}}$ generated by a \emph{correlated Pauli noise} channel $\Psi: \B{(\C{2})^{\otimes n}}\to \B{(\C{2})^{\otimes n}}$ defined as
\begin{equation}
    \Psi(X) := p_0 X + p_x \sigma_x^{\otimes n} X \sigma_x^{\otimes n} + p_y \sigma_y^{\otimes n} X \sigma_y^{\otimes n} + p_z \sigma_z^{\otimes n} X \sigma_z^{\otimes n},
\end{equation}
where $\sigma_x,\sigma_y,\sigma_z\in \B{\C{2}}$ are the Pauli matrices and $p_0,p_x,p_y,p_z>0$ such that $p_0+p_x+p_y+p_z=1$ \cite{Li2011correlated, Kondo2013correlated}. The \emph{interaction algebra} of the noise \cite{Knill2000error} is defined as the algebra generated by its Kraus operators $A(\Psi):= \operatorname{Alg} \{\sigma_x^{\otimes n}, \sigma_y^{\otimes n}, \sigma_z^{\otimes n} \}$. Since Pauli matrices form a basis of $\B{\C{2}}$, we get 
\begin{align}
    A(\Psi):= \operatorname{Alg} \{\sigma_x^{\otimes n}, \sigma_y^{\otimes n}, \sigma_z^{\otimes n} \} = \operatorname{Alg} \{ M^{\otimes n} : M\in \B{\C{2}} \}.
\end{align}
Moreover, since $\Psi$ is unital, the set of fixed points can be obtained by taking the \emph{commutant} of the Kraus operators \cite{Busch1998fixed, kribs2003fixed}:
\begin{align}
    \operatorname{Fix}(\Psi) &:= \{X\in \B{\C{2}}^{\otimes n} : \Psi(X) = X \} = A(\Psi)' \nonumber \\
    &= \{\sigma_x^{\otimes n}, \sigma_y^{\otimes n}, \sigma_z^{\otimes n} \}' := \{X\in \B{\C{2}}^{\otimes n} : [X,\sigma_x^{\otimes n}] = [X, \sigma_y^{\otimes n}] = [X, \sigma_z^{\otimes n}]=0 \}.
\end{align}
Note that the peripheral space of $\Psi$ is contained in the \emph{multiplicative domain} \cite{choi1974, paulsen-book, rahaman2017}: 
\begin{align}
    \chi(\Psi) \subseteq \mathcal{M}(\Psi):= \{X\in \B{\C{2}}^{\otimes n} : \forall Y\in \B{\C{2}}^{\otimes n}, \, \Psi(XY) &= \Psi(X)\Psi(Y) \nonumber \\ 
    \Psi(YX) &= \Psi(Y)\Psi(X) \}.
\end{align}
Since the multiplicative domain arises as the fixed point set of $\Psi^* \circ \Psi$ \cite{rahaman2017}, we obtain 
\begin{align}
    \mathcal{M}(\Psi) = \operatorname{Fix}(\Psi^*\circ \Psi) = A(\Psi^* \circ \Psi)' &= \{ (\sigma_i \sigma_j)^{\otimes n} : i,j=x,y,z \}' \nonumber \\ 
    &= \{\sigma_x^{\otimes n}, \sigma_y^{\otimes n}, \sigma_z^{\otimes n} \}'.
\end{align}
Thus, the following chain of inequalities collapes into equalities:
\begin{align}
    \{\sigma_x^{\otimes n}, \sigma_y^{\otimes n}, \sigma_z^{\otimes n} \}' = \operatorname{Fix}(\Psi) \subseteq \chi (\Psi) \subseteq \mathcal{M}(\Psi) = \{\sigma_x^{\otimes n}, \sigma_y^{\otimes n}, \sigma_z^{\otimes n} \}'.
\end{align}
To compute the relevant algebras, we employ Schur-Weyl duality \cite{Fulton2004rep}  to write
\begin{align}
    A(\Psi) &\simeq \bigoplus_{0\leq j\leq \floor{n/2}} \iden_{f(n-j,j)} \otimes \B{\C{^{g(n-j,j)}}} \nonumber  \\
    \operatorname{Fix}(\Psi) = \chi(\Psi)=A(\Psi)' &\simeq \bigoplus_{0\leq j\leq \floor{n/2}} 
    \B{\C{^{f(n-j,j)}}} \otimes \iden_{g(n-j,j)},
\end{align}
where $f(n-j,j) = {n \choose j} - {n \choose j-1} $ and $g(n-j,j) = n+1-2j$ for $j=0,1,\ldots \floor{n/2}$ and $f(n,0)=1$. Here, we are using the notation from \cite{Li2015collective} (see also \cite{Holbrook2003commutant, Holbrook2005collective}, where the set of fixed points were studied in detail). Thus, Theorem~\ref{theorem:main-storage} shows that the maximum number of logical qubits that can be stored in memory for an arbitrarily long time with error $\epsilon\in[0,1)$ is given by
\begin{equation}
    \log (\max_{0\leq j\leq \floor{n/2}} f(n-j,j)) \leq \lim_{t\to \infty} Q_{\epsilon}(\Psi^t) \leq \log (\max_{0\leq j\leq \floor{n/2}} f(n-j,j)) + \log (\frac{1}{1-\epsilon}).
\end{equation}
Computing the maximum above is an interesting discrete optimization problem \cite{Li2015collective}. Here, we just note that the rate of data storage can be shown to approach $1$ as the number of physical qubits $n\to \infty$ \cite{Kempe2001collective, Li2015collective}:
\begin{align}
    \lim_{n\to \infty} \lim_{t\to \infty} \frac{1}{n} Q_{\epsilon}(\Psi^t) &= \lim_{n\to \infty} \frac{1}{n} \log \max_{0\leq j\leq \floor{n/2}} f(n-j,j) \nonumber \\
    &\geq \lim_{n\to \infty} \frac{1}{n} \log f(n- \floor{n/2},\floor{n/2}) = 1.
\end{align}
One can generalize the above results to general collective decoherence processes $\Psi:\B{(\C{d})^{\otimes n}} \to \B{(\C{d})^{\otimes n}}$, where each qudit undergoes an identical random rotation given by a unitary $U\in\mathbb{U}(d)$ \cite{Zanardi1997collective, Viola2001collective, Fortunato2003collective, Boileau2004collective}: 
\begin{equation}
    \Psi(X) := \int_{U\in\mathbb{U(d)}} U^{\otimes n} X (U^{\dagger})^{\otimes n} d\mu (U) ,
\end{equation}
where $\mu$ is some Borel probability measure on the unitary group $\mathbb{U}(d)$. The relevant algebras in this setting can again be described using Schur-Weyl duality, and have been analyzed in detail in \cite{Junge2005collective, Li2015collective}.

\section{Transmission perspective} \label{sec:qms-transmission}

In this section, we return to the point-to-point communication setting between two spatially separated parties, where Alice wants to use a noisy communication channel $\Phi:\B{\Hil_A}\to \B{\Hil_B}$ many times to faithfully transmit information to Bob. We will assume that the noise in the communication link between Alice and Bob is \emph{Markovian}. We model this by considering a dQMS $(\Psi^l)_{l \in {\mathbb{N}}}$, where $\Psi :\B{\Hil}\to \B{\Hil}$ is a noisy channel, $\Psi^l = \Psi \circ \Psi \cdots \circ \Psi$ is the $l$-fold composition of $\Psi$ with itself, and $l \in {\mathbb{N}}$ plays the role of the length of communication link. Physically, the noise in each unit length of the communication link is modelled by $\Psi$, and since the noise is Markovian, the cumulative noise in a length $l$ segment is given by $\Psi^l$. We will restrict ourselves to `long' noisy communication links $\Psi^l$ of length $l \gtrsim d^2 \ln (d)$, where $d=\dim\Hil$. Mathematically, we are interested in the capacities of channels $\Phi:\B{\Hil}\to \B{\Hil}$ that are $l$-Markovian divisible\footnote{The notion of  {\em{divisibility}} of quantum channels has long been the focus of active research, especially in the study of open quantum systems. See e.g.~\cite{Wolf2008dividing}, \cite{Rivas2014dividing, Breuer2016dividing, Chruciski2022dividing} and references therein.} for `large' $l$, i.e., channels $\Phi$ for which there exists another channel $\Psi$ such that $\Phi=\Psi^l$ with $l \gtrsim d^2\ln (d)$.

\subsection{Zero-error setting}\label{sec:main-zero}
We first tackle the problem of zero-error communication (Definition~\ref{def:zero-error-capacity}) via highly Markovian divisble quantum channels. Recall that the zero-error capacities of a quantum channel $\Phi$ can be expressed in terms of the independence numbers of its operator system $S_{\Phi}$ (Section~\ref{sec:zero-error}). However, computing these independence numbers is a challenging task \cite{Shor2008complexity}.  Our first result shows that if an operator system $S\subseteq \B{\Hil}$ is a $*-$algebra (i.e., $S$ is a $\dagger$-closed subspace containing the identity $\iden$ and closed under matrix multiplication), then its independence numbers can be explicitly computed and are multiplicative. Recall that if $S\subseteq \B{\Hil}$ is a $*-$algebra, then there exists a decomposition $\Hil = \bigoplus_{k} \Hil_{k,1}\otimes \Hil_{k,2}$ such that $S=\bigoplus_{k} \left( \iden_{k,1} \otimes \B{\Hil_{k,2}} \right)$ \cite{Arveson1976algebra, Takesaki1979algebra}. Moreover, this decomposition can be efficiently computed \cite{zarikan2003algebra, Holbrook2003commutant, fawzi2024error}.  

\begin{lemma}\label{lemma:alpha-algebra}
    Let $S= \bigoplus_{k} \left( \iden_{k,1} \otimes \B{\Hil_{k,2}} \right) \subseteq \B{\Hil}$ be a $*-$algebra, where the block structure is with respect to the underlying decomposition $\Hil = \bigoplus_k \Hil_{k,1}\otimes \Hil_{k,2}$ with $d_k = \dim \Hil_{k,1}$. Then,  
    \begin{align*} 
     \sum_k d_k &= \alpha(S) \\
     \sum_k d^2_k &= \alpha_{ea}(S) \\
     \max_k d_k &= \alpha_q (S) = \alpha_p (S) 
    \end{align*}
    Furthermore, if $T$ is another $*-$algebra as above,
    \begin{align*}
    \alpha(S\otimes T) &= \alpha(S)\alpha(T), \\
    \alpha_{ea}(S\otimes T) &= \alpha_{ea}(S)\alpha_{ea}(T), \\
    \alpha_p(S\otimes T) &= \alpha_p(S) \alpha_p(T), \\
    \alpha_q(S\otimes T) &= \alpha_q(S) \alpha_q(T).
    \end{align*}
\end{lemma}

\begin{proof}
Consider the channel $\overbar{\mathcal{P}}:\B{\Hil}\to \B{\Hil}$ defined as
\begin{align}
    \overbar{\mathcal{P}} &= \bigoplus_k \id_{k,1} \otimes \mathcal{R}_{k,2},
\end{align}
where $\id_{k,1}:\B{\Hil_{k,1}}\to \B{\Hil_{k,1}}$ is the identity channel and $\mathcal{R}_{k,2}:\B{\Hil_{k,2}}\to \B{\Hil_{k,2}}$ is the replacer channel defined as $\mathcal{R}_{k,2}(X) = \Tr(X)\delta_k$ for some states $\delta_k\in \State{\Hil_{k,2}}$. Then, since the operator systems of the identity and replacer channels are $S_{\id_{k,1}}=\mathbb{C}\iden_{k,1}$ and $S_{\mathcal{R}_{k,2}}=\B{\Hil_{k,2}}$, respectively, it is clear that $S = S_{\overbar{\mathcal{P}}}$. Moreover, Corollary~\ref{corollary:cap-projection} shows that for $\epsilon\in [0,1)$, 
\begin{align}
    \log (\max_k d_k) \leq Q_{\epsilon}(\overbar{\mathcal{P}}) &\leq  \log (\max_k d_k) + \log (\frac{1}{1-\epsilon}),  \\ 
     \log (\max_k d_k) \leq C^{\operatorname{p}}_{\epsilon}(\overbar{\mathcal{P}}) &\leq  \log (\max_k d_k) + \log (\frac{1}{1-\epsilon}),  \\ 
     \log\left(\sum_k d_k\right)\leq  C_{\epsilon}(\overbar{\mathcal{P}}) &\leq \log\left(\sum_k d_k\right) + \log (\frac{1}{1-\epsilon}), \\
     \log\left(\sum_k d^2_k\right)\leq  C^{\operatorname{ea}}_{\epsilon}(\overbar{\mathcal{P}}) &\leq \log\left(\sum_k d^2_k\right) + \log (\frac{1}{1-\epsilon}).
     \end{align}
     Notice that the upper and lower bounds coincide if $\epsilon=0$. Hence, the desired expressions follow from the correspondence between the zero-error capacities and independence numbers (Theorem~\ref{thm:DSW}).

    Now, let $T=\bigoplus_j \left( \iden_{j,1} \otimes \B{\Kil_{j,2}} \right) \subseteq \B{\Kil}$ be another $*-$algebra, where the block structure is with respect to the decomposition $\Kil = \bigoplus_j \Kil_{j,1}\otimes \Kil_{j,2}$ and $d_j'=\dim \Kil_{j,1}$. Then, $S\otimes T\subseteq \B{\Hil\otimes \Kil}$ is also a $*-$algebra with the block structure
    \begin{equation}
        S\otimes T = \bigoplus_{k,j} \iden_{k,1} \otimes \iden_{j,1} \otimes \B{\Hil_{k,2}\otimes \Kil_{j,2}}
    \end{equation}
    with respect to the decomposition $\Hil\otimes\Kil= \bigoplus_{k,j} \Hil_{k,1} \otimes \Kil_{j,1} \otimes \Hil_{k,2} \otimes \Kil_{j,2}$. Hence,
    \begin{align}
        \alpha(S\otimes T) = \sum_{k,j} d_k d_j' = \sum_k d_k \sum_j d_j' = \alpha(S)\alpha(T).
    \end{align}
    The multiplicativity of $\alpha_p, \alpha_{ea}$ and $\alpha_q$ follow similarly.
\end{proof}

Next, we introduce the notion of Markovian divisbility. 

\begin{definition}
    A quantum channel $\Phi:\B{\Hil}\to \B{\Hil}$ is said to be $l$-Markovian divisible if there exists another quantum channel $\Psi:\B{\Hil}\to \B{\Hil}$ such that 
    \begin{equation}
        \Phi = \underbrace{\Psi\circ \Psi \circ \ldots \circ\Psi}_{l \operatorname{times}} =: \Psi^l.
    \end{equation}
\end{definition}

The following lemma notes a very important stabilization property of operator systems under identical sequential compositions of quantum channels $\Psi\circ \Psi \circ \ldots \Psi$. 

\begin{lemma}\label{lemma:op-chain}
    Let $\Psi: \B{\Hil}\to \B{\Hil}$ be a quantum channel with $d=\dim \Hil$. Then, there exists $L\leq d^2 - \dim S_{\Psi}$ such that the following chain of (strict) inclusions and equalities are true:
    \begin{equation*}
        S_{\Psi} \subset S_{\Psi^2} \subset \ldots \subset S_{\Psi^L} = S_{\Psi^{L+1}} = \ldots
    \end{equation*}
\end{lemma}
\begin{proof}
    Lemma~\ref{lemma:op-homo} shows that for all $l\in\mathbb{N}$, $S_{\Psi^l} \subseteq S_{\Psi^{l+1}}$. Moreover, if $S_{\Psi^l}=S_{\Psi^{l+1}}$ for some $l$, 
    \begin{align}
        S_{\Psi^{l+2}} &= \operatorname{span}\{ K_i^{\dagger} X K_j : 1\leq i,j\leq m, X\in S_{\Psi^{l+1}} \} \\
        &= \operatorname{span}\{ K_i^{\dagger} X K_j : 1\leq i,j\leq m, X\in S_{\Psi^{l}} \} \\
        &= S_{\Psi^{l+1}} = S_{\Psi^l}
    \end{align}
    where $\Psi(X)=\sum_{i=1}^m K_i XK_i^{\dagger}$ is some Kraus representation of $\Psi$. Proceeding by induction, we get that $S_{\Psi^{l+k}}=S_{\Psi^l}$ for all $k$. Define $L:= \min \{l\in \mathbb{N} : S_{\Psi^l}=S_{\Psi^{l+1}} \}$, so that 
    \begin{equation}
        S_{\Psi} \subset S_{\Psi^2} \subset \ldots \subset S_{\Psi^L} = S_{\Psi^{L+1}} = \ldots S_{\Psi^{L+k}} = \ldots ,
    \end{equation}
    where the strictness of the inclusions follows from the minimality of $L$. Moreover, since $\dim \B{\Hil} = d^2$, we must have $L\leq d^2 - \dim S_{\Psi}$.
\end{proof}

The above result motivates us to introduce the notion of the ``stabilized operator system'' of a quantum channel. The definition is similar in spirit to that of the stabilized multiplicative domain or the decoherence-free algebra of a unital completely positive map \cite{rahaman2017, Carbone2019stabilized}. 

\begin{definition}
    Let $\Psi:\B{\Hil}\to \B{\Hil}$ be a quantum channel with $d=\dim \Hil$. We define the \emph{stabilized operator system} of $\Psi$ as follows:
    \begin{equation*}
        S_{\Psi^{\infty}} := \bigcup_{l\in \mathbb{N}} S_{\Psi^l} = S_{\Psi^{d^2}} = S_{\Psi^{d^2+1}} =\ldots ,
    \end{equation*}
    where the latter equalities follow from Lemma~\ref{lemma:op-chain}.
\end{definition}

\subsubsection{Structure of the stabilized operator system} \label{sec:stablized}
We now provide a complete characterization of the structure of the stabilized operator system of a quantum channel, which might be of independent interest, especially from the perspective of non-commutative graph theory \cite{Daws2024qgraph}. Recall from Section~\ref{sec:periphery} that for any channel $\Psi:\B{\Hil}\to \B{\Hil}$, there exists a decomposition $\Hil = \Hil_0 \oplus \Hil_0^{\perp} = \Hil_0 \oplus \bigoplus_k \Hil_{k,1} \otimes \Hil_{k,2}$ such that the peripheral space
\begin{equation}\label{eq:phasespace-transmission}
       \mathscr{X} (\Psi) = 0 \oplus \bigoplus_{k=1}^K (\B{\Hil_{k,1}}\otimes \delta_k). 
    \end{equation}

\begin{theorem}\label{thm:stab-opsys}
    For a quantum channel $\Psi:\B{\Hil}\to \B{\Hil}$, the stabilized operator system $S_{\Psi^{\infty}}$ is isomorphic, in the sense of Definition~\ref{def:op-homo}, to a $*-$algebra. More precisely,  we have
    \begin{equation*}
        S_{\Psi^{\infty}} = S_{\mathcal{P}_{\Psi}} \longrightarrow S_{\overbar{\mathcal{P}}_{\Psi}}, \quad S_{\overbar{\mathcal{P}}_{\Psi}} \longrightarrow S_{\mathcal{P}_{\Psi}}, 
    \end{equation*}
    where $\mathcal{P}_{\Psi}:\B{\Hil}\to \B{\Hil}$ and $\overbar{\mathcal{P}}_{\Psi}:\B{\Hil_0^{\perp}}\to \B{\Hil_0^{\perp}}$ are the peripheral projection channels from Section~\ref{sec:periphery}. Furthermore, let $\Hil = \Hil_0 \oplus \Hil_0^{\perp} = \Hil_0 \oplus \bigoplus_k \Hil_{k,1} \otimes \Hil_{k,2}$ be the decomposition such that the peripheral space $\mathscr{X} (\Psi)$ assumes the block structure $\mathscr{X}(\Psi) = 0 \oplus \bigoplus_{k} (\B{\Hil_{k,1}}\otimes \delta_k)$. Then,
    \begin{equation*}
        S_{\overbar{\mathcal{P}}_{\Psi}} = \bigoplus_{k} \left( \mathbb{I}_{k,1} \otimes \B{\Hil_{k,2}} \right) \subseteq \B{\Hil_0^{\perp}}.
    \end{equation*}
\end{theorem}
\begin{proof}
    Recall from Lemma~\ref{lemma:PPar-op} that $S_{\mathcal{P}_{\Psi}} \longrightarrow S_{\overbar{\mathcal{P}}_{\Psi}}$ and $ S_{\overbar{\mathcal{P}}_{\Psi}} \longrightarrow S_{\mathcal{P}_{\Psi}}$. 
    Moreover, recall from the same remark that $\overbar{\mathcal{P}}_{\Psi}$ can be written as a direct sum of identity and replacer channels:
    \begin{align}
    \overbar{\mathcal{P}}_{\Psi} &= \bigoplus_k \id_{k,1} \otimes \mathcal{R}_{k,2},
    \end{align}
    which shows that $S_{\overbar{\mathcal{P}}_{\Psi}}$ is a $*-$algebra:
    \begin{equation}
        S_{\overbar{\mathcal{P}}_{\Psi}} = \bigoplus_{k} \left( \iden_{k,1} \otimes \B{\Hil_{k,2}} \right).
    \end{equation} 
    Here, we used the fact that $S_{\id_{k,1}}=\mathbb{C}\iden_{k,1}$ and $S_{\mathcal{R}_{k,2}}=\B{\Hil_{k,2}}$. It remains to prove $S_{\Psi^{\infty}}=S_{\mathcal{P}_{\Psi}}$. Recall that there exists a channel $\mathcal{R}:\B{\Hil}\to \B{\Hil}$ such that (Lemma~\ref{lemma:reverse}): $\mathcal{R}\circ \Psi = \mathcal{P}_{\Psi}$, which implies that $\mathcal{R}^l\circ \Psi^l = \mathcal{P}_{\Psi}$ for all $l\in \mathbb{N}$. Then, Lemma~\ref{lemma:op-homo} shows that
    \begin{equation}
      \forall l\in\mathbb{N}: \quad  S_{\Psi^l} \subseteq S_{\mathcal{P}_{\Psi}}.
    \end{equation}
    For the reverse inclusion, let $L\leq d^2$ from Lemma~\ref{lemma:op-chain} be such that $S_{\Psi^{\infty}}=S_{\Psi^L}$. Let $\Hil_E$ be a common Stinespring dilation space for $\Psi^l$ for all $l\in \mathbb{N}$. Then, we can write $S_{\Psi^l} = [(\Psi^l)_c]^*(\B{\Hil_E})$. Hence, for all $k\in \mathbb{N}$, we get $S_{\Psi^{\infty}}= [(\Psi^{L+k})_c]^*(\B{\Hil_E}) $. Choose a subsequence $(k_i)_{i\in \mathbb{N}}$ such that $\lim_{i\to \infty}\Psi^{L+k_i} = \mathcal{P}_{\Psi}$ (see Lemma~\ref{lemma:reverse}). This implies that $\lim_{i\to \infty} [(\Psi^{L+k_i})_c]^* = [(\mathcal{P}_{\Psi})_c]^*$. Thus, 
    \begin{equation}
        \forall X\in \B{\Hil_E}: \quad [(\mathcal{P}_{\Psi})_c]^*(X) =  \lim_{i\to \infty} [(\Psi^{L+k_i})_c]^*(X) \in S_{\Psi^{\infty}}, 
    \end{equation}
    which proves that $[(\mathcal{P}_{\Psi})_c]^*(\B{\Hil_E}) = S_{\mathcal{P}_{\Psi}} \subseteq S_{\Psi^{\infty}}$. Hence, we obtain $S_{\Psi^{\infty}}=S_{\mathcal{P}_{\Psi}}$.

\end{proof}

\subsubsection{Capacity formulas}

With the structure theorem for the stabilized operator system in hand, we are ready to derive expressions for the zero-error capacities of highly Markovian divisible channels. Recall that $C_0(\Phi)$ denotes the one-shot zero-error classical capacity (Definition~\ref{def:classical-protocol}) and 
\begin{equation}
    C_{\operatorname{zero}}(\Phi)=\lim_{n\to \infty} \frac{C_0(\Phi^{\otimes n})}{n}
\end{equation}
is the asymptotic zero-error capacity of a channel $\Phi$ (Definition~\ref{def:zero-error-capacity}).

\begin{theorem}\label{theorem:main-zero}
Let $\Phi:\B{\Hil}\to \B{\Hil}$ be an $l$-Markovian divisible channel with $l\geq (\dim\Hil)^2$. Then, its operator system is isomorphic to a $*-$algebra: $S_{\Phi}=S_{\mathcal{P}_{\Phi}}\longrightarrow S_{\overbar{\mathcal{P}}_{\Phi}}$ and $S_{\overbar{\mathcal{P}}_{\Phi}} \longrightarrow S_{\mathcal{P}_{\Phi}}$. Consequently,   
 \begin{align*}
   C_0(\Phi) &= \log \left(\sum_k d_k \right),  \\
 C^{\operatorname{ea}} _0(\Phi) &= \log \left(\sum_k d^2_k \right),  \\
Q_{0}(\Phi)= C^{\operatorname{p}}_0(\Phi) &= \log \left(\max_k d_k \right), 
 \end{align*}
 where $d_k = \dim \Hil_{k,1}$ for $k=1,2,\ldots ,K$ are the block dimensions in the decomposition of $\mathscr{X} (\Phi)$.
 In addition, for any other $l'$-Markovian divisible channel $\Gamma : \B{\mathcal{K}} \to \B{\mathcal{K}}$ with $l'\geq (\dim\mathcal{K})^2$, 
 \begin{align*}
    C_0 (\Phi \otimes \Gamma) &= C_0(\Phi) + C_0(\Gamma), \\
    C^{\operatorname{ea}}_0 (\Phi \otimes \Gamma) &= C^{\operatorname{ea}}_0(\Phi) + C^{\operatorname{ea}}_0(\Gamma), \\
    C^{\operatorname{p}}_0(\Phi \otimes \Gamma) &= C^{\operatorname{p}}_0(\Phi) + C^{\operatorname{p}}_0(\Gamma), \\
     Q_0(\Phi \otimes \Gamma) &= Q_0(\Phi) + Q_0(\Gamma).
 \end{align*}
\end{theorem}
\begin{proof}

We use the notation from Theorem~\ref{thm:stab-opsys}. Firstly, since $\Phi$ is $l$-divisible, there exists a channel $\Psi:\B{\Hil}\to \B{\Hil}$ such that $\Phi=\Psi^l$. Since $l\geq(\dim \Hil)^2$, Lemma~\ref{lemma:op-chain} and Theorem~\ref{thm:stab-opsys} show
\begin{equation}
        S_{\Phi} = S_{\Psi^l} = S_{\Psi^{\infty}} = S_{\Phi^{\infty}} = S_{\mathcal{P}_{\Phi}} \longrightarrow S_{\overbar{\mathcal{P}}_{\Phi}}, \quad S_{\overbar{\mathcal{P}}_{\Phi}} \longrightarrow S_{\Phi}.
\end{equation}
Hence, using Theorem~\ref{thm:DSW} and Lemma~\ref{lemma:op-bottleneck}, we obtain
\begin{equation}
    C_0(\Phi) = \log \alpha(S_{\overbar{\mathcal{P}}_{\Phi}}) = \log(\sum_k \dim \Hil_{k,1}),
\end{equation}
where we used the fact that $S_{\overbar{\mathcal{P}}_{\Phi}} = \bigoplus_{k} \left( \iden_{k,1} \otimes \B{\Hil_{k,2}} \right) \subseteq \B{\Hil_0^{\perp}}$ is a $*-$algebra (Theorem~\ref{thm:stab-opsys}), so that Lemma~\ref{lemma:alpha-algebra} gives the expression for its independence number.

Similarly, if $\Gamma : \B{\mathcal{K}} \to \B{\mathcal{K}}$ is $l'$-divisible with $l'\geq (\dim\Kil)^2$, we get
\begin{equation}
        S_{\Gamma} \longrightarrow S_{\overbar{\mathcal{P}}_{\Gamma}}, \quad S_{\overbar{\mathcal{P}}_{\Gamma}} \longrightarrow S_{\Gamma},
\end{equation}
and using Theorem~\ref{thm:DSW} and Lemma~\ref{lemma:op-bottleneck}, 
\begin{equation}
    C_0(\Gamma)=\log \alpha(S_{\overbar{\mathcal{P}}_{\Gamma}}) = \log(\sum_j \dim \Kil_{j,1}),
\end{equation}
where we used the fact that $S_{\overbar{\mathcal{P}}_{\Gamma}} = \bigoplus_{j} \left( \iden_{j,1} \otimes \B{\Kil_{j,2}} \right) \subseteq \B{\Kil_0^{\perp}}$ is a $*-$algebra (Theorem~\ref{thm:stab-opsys}).

Finally, it follows that (see Lemma~\ref{lemma:op-homo-2})
\begin{align}
    S_{\Phi} \otimes S_{\Gamma} \longrightarrow S_{\overbar{\mathcal{P}}_{\Phi}} \otimes S_{\overbar{\mathcal{P}}_{\Gamma}}, \quad S_{\overbar{\mathcal{P}}_{\Phi}} \otimes S_{\overbar{\mathcal{P}}_{\Gamma}} \longrightarrow S_{\Phi} \otimes S_{\Gamma}.
\end{align}
Hence, using the multiplicativity of independence numbers of $*-$algebras from Lemma~\ref{lemma:alpha-algebra}, we obtain
\begin{align}
    C_0(\Phi \otimes \Gamma) = \log \alpha(S_{\overbar{\mathcal{P}}_{\Phi}} \otimes S_{\overbar{\mathcal{P}}_{\Gamma}}) = \log \alpha(S_{\overbar{\mathcal{P}}_{\Phi}} ) + \log \alpha(S_{\overbar{\mathcal{P}}_{\Gamma}} ) = C_0(\Phi) + C_0(\Gamma).
\end{align}

The proofs for the other capacities follow similarly.
\end{proof}

\begin{remark}
    We should note that the results obtained in this section were implicit in \cite{singh2024zero, fawzi2024error}. In particular, the stabilization of the operator systems of Markovian semigroups $(\Psi^l)_{l\in \mathbb{N}}$ (Lemma~\ref{lemma:op-chain}) was obtained in \cite{singh2024zero}, and the additivity of the one-shot zero-error capacities in the $l\to \infty$ limit was noted in \cite{fawzi2024error}. In this section, we essentially combine these two results by introducing the notion of the stabilized operator system, which allows us to lift the additivity of the capacities from the $l\to \infty$ limit to a finite $l\geq d^2$ level.
\end{remark}

The formulas for the asymptotic capacities follow as a simple corollary of the above result.

\begin{corollary}\label{corollary:main-zero}
Let $\Phi:\B{\Hil}\to \B{\Hil}$ be a channel that is $l$-Markovian divisible with $l\geq (\dim\Hil)^2$. Then,
 \begin{align*}
   C_{\operatorname{zero}}(\Phi) &= \log \left(\sum_k d_k \right),  \\
  C^{\operatorname{ea}} _{\operatorname{zero}}(\Phi) &= \log \left(\sum_k d^2_k \right),  \\
 Q_{\operatorname{zero}}(\Phi)= P_{\operatorname{zero}}(\Phi) &= \log \left(\max_k d_k \right),
 \end{align*}
 where $d_k = \dim \Hil_{k,1}$ for $k=1,2,\ldots ,K$ are the block dimensions in the decomposition of $\mathscr{X} (\Phi)$.
 Furthermore, or any other $l'$-Markovian divisible channel $\Gamma : \B{\mathcal{K}} \to \B{\mathcal{K}}$ with $l'\geq (\dim\mathcal{K})^2$, 
 \begin{align*}
    C_{\operatorname{zero}} (\Phi \otimes \Gamma) &= C_{\operatorname{zero}}(\Phi) + C_{\operatorname{zero}}(\Gamma), \\
    C^{\operatorname{ea}}_{\operatorname{zero}} (\Phi \otimes \Gamma) &= C^{\operatorname{ea}}_{\operatorname{zero}}(\Phi) + C^{\operatorname{ea}}_{\operatorname{zero}}(\Gamma), \\
    P_{\operatorname{zero}}(\Phi \otimes \Gamma) &= P_{\operatorname{zero}}(\Phi) + P_{\operatorname{zero}}(\Gamma), \\
     Q_{\operatorname{zero}}(\Phi \otimes \Gamma) &= Q_{\operatorname{zero}}(\Phi) + Q_{\operatorname{zero}}(\Gamma).
 \end{align*}
\end{corollary}

\subsubsection{Continuous-time quantum Markov semigroups}
The class of continuous-time quantum Markov semigroups (cQMS) $(\Psi_l)_{l\geq0}$ on $\B{\Hil}$ are of the form $\Psi_l=e^{l\mathcal{L}}$, where $\mathcal{L}:\B{\Hil}\to \B{\Hil}$ is a Lindbladian \cite{Gorini1976qms, Lindblad1976qms}, \cite[Chapter 7]{Wolf2012Qtour}. It is easy to check that the peripheral space of the semigroup at any length $l>0$ can be written as
\begin{align}
     \mathscr{X}(\Psi_l) &= \operatorname{span}\{X\in \B{\Hil} : \exists \,\theta\in \mathbb{R} \,\text{ s.t. }\, \mathcal{L}(X)= i\theta X \}  \nonumber \\
     &:= \mathscr{X}((\Psi_l)_{l\geq0}).
\end{align}
Moreover, for any $k\in\mathbb{N}$ and any $l> 0$, we can write 
\begin{equation}
    \Psi_l = e^{l\mathcal{L}} = (e^{\frac{l}{k}\mathcal{L}})^k = (\Psi_{l/k})^k,
\end{equation}
so that for any $l>0$, $\Psi_l$ is $k-$Markovian divisible for all $k\in\mathbb{N}$. Hence, Theorem~\ref{theorem:main-zero} applies to such semigroups. In particular, we note that the zero-error capacities are independent of $l$:
\begin{align}
    \forall l>0: \quad C_{\operatorname{zero}}(\Psi_l) &= \log \left(\sum_k d_k \right), \nonumber \\ 
    P_{\operatorname{zero}}(\Psi_l) &= \log \left(\max_k d_k \right) = Q_{\operatorname{zero}}(\Psi_l),
\end{align}
where $d_k = \dim \Hil_{k,1}$ for $k=1,2,\ldots ,K$ are the block dimensions in the decomposition of $\mathscr{X} ((\Psi_l)_{l\geq 0})$. Moreover, the following additivity result holds for any two cQMS $(\Psi_l)_{l\geq 0}$ and $(\Gamma_l)_{l\geq 0}$:
\begin{align}
    \forall l_1,l_2>0: \quad  C_{\operatorname{zero}} (\Psi_{l_1} \otimes \Gamma_{l_2}) &= C_{\operatorname{zero}}(\Psi_{l_1}) + C_{\operatorname{zero}}(\Gamma_{l_2}), \nonumber \\
     P_{\operatorname{zero}}(\Psi_{l_1} \otimes \Gamma_{l_2}) &= P_{\operatorname{zero}}(\Psi_{l_1}) + P_{\operatorname{zero}}(\Gamma_{l_2}), \nonumber \\
     Q_{\operatorname{zero}}(\Psi_{l_1} \otimes \Gamma_{l_2}) &= Q_{\operatorname{zero}}(\Psi_{l_1}) + Q_{\operatorname{zero}}(\Gamma_{l_2}).
\end{align}

\subsection{Non-zero error setting}\label{sec:main:non-zero}

In this section, we study the capacities of dQMS $(\Psi^l)_{l\in\mathbb{N}}$ in the non-zero error setting. We first examine the capacities of the asymptotic part $\Psi_{\infty}$ and then use continuity to obtain bounds on the capacities of finite-length channels $\Psi^l$. 

\subsubsection{Classical and entanglement-assisted classical capacities}\label{subsec:non-zero-classical}

\begin{theorem}\label{thm:Cinf}
    Let $\Psi:\B{\Hil}\to \B{\Hil}$ be a quantum channel with asymptotic part $\Psi_{\infty}$. Then,
    \begin{align}
       \log \left( \sum_k d_k \right) &= C(\Psi_{\infty}) = C^{\dagger}(\Psi_{\infty}) \\
       \log \left( \sum_k d^2_k \right) &= C_{\operatorname{ea}}(\Psi_{\infty}) = C_{\operatorname{ea}}^{\dagger}(\Psi_{\infty}) .
    \end{align}
where $d_k = \dim \Hil_{k,1}$ for $k=1,2,\ldots ,K$ are the block dimensions in the decomposition of $\mathscr{X} (\Psi)$. 
\end{theorem}
\begin{proof}

    The classical and entanglement-assisted classical codes constructed in  Theorem~\ref{theorem:main-storage} show 
    \begin{align}
        \forall n\in \mathbb{N}, \forall \epsilon\in [0,1): \quad \log(\sum_k d_k) \leq C_0(\Psi_{\infty}) \leq \frac{1}{n} C_0(\Psi_{\infty}^{\otimes n}) &\leq \frac{1}{n} C_{\epsilon}(\Psi_{\infty}^{\otimes n}) \\        
        \log(\sum_k d^2_k) \leq C_0^{\operatorname{ea}}(\Psi_{\infty}) \leq \frac{1}{n} C_0^{\operatorname{ea}}(\Psi_{\infty}^{\otimes n}) &\leq \frac{1}{n} C_{\epsilon}^{\operatorname{ea}}(\Psi_{\infty}^{\otimes n}),
    \end{align}
    where the inequalities follow from superadditivity (Remark~\ref{remark:zero-super}) of the one-shot zero-error capacities. Hence, $\log(\sum_k d_k)\leq C(\Psi_{\infty})$ and $\log(\sum_k d^2_k)\leq C_{\operatorname{ea}}(\Psi_{\infty})$ (see Definition~\ref{def:capacity}). 
    
    Conversely, $C_{\operatorname{ea}}^{\dagger}(\Psi_{\infty}) = I(\Psi_{\infty}) \leq I_{\max}(\Psi_{\infty})$ (see Theorem~\ref{theorem:Cea-I}). Note that since $\Psi_{\infty}= \mathcal{P}_{\Psi}\circ \Psi$, data-processing (Lemma~\ref{lemma:channel-bottlenecks}) shows that $I_{\max}(\Psi_{\infty})\leq I_{\max}(\mathcal{P}_{\Psi})$. Recall that $\mathcal{P}_{\Psi}$ here is the projector onto the peripheral space $\mathscr{X} (\Psi)$. Finally, $I_{\max}(\mathcal{P}_{\Psi})\leq \log(\sum_k d^2_k)$ was obtained in the proof of Theorem~\ref{theorem:main-storage}.

    For the classical capacity, we use the following strong converse bound (Lemma~\ref{lemma:strong-converse})
    \begin{equation}
        C^{\dagger}(\Psi_{\infty}) \leq \limsup_{n\to \infty} \frac{1}{n} \chi_{\max} (\Psi_{\infty}^{\otimes n}).
    \end{equation}
    By again using data-processing (Lemma~\ref{lemma:channel-bottlenecks}), we write $\chi_{\max}(\Psi_{\infty})\leq \chi_{\max}({\mathcal{P}}_{\Psi})\leq \log( \sum_k d_k)$, where the last inequality was obtained in the proof of Theorem~\ref{theorem:main-storage}. Moreover, since the peripheral space is multiplicative under tensor products (Lemma~\ref{lemma:peripheral-multi}), we get
\begin{equation}
   \forall n\in \mathbb{N}: \quad \frac{1}{n} \chi_{\max} (\Psi_{\infty}^{\otimes n}) \leq \log \left( \sum_k d_k \right),
\end{equation}
which clearly implies that 
\begin{equation}
        C^{\dagger}(\Psi_{\infty}) \leq \limsup_{n\to \infty} \frac{1}{n} \chi_{\max} (\Psi_{\infty}^{\otimes n}) \leq \log \left( \sum_k d_k \right).
    \end{equation}

\end{proof}

\begin{theorem}\label{thm:Casmpy}
    Let $\Psi:\B{\Hil}\to \B{\Hil}$ be a quantum channel with $d=\dim\Hil$. Then,
    \begin{align}
       \log \left( \sum_k d_k \right) &\leq C(\Psi^l) \leq \log \left( \sum_k d_k \right) + \delta_l \log (d^2 -1 ) + 2h (\delta_l/2) \nonumber \\
       \log \left( \sum_k d_k^2 \right) &\leq C_{\operatorname{ea}}(\Psi^l) = C_{\operatorname{ea}}^{\dagger}(\Psi^l) \leq \log \left( \sum_k d^2_k \right) + \delta_l \log (d^2-1) + 2h(\delta_l/2), \nonumber
    \end{align}
    where $d_k=\dim \Hil_{k,1}$ for $k\in \{1,2,\ldots ,K \}$ are the block dimensions in the decomposition of $\chi (\Psi)$, $\delta_l=\norm{\Psi^l-\Psi^l_{\infty}}_{\diamond} \leq \kappa\mu^l\to 0$ as $l\to \infty$, where $\mu=\operatorname{spr}(\Psi-\Psi_{\infty}), \kappa$ govern the convergence as in Eq.~\eqref{eq:converge}. Here, the lower bound holds for all $l\in\mathbb{N}$ while the upper bound holds for $l$ large enough so that $\delta_l/2 \leq 1-1/d^2$.
\end{theorem}
\begin{proof}
    For the lower bound, note that $\Psi^l_{\infty} = \Psi^l \circ \mathcal{P}_{\Psi}$ for all $l\in\mathbb{N}$, so that Lemma~\ref{lemma:bottleneck-regular} shows
    \begin{align}
       \forall l\in \mathbb{N}: \quad \log (\sum_k d_k ) &= C (\Psi^l_{\infty}) \leq C(\Psi^l) \\
       \log (\sum_k d^2_k ) &= C_{\operatorname{ea}} (\Psi^l_{\infty}) \leq C_{\operatorname{ea}}(\Psi^l).
    \end{align}
    For the converse bound, note that $\norm{\Psi^l - \Psi_{\infty}^l}_{\diamond} = \delta_l$, so that we can use continuity of the channel capacity function $C$ (Theorem~\ref{theorem:cap-cont}) to write
    \begin{align}
        C(\Psi^l) &\leq C (\Psi^l_{\infty}) + \delta_l \log (d^2 -1 ) + 2h (\delta_l/2) \nonumber \\ 
        &= \log \left( \sum_k d_k \right) + \delta_l \log (d^2 -1 ) + 2h (\delta_l/2),
    \end{align}
    where we used the capacity formula for $\Psi_{\infty}$ from Theorem~\ref{thm:Cinf}. For the entanglement-assisted capacity, note that strong converse property always holds $C^{\dagger}_{\operatorname{ea}}=C_{\operatorname{ea}}$ (Theorem~\ref{theorem:Cea-I}). The rest of the proof follows identically using continuity of $C_{\operatorname{ea}}$ (Theorem~\ref{theorem:cap-cont}).
\end{proof}

\begin{remark} \label{remark:chi-cont}
    Since no continuity bounds are known for the strong converse capacity $C^{\dagger}$, the proof technique used in Theorem~\ref{thm:Casmpy} does not work for this capacity. Another way to approach the problem is via continuity analysis of the max-Holevo quantity $\chi_{\max}$ (or the similarly defined $\alpha$-sandwiched Holevo quantity $\tilde{\chi}_{\alpha}$) \cite{Wilde2014converse}, which, to the best of our knowledge, is also unexplored. Note that the Holevo quantities admit alternative expressions in terms of divergence radii \cite{Sibson1969radii, Csiszar1995radii, Mosonyi2011radii, Mosonyi2021radii}.  We leave the continuity analysis of these quantities for future study. In this regard, it would also be interesting to investigate continuity properties of the newer (and sharper) strong converse bounds for classical communication in terms of the \emph{Upsilon information} and \emph{geometric R\'enyi divergences} \cite{Wang2018classicalstrong, Wang2019classicalstrong, Fang2021geometric}.
\end{remark}

\subsubsection{Quantum and private classical capacities} \label{subsec:non-zero-quantum}

Next, we deal with the quantum and private classical capacities.

\begin{theorem}\label{thm:QPinf}
    Let $\Psi:\B{\Hil}\to \B{\Hil}$ be a quantum channel with asymptotic part $\Psi_{\infty}$. Then,
    \begin{align}
       \log \left( \max_k d_k \right) &= Q(\Psi_{\infty}) = P(\Psi_{\infty}) = Q_{\leftrightarrow}(\Psi_{\infty}) = P_{\leftrightarrow}(\Psi_{\infty}) \nonumber \\
       &= Q^{\dagger}(\Psi_{\infty}) = P^{\dagger}(\Psi_{\infty}) = Q^{\dagger}_{\leftrightarrow}(\Psi_{\infty}) = P^{\dagger}_{\leftrightarrow}(\Psi_{\infty}). \nonumber
    \end{align}
where $d_k = \dim \Hil_{k,1}$ for $k=1,2,\ldots ,K$ are the block dimensions in the decomposition of $\mathscr{X} (\Psi)$. 
\end{theorem}
\begin{proof}
    It suffices to prove that $\log(\max_k d_k )\leq Q(\Psi_{\infty})$ and $P^{\dagger}_{\leftrightarrow}(\Psi_{\infty})\leq \log (\max_k d_k)$.

    For the lower bound, the quantum code constructed in  Theorem~\ref{theorem:main-storage} shows that 
    \begin{align}
        \forall n\in \mathbb{N}, \forall \epsilon\in [0,1): \quad \log(\max_k d_k) \leq Q_0(\Psi_{\infty}) \leq \frac{1}{n} Q_0(\Psi_{\infty}^{\otimes n}) &\leq \frac{1}{n} Q_{\epsilon}(\Psi_{\infty}^{\otimes n})
    \end{align}
    where the inequalities follow from superadditivity (Remark~\ref{remark:zero-super}) of the one-shot zero-error capacity. Hence, $\log(\max_k d_k)\leq Q(\Psi_{\infty})$ follows by taking the appropriate limit (Definition~\ref{def:capacity}). 

    For the upper bound, note that $P_{\leftrightarrow}^{\dagger}(\Psi)\leq E_{\max}(\Psi)$ (Theorem~\ref{theorem:QPassisted<Emax}). Moreover, since $\Psi_{\infty}= \mathcal{P}_{\Psi}\circ \Psi$, data-processing (Lemma~\ref{lemma:channel-bottlenecks}) shows that $E_{\max}(\Psi_{\infty})\leq E_{\max}(\mathcal{P}_{\Psi})$. Recall that $\mathcal{P}_{\Psi}$ here is the projector onto the peripheral space $\mathscr{X} (\Psi)$. Finally, $E_{\max}({\mathcal{P}}_{\Psi})\leq \log(\max_k d_k)$ was obtained in the proof of Theorem~\ref{theorem:main-storage}.
\end{proof}

\begin{remark}
    It is clear from the proof of Theorems~\ref{thm:Cinf} and \ref{thm:QPinf} that the capacities of the asymptotic part $\Psi_{\infty}$ of a channel $\Psi$ are equal to those of the peripheral projection channels $\mathcal{P}_{\Psi}$ and $\overbar{\mathcal{P}}_{\Psi}$. The key property of these channels that allows us to compute their capacities is their direct sum structure in terms of the identity and replacer channels (see \cite{Fukuda2007direct, Gao2018tro}).
\end{remark}

\begin{theorem}\label{thm:QPasymp}
    Let $\Psi:\B{\Hil}\to \B{\Hil}$ be a quantum channel with $d=\dim\Hil$ and $\alpha>1$. Then,
    \begin{align}
       \log \left( \max_k d_k \right) \leq Q(\Psi^l) \leq P^{\dagger}(\Psi^l) &\leq \log \left( \max_k d_k \right) + \frac{\alpha}{\alpha-1} \log (1+ \frac{\delta_l d^{\frac{\alpha-1}{\alpha}}}{2}), \nonumber \\
       \log \left( \max_k d_k \right) \leq Q(\Psi^l) \leq P^{\dagger}_{\leftrightarrow}(\Psi^l) &\leq \log \left( \max_k d_k \right) +  \log (1+ \frac{\delta_l d}{2}), \nonumber
    \end{align}
 where $d_k=\dim \Hil_{k,1}$ for $k\in \{1,2,\ldots ,K \}$ are the block dimensions in the decomposition of $\chi (\Psi)$ and $\delta_l=\norm{\Psi^l-\Psi^l_{\infty}}_{\diamond} \leq \kappa\mu^l\to 0$ as $l\to \infty$, where $\mu=\operatorname{spr}(\Psi-\Psi_{\infty}), \kappa$ govern the convergence (Eq.~\eqref{eq:converge}). Here, the lower bound holds for all $l\in\mathbb{N}$ while the upper bound holds for $l$ large enough so that $\delta_l < 2$.
\end{theorem}
\begin{proof}
    Note that $\Psi^l_{\infty} = \Psi^l \circ \mathcal{P}_{\Psi}$ for all $l\in \mathbb{N}$, so that we can use Lemma~\ref{lemma:bottleneck-regular} to write
    \begin{equation}
       \forall l\in \mathbb{N}: \quad \log (\max_k d_k ) = Q (\Psi^l_{\infty}) \leq Q(\Psi^l).
    \end{equation}
    For the converse bound, note that $P^{\dagger}(\Psi^l) \leq \widetilde E_{\alpha}(\Psi^l)$ (Lemma~\ref{lemma:strong-converse}) and $\norm{\Psi^l - \Psi_{\infty}^l}_{\diamond}= \delta_l$. Hence, we can use continuity of $\widetilde E_{\alpha}$ (Lemma~\ref{lemma:Emax-cont}) to write
    \begin{align}
        \widetilde E_{\alpha}(\Psi^l) &\leq \widetilde E_{\alpha} (\Psi^l_{\infty}) + \frac{\alpha}{\alpha-1} \log (1+ \frac{\delta_l d^{\frac{\alpha-1}{\alpha}}}{2}) \nonumber \\ 
        &\leq \widetilde E_{\alpha}(\mathcal{P}_{\Psi}) + \frac{\alpha}{\alpha-1} \log (1+\frac{\delta_l d^{\frac{\alpha-1}{\alpha}}}{2}) \nonumber \\ 
        &\leq E_{\max}(\mathcal{P}_{\Psi}) + \frac{\alpha}{\alpha-1} \log (1+\frac{\delta_l d^{\frac{\alpha-1}{\alpha}}}{2}) \\
        &\leq \log (\max_k d_k) + \frac{\alpha}{\alpha-1} \log (1+\frac{\delta_l d^{\frac{\alpha-1}{\alpha}}}{2}),
    \end{align}
    where the second inequality follows from data-processing (since $\Psi^l_{\infty}=\Psi^l_{\infty}\circ \mathcal{P}_{\Psi}$, see Lemma~\ref{lemma:channel-bottlenecks}) and the final inequality follows because $E_{\max}( \mathcal{P}_{\Psi}) \leq \log(\max_k d_k)$ (see the proof of Theorem~\ref{theorem:main-storage}). The converse bound for $P^{\dagger}_{\leftrightarrow}$ follows similarly since $P_{\leftrightarrow}^{\dagger}(\Psi^l)\leq E_{\max}(\Psi^l)$ (Theorem~\ref{theorem:QPassisted<Emax}). Here, we can use continuity of $E_{\max}$ (Lemma~\ref{lemma:Emax-cont}) and follow exactly the same steps as above.
\end{proof}

\subsubsection{Finite block-length bounds on quantum and private classical capacities}\label{subsec:non-zero-finiteblock}

Recall that the quantum capacity $Q(\Phi)$ of a noisy channel $\Phi$ quantifies the optimal rate at which quantum information can be transmitted via $n$ uses of $\Phi$ such that the error incurred in transmission vanishes as the number of channel uses $n\to \infty$: (Definition~\ref{def:capacity})
\begin{equation}
    Q(\Phi) = \inf_{\epsilon\in (0,1)} \liminf_{n\to \infty} \frac{Q_{\epsilon}(\Phi^{\otimes n})}{n}.
\end{equation}

However, in practice, resources are finite and the number of channel uses are limited. Hence, from a realistic viewpoint, it is more important to analyze how well 
information can be transmitted over a finite number of channel uses by using encoders and decoders (see Definitions~\ref{def:classical-protocol}-\ref{def:quantum-protocol}) with fixed finite sizes. In this section, we adopt this perspective and obtain bounds on the one-shot quantum and private classical capacities of dQMS $(\Psi^l)_{l\in\mathbb{N}}$ in the finite blocklength regime.

\begin{theorem} \label{theorem:QPfinite-alpha}
    Let $\Psi:\B{\Hil}\to \B{\Hil}$ be a quantum channel, $n\in \mathbb{N}$ and $\epsilon\in [0,1)$. Then, for all $\alpha>1$,
    \begin{align*}
        \log(\max_k d_k) \leq \frac{1}{n} Q_{\epsilon}((\Psi^l)^{\otimes n})) &\leq \log(\max_k d_k) +  \frac{\alpha}{\alpha-1} \log (1+\frac{\delta_l d^{\frac{\alpha-1}{\alpha}}}{2}) + \frac{\alpha}{n(\alpha-1)} \log (\frac{n^{d^2}}{1-\epsilon}), \\
        \log(\max_k d_k) \leq \frac{1}{n} C^{\operatorname{p}}_{\epsilon}((\Psi^l)^{\otimes n})) &\leq \log(\max_k d_k) +  \frac{\alpha}{\alpha-1} \log (1+\frac{\delta_l d^{\frac{\alpha-1}{\alpha}}}{2}) + \frac{\alpha }{n(\alpha-1)} \log (\frac{n^{d^2}}{1-\epsilon}),
    \end{align*}
    where $d_k=\dim \Hil_{k,1}$ for $k\in \{1,2,\ldots ,K \}$ are the block dimensions in the decomposition of $\chi (\Psi)$ and $\delta_l=\norm{\Psi^l-\Psi^l_{\infty}}_{\diamond} \leq \kappa\mu^l\to 0$ as $l\to \infty$ where $\mu=\operatorname{spr}(\Psi-\Psi_{\infty}), \kappa$ govern the convergence (Eq.~\eqref{eq:converge}). Here, the lower bounds hold for all $l\in\mathbb{N}$ while the upper bound holds for $l$ large enough so that $\delta_l < 2$.
\end{theorem}
\begin{proof}
    The codes constructed in  Theorem~\ref{theorem:main-storage} show that for all $l\in \mathbb{N}$, $n\in \mathbb{N}$ and $\epsilon\in [0,1)$, we have 
    \begin{align}
    \log(\max_k d_k) \leq Q_0(\Psi^l) \leq \frac{1}{n} Q_0((\Psi^l)^{\otimes n}) &\leq \frac{1}{n} Q_{\epsilon}((\Psi^l)^{\otimes n}) \\
    \log(\max_k d_k) \leq C^{\operatorname{p}}_0(\Psi^l) \leq \frac{1}{n} C^{\operatorname{p}}_0((\Psi^l)^{\otimes n}) &\leq \frac{1}{n} C^{\operatorname{p}}_{\epsilon}((\Psi^l)^{\otimes n})
    \end{align}
    where the inequalities follow from superadditivity (Remark~\ref{remark:zero-super}) of the one-shot zero-error capacity. For the upper bounds, we can do the following calculation:
    \begin{align}
    \frac{1}{n} Q_{\epsilon}((\Psi^l)^{\otimes n})) &\leq \frac{1}{n} \widetilde E_{\alpha}((\Psi^l)^{\otimes n}) + \frac{\alpha}{n(\alpha-1)} \log(\frac{1}{1-\epsilon}) \nonumber \\
    &\leq \widetilde E_{\alpha}(\Psi^l) + \frac{\alpha d^2}{\alpha-1} \frac{\log n}{n} + \frac{\alpha }{n(\alpha-1)} \log (\frac{1}{1-\epsilon}) \nonumber \\
    &\leq \widetilde E_{\alpha}(\Psi^l_{\infty}) + \frac{\alpha}{\alpha-1} \log (1+\frac{\delta_l d^{\frac{\alpha-1}{\alpha}}}{2}) + \frac{\alpha }{n(\alpha-1)} \log (\frac{n^{d^2}}{1-\epsilon})  \nonumber \\
    &\leq \log(\max_k d_k) +  \frac{\alpha}{\alpha-1} \log (1+\frac{\delta_l d^{\frac{\alpha-1}{\alpha}}}{2}) + \frac{\alpha }{n(\alpha-1)} \log (\frac{n^{d^2}}{1-\epsilon}).  
\end{align}
Here, the first inequality follows from the converse bound in Theorem~\ref{lemma:QP<=E}, the second inequality follows from weak subadditivity of $\widetilde E_{\alpha}$ \cite[Theorem 15]{Wilde2017private} \cite[Theorem 6]{Tomamichel2017strong}, the third inequality follows from continuity of $\widetilde E_{\alpha}$ (Lemma~\ref{lemma:Emax-cont}), and the final inequality $\widetilde E_{\alpha}(\Psi^l_{\infty})\leq \log(\max_k d_k)$ was shown in Theorem~\ref{thm:QPasymp}. 
\end{proof}

\begin{theorem} \label{theorem:QPfinite-max}
    Let $\Psi:\B{\Hil}\to \B{\Hil}$ be a quantum channel, $n\in \mathbb{N}$ and $\epsilon\in [0,1)$. Then, 
    \begin{align*}
        \log(\max_k d_k) \leq \frac{1}{n} Q_{\epsilon}((\Psi^l)^{\otimes n})) &\leq \log(\max_k d_k) +  \log (1+\frac{\delta_l d}{2}) + \frac{1}{n} \log (\frac{1}{1-\epsilon}), \\
        \log(\max_k d_k) \leq \frac{1}{n} C^{\operatorname{p}}_{\epsilon}((\Psi^l)^{\otimes n})) &\leq \log(\max_k d_k) +   \log (1+\frac{\delta_l d}{2}) + \frac{1}{n} \log (\frac{1}{1-\epsilon}),
    \end{align*}
    where $d_k=\dim \Hil_{k,1}$ for $k\in \{1,2,\ldots ,K \}$ are the block dimensions in the decomposition of $\chi (\Psi)$ and $\delta_l=\norm{\Psi^l-\Psi^l_{\infty}}_{\diamond} \leq \kappa\mu^l\to 0$ as $l\to \infty$, where $\mu=\operatorname{spr}(\Psi-\Psi_{\infty}), \kappa$ govern the convergence (Eq.~\eqref{eq:converge}). Here, the lower bound holds for all $l\in\mathbb{N}$ while the upper bound holds for $l$ large enough so that $\delta_l < 2$.
\end{theorem}
\begin{proof}
    The lower bounds were obtained in Theorem~\ref{theorem:QPfinite-alpha}. For the converse bounds, we follow the same steps as in Theorem~\ref{theorem:QPfinite-alpha}, with $\widetilde E_{\alpha}$ replaced with $E_{\max}$:
\begin{align}
   \frac{1}{n} Q_{\epsilon}((\Psi^{\otimes n})^l) &\leq \frac{1}{n}E_{\max}((\Psi^l)^{\otimes n}) + \frac{1}{n} \log(\frac{1}{1-\epsilon}) \nonumber \\
    &\leq E_{\max}(\Psi^l) + \frac{1}{n}\log(\frac{1}{1-\epsilon}) \nonumber \\
    &\leq E_{\max}(\Psi^l_{\infty}) + \log(1+ \frac{\delta_l d}{2}) + \frac{1}{n}\log(\frac{1}{1-\epsilon}) \nonumber \\
    &\leq \log(\max_k d_k) +  \log(1+ \frac{\delta_l d}{2}) + \frac{1}{n}\log(\frac{1}{1-\epsilon}).
\end{align}
Here, the first inequality follows from Theorem~\ref{lemma:QP<=E}, the second inequality follows from the subadditivity of $E_{\max}$ \cite{Christandl2017max, Berta2018amort}, the third inequality follows from the continuity of $E_{\max}$ (Lemma~\ref{lemma:Emax-cont}), and the final inequality $E_{\max}(\Psi^l_{\infty})\leq \log(\max_k d_k)$ was shown in Theorem~\ref{thm:QPasymp}.
\end{proof}

Observe that the $n\to \infty$ limit of the bounds obtained in Theorems~\ref{theorem:QPfinite-alpha} and \ref{theorem:QPfinite-max} yield the asymptotic bounds in Theorem~\ref{thm:QPasymp}. Finally, we note that since $E_{\max}$ also serves as a converse bound on the LOCC-assisted capacities (Theorem~\ref{theorem:QPassisted<Emax}), the same techniques as above work to prove bounds on the assisted capacities.

\begin{theorem} \label{theorem:QPfinite-assisted}
    Let $\Psi:\B{\Hil}\to \B{\Hil}$ be a quantum channel, $n\in \mathbb{N}$ and $\epsilon\in [0,1)$. Then, 
    \begin{align*}
        \log(\max_k d_k) \leq \frac{1}{n}Q^{\leftrightarrow}_{n,\epsilon}(\Psi^l) \leq \frac{1}{n}P^{\leftrightarrow}_{n,\epsilon}(\Psi^l)  &\leq \log(\max_k d_k) +  \log (1+\frac{\delta_l d}{2}) + \frac{1}{n} \log (\frac{1}{1-\epsilon}),
    \end{align*}
    where $d_k=\dim \Hil_{k,1}$ for $k\in \{1,2,\ldots ,K \}$ are the block dimensions in the decomposition of $\chi (\Psi)$ and $\delta_l=\norm{\Psi^l-\Psi^l_{\infty}}_{\diamond} \leq \kappa\mu^l\to 0$ as $l\to \infty$, where $\mu=\operatorname{spr}(\Psi-\Psi_{\infty}), \kappa$ govern the convergence (Eq.~\eqref{eq:converge}). Here, the lower bound holds for all $l\in\mathbb{N}$ while the upper bound holds for $l$ large enough so that $\delta_l < 2$.
\end{theorem}

\subsubsection{Strong additivity}\label{subsec:strong}

Next, we analyze the additivity of the capacities of dQMS $(\Psi^l)_{l\in \mathbb{N}}$. We begin by showing that the asymptotic part $\Psi_{\infty}$ of any channel is strongly additive.

\begin{theorem}\label{thm:asym-strongadd}
    Let $\Psi:\B{\Hil_A}\to \B{\Hil_A}$ and $\Gamma:\B{\Hil_B}\to \B{\Hil_C}$ be arbitrary channels. Then,
    \begin{align*}
        Q(\Psi_{\infty}\otimes \Gamma) &= Q(\Psi_{\infty}) + Q(\Gamma),  \\
        P(\Psi_{\infty}\otimes \Gamma) &= P(\Psi_{\infty}) + P(\Gamma), \\
        C(\Psi_{\infty}\otimes \Gamma) &= C(\Psi_{\infty}) + C(\Gamma).
    \end{align*}
\end{theorem}
\begin{proof}
    
    It suffices to prove the theorem for the single letter quantities $I_c$, $I_p$, $\chi$ and the stated result would then follow from regularization (Theorem~\ref{theorem:LSD+CP}). Throughout the proof, we work with the decomposition $\Hil_A = \Hil_0 \oplus \bigoplus_k \Hil_{k,1}\otimes \Hil_{k,2}=\Hil_0 \oplus \Hil^{\perp}_0$ of the underlying Hilbert space, with respect to which the peripheral space $\mathscr{X}(\Psi)$ assumes the stated decomposition $\mathscr{X} (\Psi) = 0 \oplus \bigoplus_k (\B{\Hil_{k,1}}\otimes \delta_k)$. 
    
    Firstly, note that since the coherent information is super-additivite, we can write
    \begin{equation}
      \log (\max_k d_k) + I_c (\Gamma) \leq I_c (\Psi_{\infty}) + I_c (\Gamma) \leq I_c (\Psi_{\infty} \otimes \Gamma),
     \end{equation}
     where we used $\log(\max_k d_k )\leq I_c(\Psi_{\infty})$ (see Theorem~\ref{thm:QPinf}).

    To prove the opposite inequality, recall that $\Psi_{\infty}=\mathcal{P}_{\Psi}\circ \Psi_{\infty}=\Psi_{\infty}\circ \mathcal{P}_{\Psi}$, where $\mathcal{P}_{\Psi}:\B{\Hil_A}\to \B{\Hil_A}$ projects onto the peripheral space $\mathscr{X} (\Psi)$. Then, by using data-processing (Lemma~\ref{lemma:channel-bottlenecks}), we obtain
    \begin{equation}
        I_c (\Psi_{\infty} \otimes \Gamma) = I_c ( (\Psi_{\infty} \otimes \id_{C\to C}) \circ ( \mathcal{P}_{\Psi}\otimes \Gamma)) \leq I_c (\mathcal{P}_{\Psi} \otimes \Gamma).
    \end{equation}
Another data-processing argument similar to what is used in Lemma~\ref{lemma:PPar} shows
\begin{equation}
        I_c (\mathcal{P}_{\Psi} \otimes \Gamma) \leq I_c (\overbar{\mathcal{P}}_{\Psi} \otimes \Gamma),
\end{equation} 
where $\overbar{\mathcal{P}}_{\Psi}:\B{\Hil_0^{\perp}}\to \B{\Hil_0^{\perp}}$ is the restricted projection channel from Section~\ref{sec:periphery}. Recall from the same remark that we can write the action of $\overbar{\mathcal{P}}_{\Psi}$ and $\overbar{\mathcal{P}}_{\Psi}\otimes \Gamma$ as follows:
\begin{align}
    \overbar{\mathcal{P}}_{\Psi} &= \bigoplus_k \id_{k,1} \otimes \mathcal{R}_{k,2} \nonumber \\ \overbar{\mathcal{P}}_{\Psi} \otimes \Gamma &= \bigoplus_k \id_{k,1} \otimes \mathcal{R}_{k,2} \otimes \Gamma,
\end{align}
    where for each $k$, $\id_{k,1}:\B{\Hil_{k,1}}\to \B{\Hil_{k,1}}$ is the identity channel and $\mathcal{R}_{k,2}:\B{\Hil_{k,2}}\to \B{\Hil_{k,2}}$ is the replacer channel defined as $\mathcal{R}_{k,2}(X) = \Tr(X)\delta_k$. From the formulas of capacities of direct sum channels in \cite{Fukuda2007direct}, it follows that 
    \begin{align}
        I_c (\overbar{\mathcal{P}}_{\Psi} \otimes \Gamma) &= \max_k I_c \left( \id_{k,1} \otimes \mathcal{R}_{k,2} \otimes \Gamma \right) \nonumber  \\ 
        &= \max_k \left( \log d_k + I_c (\mathcal{R}_{k,2} \otimes \Gamma ) \right) \nonumber \\
        &= \log (\max_k d_k) + I_c (\Gamma),
    \end{align}
    where latter two equalities follow from the fact that both the identity and replacer channels are strongly additive (Appendix~\ref{appen:additive}), with $I_c (\id_{k,1})=\log d_k$ and $I_c (\mathcal{R}_{k,2})=0$ for all $k$ (see also \cite{Gao2018tro}). Retracing our steps, we have the following chain of inequalities:
    \begin{align}
        I_c (\Psi_{\infty} \otimes \Gamma) \leq I_c (\mathcal{P}_{\Psi} \otimes \Gamma) \leq I_c (\overbar{\mathcal{P}}_{\Psi} \otimes \Gamma) = \log (\max_k d_k) + I_c (\Gamma),
    \end{align}
    which completes the proof for the quantum capacity. The proofs for $\chi$, $I_p$ follow exactly the same steps, since the identity and replacer channels are also strongly additive for $\chi$, $I_p$ (Appendix~\ref{appen:additive}).
\end{proof}

Finally, we can lift the strong additivity from the asymptotic part $\Psi_{\infty}$ to finite length by using continuity of channel capacities as shown below.

\begin{theorem}\label{theorem:strong-add}
    Let $\Psi:\B{\Hil_A}\to \B{\Hil_A}$ and $\Gamma:\B{\Hil_B}\to \B{\Hil_C}$ be quantum channels. Then,
    \begin{align*}
       \log \left( \sum_k d_k \right) + C(\Gamma) \leq C(\Psi^l &\otimes \Gamma) \leq \log\left( \sum_k d_k \right) + C(\Gamma) + \delta_l \log (d^2_A d^2_C -1 ) + 2h (\delta_l/2), \\
       \log(\max_k d_k) + P(\Gamma) \leq P(\Psi^l &\otimes \Gamma) \leq \log(\max_k d_k) + P(\Gamma) + 2\delta_l \log (d^2_A d^2_C -1 ) + 4h (\delta_l/2), \\
       \log(\max_k d_k) + Q(\Gamma) \leq Q(\Psi^l &\otimes \Gamma) \leq \log(\max_k d_k) + Q(\Gamma) + \delta_l \log (d^2_A d^2_C -1 ) + 2h (\delta_l/2),
    \end{align*}
    where $d_k=\dim \Hil_{k,1}$ for $k\in \{1,2,\ldots ,K \}$ are the block dimensions in the decomposition of $\chi (\Psi)$ and $\delta_l=\norm{\Psi^l-\Psi^l_{\infty}}_{\diamond} \leq \kappa\mu^l\to 0$ as $l\to \infty$, where $\mu=\operatorname{spr}(\Psi-\Psi_{\infty}), \kappa$ govern the convergence. The lower bound holds for all $l\in\mathbb{N}$ while the upper bound holds for $l$ large enough so that $\delta_l/2 \leq 1 - 1/(d_A^2d_C^2)$.
\end{theorem}
\begin{proof}
    The lower bounds follow from the superadditivity of channel capacities along with the capacity estimates from Theorems~\ref{thm:QPasymp} and \ref{thm:Casmpy}.

    For the upper bound, note that $\norm{\Psi^l \otimes \Gamma - \Psi^l_{\infty} \otimes \Gamma}_{\diamond} \leq \norm{\Psi^l-\Psi^l_{\infty}}_{\diamond} = \delta_l \leq \kappa\mu^l $, so that we can use continuity of the channel capacity function $C$ (Theorem~\ref{theorem:cap-cont}) to write 
    \begin{align*}
       C(\Psi^l \otimes \Gamma) &\leq C(\Psi^l_{\infty}\otimes \Gamma) +  \delta_l \log (d^2_A d^2_C -1 ) + 2h (\delta_l/2) \\
       &= C(\Psi^l_{\infty}) + C(\Gamma) + \delta_l \log (d^2_A d^2_C -1 ) + 2h (\delta_l/2) \\
       &=\log\left(\sum_k d_k\right) + C(\Gamma) + \delta_l \log (d^2_A d^2_C -1 ) + 2h (\delta_l/2),
    \end{align*}
    where we used the strong additivity of $\Psi_{\infty}$ from Theorem~\ref{thm:asym-strongadd} along with the capacity formula from Theorem~\ref{thm:Cinf}. The proofs for private and quantum capacities follow similarly.
\end{proof}

We can reformulate the above strong additivity results in the language of \emph{potential} capacities \cite{Winter2016potential}. The potential capacity of a channel $\Phi$ quantifies the maximum possible capability of a channel to transmit information when it is used in combination with any other contextual channel. More precisely, for a channel $\Phi:\B{\Hil_A}\to \B{\Hil_B}$, we define its $q$-dimensional potential  classical, private classical, and quantum capacity, respectively, as follows: 

\begin{align}
    C^{(q)}_{\operatorname{pot}}(\Phi) := \sup_{\Gamma} \left[C(\Phi\otimes \Gamma) - C(\Gamma) \right], \\
    P^{(q)}_{\operatorname{pot}}(\Phi) := \sup_{\Gamma} \left[P(\Phi\otimes \Gamma) - P(\Gamma) \right], \\
    Q^{(q)}_{\operatorname{pot}}(\Phi) := \sup_{\Gamma} \left[Q(\Phi\otimes \Gamma) - Q(\Gamma) \right],
\end{align}
where the supremum is over all contextual channels $\Gamma:\B{\Hil_C}\to \B{\Hil_D}$ with fixed output dimension $q=d_D=\dim\Hil_D$. The potential capacities defined in \cite{Winter2016potential} are obtained by taking a further supremum over $q\in \mathbb{N}$. We can now restate Theorem~\ref{theorem:strong-add} as follows.

\begin{theorem}\label{theorem:stradd}
    Let $\Psi:\B{\Hil_A}\to \B{\Hil_A}$ a quantum channel and $q\in \mathbb{N}$. Then,
    \begin{align*}
       \log \left( \sum_k d_k \right) \leq C^{(q)}_{\operatorname{pot}}(\Psi^l) &\leq \log\left( \sum_k d_k \right) + \delta_l \log (q^2d^2_A -1 ) + 2h (\delta_l/2), \\
       \log(\max_k d_k) \leq P^{(q)}_{\operatorname{pot}}(\Psi^l) &\leq \log(\max_k d_k) +  2\delta_l \log (q^2d^2_A -1 ) + 4h (\delta_l/2), \\
       \log(\max_k d_k) \leq Q^{(q)}_{\operatorname{pot}}(\Psi^l) &\leq \log(\max_k d_k) + \delta_l \log (q^2d^2_A -1 ) + 2h (\delta_l/2),
    \end{align*} 
    where $d_k=\dim \Hil_{k,1}$ for $k\in \{1,2,\ldots ,K \}$ are the block dimensions in the decomposition of $\chi (\Psi)$ and $\delta_l=\norm{\Psi^l-\Psi^l_{\infty}}_{\diamond} \leq \kappa\mu^l\to 0$ as $l\to \infty$, where $\mu=\operatorname{spr}(\Psi-\Psi_{\infty}), \kappa$ govern the convergence (see Eq.~\eqref{eq:converge}). The lower bound holds for all $l\in\mathbb{N}$ while the upper bound holds for $l$ large enough so that $\delta_l/2 \leq 1 - 1/(qd_A)^2$.
\end{theorem}

\subsubsection{Rate of convergence}\label{subsec:convergence-transmission}

The finite length capacity bounds from Theorems~\ref{thm:Casmpy}, \ref{thm:QPasymp}, \ref{theorem:strong-add} immediately yield the following infinite-length capacity formulas.

\begin{corollary}
    Let $\Psi:\B{\Hil}\to \B{\Hil}$ be a quantum channel. Then,
    \begin{align*}
        \lim_{l\to \infty} C(\Psi^l) &= \log (\sum_k d_k), \\
        \lim_{l\to \infty} Q(\Psi^l) &= \lim_{l\to \infty} Q^{\dagger}(\Psi^l) = \lim_{l\to \infty} P_{\leftrightarrow}(\Psi^l) = \lim_{l\to \infty} P_{\leftrightarrow}^{\dagger}(\Psi^l) = \log (\max_k d_k)
    \end{align*}
    Moreover, for any other channel $\Gamma:\B{\Hil_B}\to \B{\Hil_C}$,
    \begin{align*}
        \lim_{l\to \infty} C(\Psi^l \otimes \Gamma) &= \log(\sum_k d_k) + C(\Gamma) \\
        \lim_{l\to \infty} P(\Psi^l \otimes \Gamma) &= \log(\max_k d_k) + P(\Gamma) \\ 
        \lim_{l\to \infty} Q(\Psi^l \otimes \Gamma) &= \log(\max_k d_k) + Q(\Gamma). 
    \end{align*}
    Here, $d_k = \dim \Hil_{k,1}$ for $k=1,2,\ldots ,K$ are the block dimensions in the decomposition of $\mathscr{X} (\Psi)$.
\end{corollary}

Given the infinite-length capacity bounds in the previous corollary, it is natural to ask for estimates on the length $l$ after which the capacities are close to their infinite-length values. From the convergence analysis presented in Section~\ref{subsec:convergence-storage}, we know that the asymptotic behavior of a channel becomes dominant (i.e., the diamond norm distance $\delta_l=\norm{\Psi^l-\Psi^l_{\infty}}_{\diamond}$ becomes small) after length $l\gtrsim d^2\ln (d)$, so that according to Theorems~\ref{thm:Casmpy}, \ref{thm:QPasymp}, \ref{theorem:strong-add}, the infinite-length capacities are reached after the same length. For example, suppose that $\Psi:\B{\Hil}\to \B{\Hil}$ is a quantum channel with $d=\dim\Hil$, $D=d^2$, and $\mu=\operatorname{spr}(\Psi-\Psi_{\infty})$. Fix $\delta\in(0,1)$, $\mu_0\in(\mu,1)$, and $\delta'=\frac{\delta(1-\mu_0)^{3/2}}{8e^2}$. Then, using Lemma~\ref{lemma:lambert-channel}, it is easy to see that for
\begin{equation}
    l \geq \frac{D}{\ln (1/\mu)} \left( \frac{ \ln (D^{D+2}/\delta') }{D} - \ln \left(\frac{\mu \ln (1/\mu)}{1-\mu^2} \right) + \sqrt{2}\sqrt{\frac{ \ln (D^{D+2}/\delta')  }{D} - \ln \left(\frac{\mu \ln (1/\mu)}{1-\mu^2} \right) -1} \right),
\end{equation}
and $l\geq \frac{\mu}{\mu_0-\mu}$, we have $d\delta_l = d\norm{\Psi^l-\Psi^l_{\infty}}_{\diamond} \leq \delta$, so that Theorem~\ref{thm:QPasymp} shows that 
\begin{align}
       \log \left( \max_k d_k \right) \leq Q(\Psi^l) \leq P^{\dagger}_{\leftrightarrow}(\Psi^l) \leq \log \left( \max_k d_k \right) +  \log (1+ \frac{\delta}{2}).
    \end{align}
Similarly, in the finite-blocklength regime, Theorem~\ref{theorem:QPfinite-max} shows that for any $n\in \mathbb{N}$ and $\epsilon\in [0,1)$,
\begin{align}
        \log(\max_k d_k) \leq \frac{1}{n} Q_{\epsilon}((\Psi^l)^{\otimes n})) &\leq \log(\max_k d_k) +  \log (1+\frac{\delta}{2}) + \frac{1}{n} \log (\frac{1}{1-\epsilon}), \\
        \log(\max_k d_k) \leq \frac{1}{n} C^{\operatorname{p}}_{\epsilon}((\Psi^l)^{\otimes n})) &\leq \log(\max_k d_k) +   \log (1+\frac{\delta}{2}) + \frac{1}{n} \log (\frac{1}{1-\epsilon}).
    \end{align}
Even if LOCC-assistance is allowed, Theorem~\ref{theorem:QPfinite-assisted} shows that
\begin{align}
    \log(\max_k d_k) \leq \frac{1}{n}Q^{\leftrightarrow}_{n,\epsilon}(\Psi^l) \leq \frac{1}{n}P^{\leftrightarrow}_{n,\epsilon}(\Psi^l)  &\leq \log(\max_k d_k) +  \log (1+\frac{\delta}{2}) + \frac{1}{n} \log (\frac{1}{1-\epsilon}).
\end{align}

\section{Conclusion}\label{sec:conclude}

In the concluding discussion, we compare and contrast the analyses presented in Sections~\ref{sec:qms-storage} and \ref{sec:qms-transmission}. Let $\Psi:\B{\Hil}\to \B{\Hil}$ be a quantum channel acting on a $d-$dimensional quantum system with $d=\dim\Hil$. Let us denote the one-shot $\epsilon-$error capacities of $n$ uses of the $m-$fold concatenation channel $\Psi^m=\Psi\circ \Psi \circ \ldots \circ \Psi$ by
\begin{align}
    Q(m,n,\epsilon) &:= Q_{\epsilon}((\Psi^m)^{\otimes n}), \\
    P(m,n,\epsilon) &:= C^{\operatorname{p}}_{\epsilon}((\Psi^m)^{\otimes n}).
\end{align}

\begin{itemize}
    \item In the data storage setup of Section~\ref{sec:qms-storage}, we fix $n\in \mathbb{N}$ and $\epsilon\in [0,1)$ and study the $t\to \infty$ behavior of the function $Q(t,n,\epsilon)$. Our analysis shows that for $t\gtrsim \ln(n)+d^2\ln(d)$, the following bound holds:
\begin{align}
    \log(\max_k d_k) &\leq \frac{1}{n} Q(t,n,\epsilon) \leq\log(\max_k d_k) + \frac{1}{n}\log (\frac{1}{1-\epsilon-\delta_t}), \label{eq:Qstorage-bound} \\
    \log(\max_k d_k) &\leq \frac{1}{n} P(t,n,\epsilon) \leq\log(\max_k d_k) + \frac{1}{n}\log (\frac{1}{1-\epsilon-\delta_t}), \label{eq:Pstorage-bound}
\end{align}
where $t\gtrsim \ln(n)+d^2\ln(d)$ ensures that $\delta_t=\norm{(\Psi^t)^{\otimes n}-(\Psi^t_{\infty})^{\otimes n}}_{\diamond}$ is small, see Section~\ref{subsec:IIDconvergence}.

\item In the point-to-point transmission setup of Section~\ref{sec:qms-transmission}, we study the behavior of $Q(l,n,\epsilon)$ for a fixed length $l\gtrsim d^2\ln (d)$ (independent of $n$). Our analysis shows that for all $n\in \mathbb{N}$ and $\epsilon\in [0,1)$,
\begin{align}
    \log(\max_k d_k) &\leq \frac{1}{n} Q(l,n,\epsilon) \leq \log(\max_k d_k) +  \log (1+\frac{\delta'_ld}{2}) + \frac{1}{n} \log (\frac{1}{1-\epsilon}), \label{eq:Qtransmission-bound} \\
    \log(\max_k d_k) &\leq \frac{1}{n} P(l,n,\epsilon) \leq \log(\max_k d_k) +  \log (1+\frac{\delta'_ld}{2}) + \frac{1}{n} \log (\frac{1}{1-\epsilon}), \label{eq:Ptransmission-bound}
\end{align}
where $l\gtrsim d^2\ln (d)$ ensures that the norm $\delta'_l=\norm{\Psi^l-\Psi^l_{\infty}}_{\diamond}$ is small, see Section~\ref{subsec:convergence-transmission}.
\end{itemize}

Notice that the $t\to\infty$ limit of the bounds~\eqref{eq:Qstorage-bound},\eqref{eq:Pstorage-bound} and the $l\to \infty$ limit of the bounds~\eqref{eq:Qtransmission-bound},\eqref{eq:Ptransmission-bound} agree with each other, giving us the following bounds for all $n\in \mathbb{N}$ and $\epsilon\in [0,1)$:
\begin{equation}
    \log(\max_k d_k) \leq \lim_{m\to \infty} \frac{1}{n} Q(m,n,\epsilon) \leq\log(\max_k d_k) + \frac{1}{n}\log (\frac{1}{1-\epsilon}).
\end{equation}
Importantly, the bounds~\eqref{eq:Qtransmission-bound},\eqref{eq:Ptransmission-bound} hold for a wider range of parameters $(m,n,\epsilon)$ compared to the bounds~\eqref{eq:Qstorage-bound},\eqref{eq:Pstorage-bound}. However, in the parameter space $m\gtrsim \ln(n)+d^2\ln(d)$ that is common to both, \eqref{eq:Qstorage-bound},\eqref{eq:Pstorage-bound} are tighter than \eqref{eq:Qtransmission-bound},\eqref{eq:Ptransmission-bound}. Both bounds hold in the `large' $m$ regime, and it would be interesting to scrutinize how non-trivial capacity estimates can be obtained in the `small' $m$ regime. In this regard, we should mention that by imposing additional constraints on the channel $\Psi$, such as the existence of a full rank invariant state and reversibility (given by a suitable detailed balance condition), it is possible to employ the framework of quantum functional inequalities to obtain converse bounds on the asymptotic capacities $Q(\Psi^m),P(\Psi^m)$ (Definition~\ref{def:capacity}) that hold for all $m$ (see \cite{MullerHermes2018capacity, Bardet2021group}). It would be interesting to analyze whether similar techniques can be applied in the one-shot finite-blocklength regime and to a more general class of channels.

Another interesting direction is to study capacities of Markovian semigroups when active error-correction is allowed in between time steps \cite{MullerHermes2015subdivision, fawzi2022lower}. Recall from Section~\ref{subsec:fault-tolerance} that by setting $\Psi = \Psi_{\text{ecc}} \circ \Psi_{\text{noise}}$, we can accommodate a fixed time-independent error correction mechanism $\Psi_{\text{ecc}}$ in our model that is engineered to detect and correct for errors induced by the noise $\Psi_{\operatorname{noise}}$ actively as they occur. We leave the analysis of capacities for more general time-dependent active error-correction mechanisms for future study. 

Finally, it would be interesting to analyze the finite block-length behavior of the classical capacity in the large $m\gtrsim d^2\ln(d)$ regime, similar to what is done in Section~\ref{subsec:non-zero-finiteblock} for quantum and private classical capacity. In this regard, as noted in Remark~\ref{remark:chi-cont}, it would be pertinent to perform a continuity analysis for the Holevo channel measures $\chi_{\max}$ and $\widetilde{\chi}_{\alpha}$, which might be useful to obtain strong converse bounds for classical capacity in the stated regime. Note that in the $m\gtrsim \ln(n)+d^2\ln(d)$ regime, the analysis from Section~\ref{sec:qms-storage} does give a bound similar to \eqref{eq:Qstorage-bound},\eqref{eq:Pstorage-bound} also for the classical capacity:
\begin{equation}
    \log(\sum_k d_k) \leq \frac{1}{n} C(m,n,\epsilon) \leq\log(\sum_k d_k) + \frac{1}{n}\log (\frac{1}{1-\epsilon-\delta_m}),
\end{equation}
where, as before, $m\gtrsim \ln(n)+d^2\ln(d)$ ensures that the norm $\delta_m=\norm{(\Psi^m)^{\otimes n}-(\Psi^m_{\infty})^{\otimes n}}_{\diamond}$ is small.

\section{Acknowledgements}
We thank \'Angela Capel for helpful discussions. S.S.~is supported by the Cambridge Trust International Scholarship. N.D.~is supported by the Engineering and Physical Sciences Research Council [Grant Ref: EP/Y028732/1].

\appendix

\section{LOCC-assisted quantum and private classical communication}\label{appen:assisted}

The ability of Alice and Bob to classically communicate with each other is modeled by an \emph{LOCC}\footnote{LOCC is an acronym for \emph{local operations assisted with classical communication}.} channel \cite{Chitambar2014LOCC}. A quantum channel $\Phi_{AB\to A'B'}$ is called $(A\to B)$ \emph{LOCC} if it can be written as 
\begin{equation}
    \Phi_{AB\to A'B'} = \sum_{x\in [\mathscr{X}]} \mathcal{M}^{(x)}_{A\to A'} \otimes \mathcal{N}^{(x)}_{B\to B'},
\end{equation}
where the sum is over a finite set of indices, $\{\mathcal{M}^{(x)}_{A\to A'}\}_{x\in [\mathscr{X}]}$ is a set of linear completely positive maps that sum to a trace preserving map, and $\{\mathcal{N}^{(x)}_{B\to B'} \}_{x\in [\mathscr{X}]}$ is a set of quantum channels. A channel $\Phi_{AB\to A'B'}$ is called $(A\leftrightarrow B)$ \emph{LOCC} (or \emph{two-way LOCC} or simply \emph{LOCC}) if it can be written as a concatenation of a finite sequence of $(A\to B)$ and $(B\to A)$ LOCC channels.  

\begin{definition}  (LOCC-assisted quantum communication) \label{def:LOCCquantum} \\
    Let $n,d\in \mathbb{N}$ and $\epsilon\in [0,1)$. An $(n, d,\epsilon)$ \emph{LOCC-assisted quantum code} for a channel $\Phi_{A\to B}$ consists of a separable state $\sigma_{A_1'A_1B_1'}\in \operatorname{SEP}(A_1'A_1:B_1')$, a set $\{\mathcal{L}^{(i-1)}_{A_{i-1}'B_{i-1}B'_{i-1}\to A_i'A_iB_i'} \}_{i=2}^n$ of LOCC channels, and another LOCC channel $\mathcal{L}^{(n)}_{A_n'B_nB'_n\to M_AM_B}$, such that the final state of the protocol defined by this code
    \begin{align}\label{eq:LOCCquantum}
        \eta_{M_AM_B} = \mathcal{L}^{(n)}_{A_n'B_nB'_n\to M_AM_B} \circ \Phi_{A_n\to B_n} \circ \mathcal{L}^{(n-1)}_{A_{n-1}'B_{n-1}B'_{n-1}\to A_n'A_nB_n'} \circ \ldots  \nonumber \\ 
        \circ\mathcal{L}^{(1)}_{A_{1}'B_{1}B'_{1}\to A_2'A_2B_2'}\circ \Phi_{A_1\to B_1}(\sigma_{A_1'A_1B_1'})
    \end{align}
    satisfies
    \begin{equation}
        F(\psi^+_{M_AM_B}, \eta_{M_AM_B}) \geq 1- \epsilon,
    \end{equation}
    where $\psi^+_{M_AM_B}$ is a maximally entangled state of Schmidt rank $d=d_{M_A}=d_{M_B}$. 
    
    Observe that since free classical communication is allowed in this protocol, Alice and Bob can use the final state $\eta_{M_AM_B}$ of the protocol to approximately simulate an identity channel of dimension $d$ by using quantum teleportation \cite{Bennett1993teleport}. The $n-$\emph{shot} $\epsilon$-error \emph{LOCC-assisted quantum capacity}
    of $\Phi$ is defined as 
    \begin{align}
    Q_{n,\epsilon}^{\leftrightarrow}(\Phi):= \sup \{ \log d: \exists (n, d,\epsilon)& \text{ LOCC-assisted} \nonumber \\ 
    &\text{quantum code for } \Phi\}.
\end{align}
The \emph{(asymptotic) LOCC-assisted quantum capacity} of $\Phi$ is defined as 
\begin{equation}
    Q_{\leftrightarrow}(\Phi) := \inf_{\epsilon\in (0,1)} \liminf_{n\to \infty} \frac{1}{n} Q_{n,\epsilon}^{\leftrightarrow}(\Phi).
\end{equation}
The corresponding strong converse capacity is defined as 
\begin{equation}
    Q^{\dagger}_{\leftrightarrow}(\Phi) := \sup_{\epsilon\in (0,1)} \limsup_{n\to \infty} \frac{1}{n} Q_{n,\epsilon}^{\leftrightarrow}(\Phi).
\end{equation}

\end{definition}

\begin{figure}[H]
    \centering
    \includegraphics[width=0.85\linewidth]{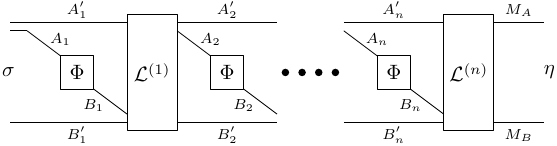}
    \caption{An LOCC-assisted quantum communication protocol via $n$ uses of a quantum channel $\Phi_{A\to B}$. Alice starts with a separable state $\sigma\in \operatorname{SEP}(A_1'A_1:B_1)$. She sends the system $A_1$ to Bob via $\Phi_{A\to B}$, after which the two parties perform an LOCC operation $\mathcal{L}^{(1)}_{A_1'B_1B_1'\to A_2'A_2B_2 }$. This procedure is repeated $n$ times. If the state $\eta$ at the end of the protocol is close to a maximally entangled state $\psi^+$ of Schmidt rank $d$: $F(\eta, \psi^+) \geq 1- \epsilon$, we say that the protocol is an $(n,d,\epsilon)$ LOCC-assisted quantum code for $\Phi$.  }
    \label{fig:assisted-quantum}
\end{figure}

\begin{definition}  (LOCC-assisted private classical communication) \label{def:LOCCprivate} \\
    Let $n,\mathscr{M}\in \mathbb{N}$ and $\epsilon\in [0,1)$. An $(n, \mathscr{M},\epsilon)$ \emph{LOCC-assisted private classical code} for a channel $\Phi_{A\to B}$ consists of a separable state $\sigma_{A_1'A_1B_1'}\in \operatorname{SEP}(A_1'A_1:B_1')$, a set $\{\mathcal{L}^{(i-1)}_{A_{i-1}'B_{i-1}B'_{i-1}\to A_i'A_iB_i'} \}_{i=2}^n$ of LOCC channels, and another LOCC channel $\mathcal{L}^{(n)}_{A_n'B_nB'_n\to M_AM_BS_AS_B}$, such that the final state of the protocol
    \begin{align}\label{eq:LOCCprivate}
        \eta_{M_AM_BS_AS_B} = \mathcal{L}^{(n)}_{A_n'B_nB'_n\to M_AM_BS_AS_B} \circ \Phi_{A_n\to B_n} \circ \mathcal{L}^{(n-1)}_{A_{n-1}'B_{n-1}B'_{n-1}\to A_n'A_nB_n'} \circ \ldots  \nonumber \\ 
        \circ\mathcal{L}^{(1)}_{A_{1}'B_{1}B'_{1}\to A_2'A_2B_2'}\circ \Phi_{A_1\to B_1}(\sigma_{A_1'A_1B_1'})
    \end{align}
    satisfies
    \begin{equation}
        F(\gamma^+_{M_AM_BS_AS_B}, \eta_{M_AM_BS_AS_B}) \geq 1- \epsilon,
    \end{equation}
    where $\gamma^+_{M_AM_BS_AS_B}$ is a \emph{private state} \cite{Horodecki2009private} of \emph{key dimension} $d=d_{M_A}=d_{M_B}$, i.e., $\gamma^+$ is of the form 
    \begin{equation}\label{eq:private}
        \gamma^+_{M_AM_BS_AS_B} = U_{M_AM_BS_AS_B} (\psi^+_{M_AM_B} \otimes \theta_{S_AS_B} ) U^{\dagger}_{M_AM_BS_AS_B},
    \end{equation}
    where $\psi^+_{M_AM_B}$ is a maximally entangled state of Schmidt rank $\mathscr{M}$ and $U_{M_AM_BS_AS_B}$ is a \emph{twisting} unitary:
    \begin{align}
        \psi^+_{M_AM_B} &= \frac{1}{\mathscr{M}} \sum_{m,m'\in [\mathscr{M}]} \ketbra{m}{m'}_{M_A} \otimes \ketbra{m}{m'}_{M_B}  \\
        U_{M_AM_BS_AS_B} &= \sum_{m,m'\in [\mathscr{M}]} \ketbra{m}{m}_{M_A} \otimes \ketbra{m'}{m'}_{M_B} \otimes U^{(m,m')}_{S_AS_B}. \label{eq:twisting}
    \end{align}
    Here, $U^{(m,m')}_{S_AS_B}$ are arbitrary unitary matrices and $\theta_{S_AS_B}$ is an arbitrary state.
    
    The equivalence between the above formulation of private capacity in terms of `distilling' private states and the usual formulation in terms of privately communicating classical information as in Definition~\ref{def:private-protocol} was established in the seminal works \cite{Horodecki2005private, Horodecki2009private}. The $n-$\emph{shot} $\epsilon$-error \emph{LOCC-assisted private classical capacity}
    of $\Phi$ is defined as 
    \begin{align}
    P_{n,\epsilon}^{\leftrightarrow}(\Phi):= \sup \{ \log \mathscr{M}: \exists (n, \mathscr{M},\epsilon)& \text{ LOCC-assisted} \nonumber \\ 
    &\text{private classical code for } \Phi\}.
\end{align}
The \emph{(asymptotic) LOCC-assisted private classical capacity} of $\Phi$ is defined as 
\begin{equation}
    P_{\leftrightarrow}(\Phi) := \inf_{\epsilon\in (0,1)} \liminf_{n\to \infty} \frac{1}{n} P_{n,\epsilon}^{\leftrightarrow}(\Phi).
\end{equation}
The corresponding strong converse capacity is defined as 
\begin{equation}
    P^{\dagger}_{\leftrightarrow}(\Phi) := \sup_{\epsilon\in (0,1)} \limsup_{n\to \infty} \frac{1}{n} P_{n,\epsilon}^{\leftrightarrow}(\Phi).
\end{equation}
\end{definition}

As noted before, strong converse capcities are always at least as large as the normal capacities:
\begin{equation}
    Q_{\leftrightarrow}(\Phi) \leq Q^{\dagger}_{\leftrightarrow}(\Phi),  \quad P_{\leftrightarrow}(\Phi)\leq P^{\dagger}_{\leftrightarrow}(\Phi).
\end{equation}
Furthermore, since assistance from classical communication can only enhance the capacities, we obtain the following relations between the unassisted and assisted capacities for all $n\in \mathbb{N}$ and $\epsilon\in [0,1)$:
\begin{align}
    Q_{\epsilon}(\Phi^{\otimes n}) \leq Q^{\leftrightarrow}_{n,\epsilon}(\Phi), \\
    C^{\operatorname{p}}_{\epsilon}(\Phi^{\otimes n}) \leq P^{\leftrightarrow}_{n,\epsilon}(\Phi),
\end{align}
which clearly imply that
\begin{align}
    Q(\Phi) \leq Q_{\leftrightarrow}(\Phi), \quad Q^{\dagger}(\Phi)&\leq Q^{\dagger}_{\leftrightarrow}(\Phi),  \\
    P(\Phi) \leq P_{\leftrightarrow}(\Phi), \quad  P^{\dagger}(\Phi)&\leq P^{\dagger}_{\leftrightarrow}(\Phi).
\end{align}
Moreover, since a more general target state is allowed for private classical communication, we get the following relation for all $n\in \mathbb{N}$ and $\epsilon\in [0,1)$:
\begin{align}\label{eq:Q<P-assisted}
   Q^{\leftrightarrow}_{n,\epsilon}(\Phi)\leq P^{\leftrightarrow}_{n,\epsilon}(\Phi).
\end{align}

It turns out that the max-relative entropy of entanglement of a channel (Definition~\ref{def:channel-measures}) provides a fundamental converse bound on the aforementioned assisted capacities.

\begin{theorem} \cite{Christandl2017max} \label{theorem:QPassisted<Emax}
    Let $\Phi:\B{\Hil_A}\to \B{\Hil_B}$ be a quantum channel, $\epsilon\in [0,1)$, and $n\in \mathbb{N}$. Then,
    \begin{equation}
      Q_{n,\epsilon}^{\leftrightarrow}(\Phi) \leq  P_{n,\epsilon}^{\leftrightarrow}(\Phi) \leq nE_{\max}(\Phi) + \log(\frac{1}{1-\epsilon}).
    \end{equation}
Consequently,
    Let $\Phi:\B{\Hil_A}\to \B{\Hil_B}$ be a quantum channel. Then,
    \begin{equation}
     Q^{\dagger}_{\leftrightarrow}(\Phi)\leq   P^{\dagger}_{\leftrightarrow}(\Phi) \leq E_{\max}(\Phi).
    \end{equation}
\end{theorem}

\section{Continuity of channel capacities}\label{sec:continuity}
In what follows, $S(A|B)_{\rho} := - I(A\rangle B)_{\rho} = S(AB)_{\rho}-S(B)_{\rho}$ denotes the conditional entropy of a bipartite state $\rho_{AB}\in \State{\Hil_A\otimes \Hil_B}$. Recently, the Alicki-Fannes-Winter continuity bound for conditional entropies \cite{Alicki2004continuous, Winter2016tight} was improved to the following bound.

\begin{theorem} \cite{Berta2025continuous, Audenaert2024continuous} \label{theorem:AFW}
    Let $\rho_{AB}, \sigma_{AB} \in \mathcal{D}(\cH_A \otimes \cH_B)$ be such that they have equal marginals on the $B$ subsystem $\rho_B=\sigma_B$ and $\frac{1}{2} \norm{\rho_{AB} - \sigma_{AB}}_1 \leq \delta$. Then,
    \begin{equation}
       |S(A|B)_{\rho} - S(A|B)_{\sigma}| \le \begin{cases} \delta \log( d_A^2 - 1) + h(\delta), & \text{if }\,\delta \leq 1 - 1/d_A^2 \\
       \log(d_A^2), & \text{if }\,\delta > 1-1/d_A^2
       \end{cases}\, ,
    \end{equation}
    where $h(\delta):= -\delta\log \delta - (1-\delta)\log (1-\delta)$ is the binary entropy function.
\end{theorem}

Using the above theorem, the continuity bound on channel output entropy \cite[Theorem 11]{Leung2009continuous} can now be improved as follows.

    \begin{theorem}\label{theorem:out-entropy}
        Let $\Phi, \Psi : \B{\Hil_{A}} \to \B{\Hil_B}$ be two quantum channels such that
        $\frac{1}{2}\| \Phi - \Psi\|_\diamond = \delta$. Then, for any state $\rho\in \cD(\cH_R \otimes \cH_{A}^{\otimes n})$,
         \begin{equation*}
       |S((\id_R \otimes \Phi^{\otimes n})(\rho))-
       S((\id_R \otimes \Psi^{\otimes n})(\rho))|\le \begin{cases} n\left(\delta \log( d_B^2 - 1) + h(\delta) \right), & \text{if }\,\delta \leq 1 - 1/d_B^2 \\
       n\log(d_B^2), & \text{if }\,\delta > 1-1/d_B^2
       \end{cases}\, .
    \end{equation*}
    \end{theorem}
    
    This in turn leads to the following improved continuity bounds for the capacities of quantum channels. These are improvements over the results stated as \cite[Corollary 1,2 and 3]{Leung2009continuous}.
    \begin{theorem}\label{theorem:cap-cont}
         Let $\Phi, \Psi : \B{\Hil_{A}} \to \B{\Hil_B}$ be two channels such that 
        $\frac{1}{2}\| \Phi - \Psi\|_\diamond = \delta \leq  1 - 1/d_B^2 $. Then, the following bounds hold on the channel capacities:
\begin{align*}
    |C(\Phi) - C(\Psi)| &\leq 2 \left(\delta \log( d_B^2 - 1) + h(\delta) \right), \\
    |P(\Phi) - P(\Psi)| &\leq 4 \left(\delta \log( d_B^2 - 1) + h(\delta) \right), \\
    |Q(\Phi) - Q(\Psi)| &\leq 2 \left(\delta \log( d_B^2 - 1) + h(\delta) \right), \\
    |C_{\operatorname{ea}}(\Phi) - C_{\operatorname{ea}}(\Psi)| &\leq 2 \left(\delta \log( d_B^2 - 1) + h(\delta) \right).
\end{align*}
    \end{theorem}    
    
Finally, we note some continuity bounds on the relative entropies of entanglement.
 
\begin{lemma} \cite{bluhm2023continuity}
    Let $\rho_{AB}, \sigma_{AB}\in \State{\Hil_A\otimes \Hil_B}$ be such that $\frac{1}{2} \norm{\rho-\sigma}_1 \leq \delta$. Furthermore, let $d=\min ( d_A, d_B)$ and $\alpha>1$. Then,
    \begin{align}
        |\widetilde{E}_{\alpha}(A:B)_{\rho} - \widetilde{E}_{\alpha}(A:B)_{\sigma}| &\leq \frac{\alpha}{\alpha-1}\log (1+\delta d^{\frac{\alpha-1}{\alpha}}) \\
        |E_{\max} (A:B)_{\rho} - E_{\max}(A:B)_{\sigma}| &\leq \log (1+\delta d). 
    \end{align}
\end{lemma}

\begin{lemma}\label{lemma:Emax-cont}
    Let $\Phi,\Psi : \B{\Hil_A}\to \B{\Hil_B}$ be quantum channels such that $\frac{1}{2}\norm{\Phi-\Psi}_{\diamond} = \delta$. Furtheremore, let $d=\min(d_A,d_B)$ and $\alpha>1$. Then,
    \begin{align}
        |\widetilde{E}_{\alpha}(\Phi) - \widetilde{E}_{\alpha}(\Psi)| &\leq \frac{\alpha}{\alpha-1}\log (1+\delta d^{\frac{\alpha-1}{\alpha}}) \\
        |E_{\max}(\Phi) - E_{\max}(\Psi)| &\leq \log (1+\delta d). 
    \end{align}
\end{lemma}
\begin{proof}
    Recall that $E_{\max}(\Phi) = \sup_{\psi_{RA}} E_{\max}(R:B)_{ \Phi_{A\to B}(\psi_{RA})}$ and $\widetilde{E}_{\alpha}(\Phi) = \sup_{\psi_{RA}} \widetilde{E}_{\alpha}(R:B)_{ \Phi_{A\to B}(\psi_{RA})}$, where the supremum is over all pure states with $d_R=d_A$ (see Definition~\ref{def:channel-entanglement-measures}). Now, since $\frac{1}{2} \norm{\Phi- \Psi}_{\diamond}= \delta$, it is clear that for each state $\psi_{RA}$, we have $\frac{1}{2}\norm{\Phi_{A\to B}(\psi_{RA}) - \Psi_{A\to B}(\psi_{RA}) }_1\leq \delta$. A simple application of the previous lemma then proves the desired result.
\end{proof}

\section{Strong additivity of identity and replacer channels} \label{appen:additive}

\begin{lemma}
    Let $\id:\B{\Hil_A}\to \B{\Hil_A}$ be the identity channel and $\Phi:\B{\Hil_B}\to \B{\Hil_C}$ be an arbitrary channel. Then, 
    \begin{align*}
        \chi(\id \otimes \Phi) &= \log d_A + \chi(\Phi), \\
        I_p(\id \otimes \Phi) &= \log d_A + I_p(\Phi), \\
        I_c(\id \otimes \Phi) &= \log d_A + I_c(\Phi).
    \end{align*}
\end{lemma}
\begin{proof}
    It is easy to check $\chi(\id)=I_p(\id)=I_c(\id)=\log d_A$. The desired additivity follows from the fact that identity is a Hadamard channel (i.e., its complementary channel is entanglement-breaking), and $\chi, I_p, I_c$ are strongly additive for Hadamard channels \cite{King2006hadamard, Brdler2010Hadamard, Wilde2011Hadamard, Winter2016potential}. 
\end{proof}

\begin{lemma}
    Let $\mathcal{R}:\B{\Hil_{A_1}}\to \B{\Hil_{B_1}}$ be a replacer channel of the form $\mathcal{R}(X)=\Tr (X)\delta$ for some state $\delta\in \State{\Hil_{B_1}}$ and $\Phi:\B{\Hil_{A_2}}\to \B{\Hil_{B_2}}$ be an arbitrary channel. Then, 
    \begin{align*}
        \chi(\mathcal{R} \otimes \Phi) &=  \chi(\Phi), \\
        I_p(\mathcal{R} \otimes \Phi) &=  I_p(\Phi), \\
        I_c(\mathcal{R} \otimes \Phi) &=  I_c(\Phi).
    \end{align*}
\end{lemma}
\begin{proof}
    It is clear that $\chi(\mathcal{R})=I_p(\mathcal{R})=I_c(\mathcal{R})=0$ from Definition~\ref{def:channel-measures}. Moreover, since all information measures are superadditive, it suffices to show that $\chi(\mathcal{R}\otimes \Phi)\leq \chi(\Phi)$, and similarly for $I_p, I_c$. With this end in sight, let $\rho_{XA_1A_2}$ be an arbitrary cq state, $\omega_{XB_1B_2} = (\mathcal{R}_{A_1\to B_1} \otimes \Phi_{A_2\to B_2})(\rho_{XA_1A_2})$, and $\sigma_{XE_1E_2} = (\mathcal{R}^c_{A_1\to E_1} \otimes \Phi^c_{A_2\to E_2})(\rho_{XA_1A_2})$. Then,
    \begin{align}
        I(X:B_1B_2)_{\omega} &= S(X) + S(B_1B_2) - S(XB_1B_2) \nonumber \\
        &=  S(X) + S(B_2) + S(B_1) - S(XB_2) - S(B_1) \nonumber \\
        &= S(X) + S(B_2) - S(XB_2) \nonumber \\
        &= I(X:B_2) \leq \chi (\Phi).
    \end{align}
    Moreover, 
    \begin{align}
        I(X:B_1B_2)_{\omega} - I(X:E_1E_2)_{\sigma} &= I(X:B_2) - I(X:E_2) - I(X:E_1 |E_2) \nonumber \\
        &\leq I(X:B_2) - I(X:E_2) \nonumber \\
        &\leq I_p(\Phi)
    \end{align}
    Finally, if $\psi_{RA_1A_2}$ is an arbitrary pure state  and $\omega_{RB_1B_2} = (\mathcal{R}_{A_1\to B_1} \otimes \Phi_{A_2\to B_2})(\psi_{RA_1A_2})$,
    \begin{align}
        I(R\rangle B_1 B_2)_{\omega} &= S(B_1B_2) - S(RB_1B_2) \nonumber \\
        &= S(B_1) + S(B_2) - S(B_1) - S(RB_2) \nonumber\\
        &= S(B_2) - S(RB_2) \nonumber  \\
        &= I(R\rangle B_2) \leq I_c(\Phi).
    \end{align}
    Hence, the desired claims follow.
\end{proof}

\section{Techninal lemmas}

\begin{lemma}\label{lemma:Q0<=P0}
    For a quantum channel $\Phi:\B{\Hil_A}\to \B{\Hil_B}$, the one-shot zero-error quantum and private classical capacities satisfy the relation $Q_0(\Phi)\leq C^{\operatorname{p}}_0(\Phi)$.
\end{lemma}
\begin{proof}
    Consider a $(\mathscr{M}, 0)$ quantum code $(\mathcal{E}_{A'\to A}, \mathcal{D}_{B\to A'})$ for $\Phi$ with $\mathscr{M}=d_{A'}=d_R$ (Definition~\ref{def:quantum-protocol}), which we can use to transmit one-half of a maximally entangled state 
    \begin{equation}
        \psi^+_{RA'} = \frac{1}{\mathscr{M}} \sum_{m,m'} \ketbra{m}{m'}_R \otimes \ketbra{m}{m'}_{A'}
    \end{equation}
    of Schmidt rank $\mathscr{M}$ through $\Phi$ perfectly, i.e.
    \begin{equation}
        \psi^+_{RA'} = \mathcal{D}_{B\to A'}\circ \Phi_{A\to B}\circ \mathcal{E}_{A'\to A}(\psi^+_{RA'}).
    \end{equation}
    Let $\mathcal{V}_{A\to BE}$ be an isometric extension of $\Phi_{A\to B}$ and consider the state 
    \begin{equation}
        \omega_{RA'E} = \mathcal{D}_{B\to A'}\circ \mathcal{V}_{A\to BE}\circ \mathcal{E}_{A'\to A}(\psi^+_{RA'}),
    \end{equation}
    which extends the state at the output of the protocol, i.e.,   $\omega_{RA'} = \mathcal{D}_{B\to A'}\circ \Phi_{A\to B}\circ \mathcal{E}_{A'\to A}(\psi^+_{RA'})$. Since the only possible extension of $\psi^+_{RA'}$ is of the form $\psi^+_{RA'}\otimes \sigma_E$ for some state $\sigma_E$, we get
    \begin{equation}
        \omega_{RA'E} = \psi^+_{RA'}\otimes \sigma_E =  \mathcal{D}_{B\to A'}\circ \mathcal{V}_{A\to BE}\circ \mathcal{E}_{A'\to A}(\psi^+_{RA'}).
    \end{equation}
    Applying a measurement with POVMs $\{ \ketbra{m}_R \}_{m\in [\mathscr{M}]}$ and $\{ \ketbra{m}_{A'} \}_{m\in [\mathscr{M}]}$ on the $R$ and $A'$ systems yields    \begin{equation}\label{eq:star}
        \overbar{\psi}_{RA'}^+ \otimes \sigma_E =  \overbar{\mathcal{D}}_{B\to A'}\circ \mathcal{V}_{A\to BE}\circ \mathcal{E}_{A'\to A}(\overbar{\psi}^+_{RA'}),
    \end{equation}
    where $\overbar{\psi}_{RA'}^+ = \frac{1}{\mathscr{M}} \sum_{m} \ketbra{m}{m}_R \otimes \ketbra{m}{m}_{A'}$ is a maximally classically correlated state and $\overbar{\mathcal{D}}_{B\to A'}$ is a measurement channel defined as 
    \begin{align}
        \overbar{\mathcal{D}}_{B\to A'}(X_B) &= \sum_m \Tr (\ketbra{m}_{A'} \mathcal{D}_{B\to A'}(X_B) ) \ketbra{m}_{A'} \nonumber \\ 
        &= \sum_m \Tr ( \mathcal{D}_{A'\to B}^*( \ketbra{m}_{A'}) X_B ) \ketbra{m}_{A'}. 
    \end{align} 
    Note that $\{ \mathcal{D}_{A'\to B}^*( \ketbra{m}_{A'} \}_{m\in [\mathscr{M}]}^{\mathscr{M}}$ forms a POVM, since $\mathcal{D}_{A'\to B}^*$ is unital. Expanding the LHS and RHS of the Eq.~\eqref{eq:star} by using the formula for $\overbar{\psi}^+$ and matching terms shows that for each $m\in [\mathscr{M}]$,
    \begin{equation}
        \ketbra{m}_{A'}\otimes \sigma_E = \overbar{\mathcal{D}}_{B\to A'}\circ \mathcal{V}_{A\to BE} (\rho^m_{A}),
    \end{equation}
    where the states $\rho^m$ are defined as $\rho^m_A = \mathcal{E}_{A'\to A}(\ketbra{m}_{A'})$. Thus, the encoding states $\{\rho^m_A\}_{m\in [\mathscr{M}]}$ and the decoding POVM $\{ \mathcal{D}_{A'\to B}^*( \ketbra{m}_{A'} \}_{m\in [\mathscr{M}]}$ forms a $(\mathscr{M},0)$ private classical code for $\Phi$ (Definition~\ref{def:private-protocol}). Since the quantum code that we started with was arbitrary, we obtain the desired result.
\end{proof}

\begin{lemma}\label{lemma:Emax}
    For any state $\rho_{AB}\in \State{\Hil_A \otimes \Hil_B}$,
\begin{equation}
    \inf_{\sigma\in \operatorname{SEP}(A:B)} D_{\max}(\rho_{AB}\Vert \sigma_{AB}) \leq \log \min \{d_A,d_B \}.
\end{equation}
\end{lemma}
\begin{proof}
    Consider the spectral decomposition $\rho_{AB}=\sum_i p_i \psi^i_{AB}$, where $\psi^i_{AB}$ are pure states and $\sum_i p_i =1$, so that we can write
    \begin{align}
        \inf_{\sigma\in \operatorname{SEP}(A:B)}D_{\max}(\rho_{AB}\Vert \sigma_{AB}) &\leq \inf_{\{\sigma^i \}_i \subset \operatorname{SEP}(A:B) } D_{\max}\left( \sum_i p_i \psi^i_{AB} \bigg\| \sum_i p_i\sigma^i_{AB} \right) \nonumber\\ 
        &\leq \inf_{\{\sigma^i \}_i \subset \operatorname{SEP}(A:B) } \max_i D_{\max} (\psi^i_{AB}\Vert \sigma^i_{AB}) \nonumber \\
        &= \max_i \inf_{\sigma^i \in \operatorname{SEP}(A:B) } D_{\max} (\psi^i_{AB}\Vert \sigma^i_{AB}).
    \end{align}
    Note that we made use of quasi-convexity of $D_{\max}$ (Remark~\ref{remark:Dmax-quasi}) to obtain the second inequality above and of Lemma~\ref{lemma:infmax} to obtain the last equality. Thus, it suffices to prove the claim for pure states. For a pure state $\psi_{AB}$, it is known that \cite{Datta2009maxrel}
    \begin{equation}
        \inf_{\sigma \in \operatorname{SEP}(A:B) } D_{\max} (\psi_{AB}\Vert \sigma_{AB}) = 2\log (\sum_{i=1}^d \sqrt{s_i}),
    \end{equation}
    where $s_i\geq 0$ are the Schmidt coefficients of $\psi_{AB}$ satisfying $\sum_{i=1}^d s_i=1$ and $d=\min \{d_A,d_B \} $. A simple application of Cauchy-Schwarz inequality then shows
    \begin{equation}
        \sum_{i} \sqrt{s_i} \leq \sqrt{d \sum_i s_i} = \sqrt{d}. 
    \end{equation}
    Hence,
    \begin{equation}
        \inf_{\sigma \in \operatorname{SEP}(A:B) } D_{\max} (\psi_{AB}\Vert \sigma_{AB}) = 2\log (\sum_{i=1}^d \sqrt{s_i}) \leq 2 \log \sqrt{d} = \log d = \log \min \{d_A,d_B \}.
    \end{equation}
\end{proof}

\begin{lemma}\label{lemma:infmax}
    Let $f_k : \mathcal{S}\to \mathbb{R}$ for $k=1,2,\ldots ,K$ be arbitrary mappings, where $\mathcal{S}$ is an arbitrary set. Then, 
    \begin{equation}
        \inf_{\{x_k \}_k \subset \mathcal{S} } \left( \max_k f_k (x_k) \right) = \max_k \left(\inf_{x\in \mathcal{S}} f_k (x) \right).
    \end{equation}
\end{lemma}
\begin{proof}
    Clearly, for any subset $\{x_k \}_k \subset \mathcal{S}$, we have
    \begin{equation}
       \max_k \left(\inf_{x\in \mathcal{S}} f_k (x) \right) \leq \max_k f_k (x_k), 
    \end{equation}
    so that 
    \begin{equation}
       \max_k \left(\inf_{x\in \mathcal{S}} f_k (x) \right) \leq \inf_{\{x_k \}_k \subset \mathcal{S} } \left( \max_k f_k (x_k) \right). 
    \end{equation}
    Next we justify that the above inequality cannot be strict. Note that for any $\delta>0$, the number $\max_k \left(\inf_{x\in \mathcal{S}} f_k (x) \right) + \delta$, by definition, cannot be a lower bound on the sets $\{f_k (x) \}_{x\in \mathcal{S}}$ for all $k$. In other words, for every $\delta>0$, there exists a subset $\{x_k \}_k \subset \mathcal{S}$ such that for each $k$, $f_k(x_k) < \max_k \left(\inf_{x\in \mathcal{S}} f_k (x) \right) + \delta$, which means that  
    \begin{equation}
      \max_k f_k (x_k) <  \max_k \left(\inf_{x\in \mathcal{S}} f_k (x) \right) + \delta. 
    \end{equation}
    Hence, $\max_k \left(\inf_{x\in \mathcal{S}} f_k (x) \right)$ must be the greatest lower bound on the set $\{ \max_k f_k (x_k) \}_{\{x_k \}_k \subset \mathcal{S} }$, which is what we want to prove:
    \begin{equation}
        \max_k \left(\inf_{x\in \mathcal{S}} f_k (x) \right) = \inf_{\{x_k \}_k \subset \mathcal{S} } \left( \max_k f_k (x_k) \right). 
    \end{equation}
\end{proof}

\begin{lemma}\label{lemma:EHstate}
    Let $\rho_{AB}, \omega_{AB}$ be states such that $\norm{\rho-\omega}_1 \leq \delta$. Then, for $\epsilon\in [0,1)$ such that $\epsilon+\delta<1$, $E_{H}^{\epsilon}(A:B)_{\rho} \leq E_{H}^{\epsilon+\delta}(A:B)_{\omega}$.
\end{lemma}
\begin{proof}
    Fix a separable state $\sigma_{AB}$. Then, for every $\Lambda\in \B{\Hil_A\otimes \Hil_B}$ satisfying $0\leq \Lambda \leq \iden$ and $\Tr \Lambda \rho \geq 1-\epsilon$, we have $\Tr \Lambda \omega = \Tr \Lambda \rho - \Tr \Lambda (\rho - \omega) \geq 1 - (\epsilon + \delta) $. Hence, $D^{\epsilon}_H (\rho \Vert \sigma) \leq D^{\epsilon+\delta}_H (\omega \Vert \sigma)$ (see Definition~\ref{def:divergence}). The claim follows by taking an infimum over all separable states $\sigma_{AB}$.
\end{proof}

\begin{lemma}\label{lemma:Eepsilon-delta}
    Let $\Phi, \Psi : \B{\Hil_A}\to \B{\Hil_B}$ be quantum channels such that $\norm{\Phi-\Psi}_{\diamond}= \delta$. Then, for $\epsilon\in [0,1)$ such that $\epsilon+\delta<1$,
    \begin{align}
        E_H^{\epsilon}(\Phi) \leq E_H^{\epsilon+\delta}(\Psi)
    \end{align}
\end{lemma}
\begin{proof}
    For any pure state $\psi_{RA}$, $\norm{\Phi_{A\to B}(\psi_{RA}) - \Psi_{A\to B}(\psi_{RA}) }_1 \leq \delta$ because $\norm{\Phi-\Psi}_{\diamond}= \delta$. Thus, the claim easily follows from Lemmma~\ref{lemma:EHstate}.
\end{proof}

\begin{lemma}\label{lemma:lambert}
    Let $\delta,\mu\in (0,1)$, $\alpha>1$ and $d\in \mathbb{N}$ be a positive integer. Then, 
    \begin{equation*}
        d^{\alpha} \left(\frac{t(1-\mu^2)}{\mu}\right)^d \mu^t \leq \delta 
    \end{equation*}
    if 
    \begin{equation*}
        t \geq \frac{d}{\ln (1/\mu)} \left( \frac{ \ln (d^{d+\alpha}/\delta) }{d} - \ln \left(\frac{\mu \ln (1/\mu)}{1-\mu^2} \right) + \sqrt{2}\sqrt{\frac{ \ln (d^{d+\alpha}/\delta)  }{d} - \ln \left(\frac{\mu \ln (1/\mu)}{1-\mu^2} \right) -1} \right).
    \end{equation*}
    \begin{proof}
        The inequality we want to solve is equivalent to 
        \begin{align}
            d^{\alpha} \left(\frac{t(1-\mu^2)}{\mu}\right)^d \mu^t \leq \delta  \iff
            \frac{t\ln (\mu)}{d} \exp({\frac{t\ln (\mu)}{d}}) &\geq  \frac{\mu \ln (\mu) \delta^{1/d}}{(1-\mu^2) d^{1+\alpha/d}} \nonumber \\ 
            \iff \frac{t\ln (\mu)}{d} &\leq W_{-1}\left( \frac{\mu \ln (\mu) \delta^{1/d}}{(1-\mu^2) d^{1+\alpha/d}} \right) \nonumber \\
            \iff t &\geq \frac{d}{\ln (\mu)} W_{-1}\left( \frac{\mu \ln (\mu) \delta^{1/d}}{(1-\mu^2) d^{1+\alpha/d}} \right), \label{eq:lambert}
        \end{align}
        where $W_{-1}$ is the Lambert $W$ function \cite{Corless1996lambert}, which satisfies $W_{-1}(xe^x)=x$ for all $x\leq -1$ and is non-increasing on its domain. Moreover, for all $u>0$, it is known that $W_{-1}(-e^{-u-1})> -1-\sqrt{2u}-u$ \cite{Chatzigeorgiou2013lambert}. Thus, we make the substitution
        \begin{equation}
            -e^{-u-1} = \frac{\mu \ln (\mu) \delta^{1/d}}{(1-\mu^2) d^{1+\alpha/d}} \iff u = \frac{\ln(d^{d+\alpha}/\epsilon)}{d} - \ln(\frac{\mu \ln(\mu)}{1-\mu^2}) -1,
        \end{equation}
        so that Eq.~\eqref{eq:lambert} gives us the required bound on $t$:
        \begin{equation}
            t \geq \frac{d}{\ln (1/\mu)} \left( \frac{ \ln (d^{d+\alpha}/\delta) }{d} - \ln \left(\frac{\mu \ln (1/\mu)}{1-\mu^2} \right) + \sqrt{2}\sqrt{\frac{ \ln (d^{d+\alpha}/\delta)  }{d} - \ln \left(\frac{\mu \ln (1/\mu)}{1-\mu^2} \right) -1} \right),
        \end{equation}
        which suffices to ensure that 
        \begin{equation}
            d^{\alpha} \left(\frac{t(1-\mu^2)}{\mu}\right)^d \mu^t \leq \delta.
        \end{equation}
    \end{proof}
\end{lemma}

\begin{lemma}\label{lemma:Imax}
    Let $\Phi:\B{\Hil_A}\to \B{\Hil_B}$ be a quantum channel. Then, 
    \begin{equation}
        I_{\max}(\Phi)= I_{\max}(R:B)_{\Phi_{A\to B}(\psi^+_{RA})} = \inf_{\sigma_B}D_{\max}(\Phi_{A\to B}(\psi^+_{RA})\Vert \psi^+_R \otimes \sigma_B),
    \end{equation}
    where $\psi^+_{RA}=\Omega_{RA}/d_A$ with $\ket{\Omega}_{RA}=\sum_{i} \ket{i}_R \otimes \ket{i}_A$ is a maximally entangled state with $d_R=d_A$. 
\end{lemma}
\begin{proof}
    Recall from Definition~\ref{def:channel-measures} and Remark~\ref{remark:pure} that
    \begin{equation}
        I_{\max}(\Phi) = \sup_{\psi_{RA}} \inf_{\sigma_B} D_{\max}(\Phi_{A\to B}(\psi_{RA})\Vert \psi_R \otimes \sigma_B) \geq \inf_{\sigma_B}D_{\max}(\Phi_{A\to B}(\psi^+_{RA})\Vert \psi^+_R \otimes \sigma_B),
    \end{equation}
    where the supremum is over all pure states $\psi_{RA}$ with $d_R=d_A$ and the inequality follows by choosing $\psi=\psi^+$ in the supremum.
    
    To show the opposite inequality, suppose that $\lambda\geq 0$ and $\sigma_B$ is such that $\Phi_{A\to B}(\psi^+_{RA}) \leq \lambda \psi^+_R \otimes \sigma_B$, which is is equivalent to $\Phi_{A\to B}(\Omega_{RA}) \leq \lambda \iden_R \otimes \sigma_B$. Since any pure state $\psi_{RA}$ can be written as $\psi_{RA}=X_R \Omega_{RA} X_R^{\dagger}$ for some operator $X_R$ satisfying $\Tr (X_R X_R^{\dagger})=1$, the following inequality also holds:
    \begin{equation}
        \Phi_{A\to B}(\psi_{RA}) = X_R\Phi_{A\to B}(\Omega_{RA})X_R^{\dagger} \leq \lambda X_RX_R^{\dagger} \otimes \sigma_B = \lambda \psi_R \otimes \sigma_B.
    \end{equation}
    Hence, for any pure state $\psi_{RA}$, we obtain (see Remark~\ref{remark:Dmax-quasi})
    \begin{equation}
        \inf_{\sigma_B} D_{\max}(\Phi_{A\to B}(\psi_{RA})\Vert \psi_R \otimes \sigma_B) \leq \inf_{\sigma_B}D_{\max}(\Phi_{A\to B}(\psi^+_{RA})\Vert \psi^+_R \otimes \sigma_B).
    \end{equation}
    Taking a supremum over all $\psi_{RA}$ then yields the desired inequality.
\end{proof}

\bibliography{references}
\bibliographystyle{alpha}

\vspace{10pt}
\hrule
\vspace{10pt}
\end{document}